\documentclass{article}

\usepackage{arxiv}

\usepackage[utf8]{inputenc} 
\usepackage[T1]{fontenc}    
\usepackage{hyperref}       
\usepackage{url}            
\usepackage{booktabs}       
\usepackage{amsfonts}       
\usepackage{nicefrac}       
\usepackage{microtype}      
\usepackage{xcolor}         

\usepackage[american]{babel}

\usepackage{amsmath}
\usepackage{amssymb}
\usepackage{amsthm}
\usepackage{tikz}
\usetikzlibrary{arrows,cd,decorations.markings,shapes.geometric,shapes}
\usepackage{mathtools}
\usepackage{ebproof}

\usetikzlibrary{arrows,cd,decorations.markings,shapes.geometric,shapes}

\usepackage{tikzit}
\usepackage[all,cmtip]{xy}

\usetikzlibrary{circuits.ee.IEC}

\usetikzlibrary{arrows.meta,positioning,shapes.geometric,fit}

\pgfdeclarelayer{edgelayer}
\pgfdeclarelayer{nodelayer}
\pgfsetlayers{edgelayer,nodelayer,main}
\tikzstyle{none}=[]
\tikzset{baseline=(current  bounding  box.center)}
\tikzset{every picture/.append style={scale=0.5}}

\usepackage{tikzit}

\usetikzlibrary{circuits.ee.IEC}

\usetikzlibrary{calc,arrows.meta}

\newcommand\C{\mathbf{C}}

\newcommand\Cat{\textup{Cat}}

\newcommand\cofib\rightarrowtail

\newcommand\id{\textup{id}}

\newcommand\Hom{\textup{Hom}}

\newcommand\mdel[1]{}

\usepackage{txfonts}
\usepackage{pxfonts}

\makeatletter
\newcommand{\xdashrightarrow}[2][]{\ext@arrow 0359\rightarrowfill@@{#1}{#2}}
\newcommand*{\doublerightarrow}[2]{\mathrel{
  \settowidth{\@tempdima}{$\scriptstyle#1$}
  \settowidth{\@tempdimb}{$\scriptstyle#2$}
  \ifdim\@tempdimb>\@tempdima \@tempdima=\@tempdimb\fi
  \mathop{\vcenter{
    \offinterlineskip\ialign{\hbox to\dimexpr\@tempdima+1em{##}\cr
    \rightarrowfill\cr\noalign{\kern.5ex}
    \rightarrowfill\cr}}}\limits^{\!#1}_{\!#2}}}
\newcommand*{\triplerightarrow}[1]{\mathrel{
  \settowidth{\@tempdima}{$\scriptstyle#1$}
  \mathop{\vcenter{
    \offinterlineskip\ialign{\hbox to\dimexpr\@tempdima+1em{##}\cr
    \rightarrowfill\cr\noalign{\kern.5ex}
    \rightarrowfill\cr\noalign{\kern.5ex}
    \rightarrowfill\cr}}}\limits^{\!#1}}}
\makeatother

\makeatletter
\newcommand{\twoarrows}[3][0.2ex]{%
  \mathrel{\mathpalette\twoarrows@{{#1}{#2}{#3}}}%
}
\newcommand{\twoarrows@}[2]{\twoarrows@@#1#2}
\newcommand{\twoarrows@@}[4]{%
  \vcenter{\offinterlineskip\m@th
    \ialign{\hfil##\hfil\cr
      $#1#3$\cr
      \noalign{\vskip#2}
      $#1#4$\cr
    }%
  }%
}
\makeatother

\newcommand{\beq}{\begin{equation}}
\newcommand{\eeq}{\end{equation}}

\newtheorem{theorem}{Theorem}
\newtheorem{corollary}{Corollary}
\newtheorem{definition}{Definition}
\newtheorem{lemma}{Lemma}
\newtheorem{example}{Example}

\usepackage[boxruled]{algorithm2e}
\usepackage{algorithmic}

\usepackage{listings}

\usepackage{multirow}
\usepackage{booktabs}
\usepackage{balance}
\usepackage{subfig}
\usepackage{graphicx}

\DeclareFontFamily{U}{dmjhira}{}
\DeclareFontShape{U}{dmjhira}{m}{n}{ <-> dmjhira }{}

\DeclareRobustCommand{\yo}{\text{\usefont{U}{dmjhira}{m}{n}\symbol{"48}}}

\usepackage{natbib} 
\usepackage{mathtools} 
\usepackage{booktabs} 
\usepackage{tikz} 

\usepackage{amsmath,amssymb,mathtools}
\newcommand{\E}{\mathcal{E}}                  
\newcommand{\Dist}{\mathsf{Dist}}             
\newcommand{\Kl}{\mathsf{Kl}}                 

\newcommand{\xto}[1]{\xrightarrow{#1}}
\newcommand{\Obs}{\mathsf{Obs}}    
\newcommand{\Do}{\mathsf{Do}}

\newcommand{\Set}{\mathbf{Set}}

\newcommand{\Sh}{\mathbf{Sh}} 

\usepackage{tikz-cd}
\usetikzlibrary{arrows.meta,positioning,fit}
\usetikzlibrary{shapes,arrows}
\usetikzlibrary{intersections}

\usetikzlibrary{shapes.geometric}  
\usetikzlibrary{fit}               

\usepackage{rotating}

\usepackage{mathtools}
\usepackage{amsthm, amssymb, bm}
\usepackage{tikz}
\usetikzlibrary{arrows.meta, positioning}

\newcommand{\indep}{\mathrel{\perp\!\!\!\perp}}
\newcommand{\nindep}{\not\!\!\indep}
\newcommand{\given}{\,|\,}              

\newcommand{\Topos}{\mathcal{E}}

\newcommand{\J}{\mathcal{J}}
\newcommand{\OmegaE}{\Omega_{\Topos}} 

\newcommand{\Pbb}{\mathbb{P}}

\theoremstyle{plain}
\newtheorem{proposition}[theorem]{Proposition}
\theoremstyle{definition}
\theoremstyle{remark}
\newtheorem{remark}{Remark}

\usepackage{amsmath,amssymb,amsthm,mathtools}

\usepackage{tikz}
\usetikzlibrary{arrows.meta,positioning,calc}
\usepackage{xcolor}

\newcommand{\V}{\mathcal{V}}
\newcommand{\G}{\mathsf{G}}

\usepackage{xcolor}
\usepackage{tikz}
\usetikzlibrary{arrows.meta,positioning,calc,decorations.pathreplacing,shapes.geometric}
\usepackage{pgfkeys}
\usepackage{pgfplots}
\usepackage{amsthm}
\usepackage{stmaryrd}   
\usepackage{booktabs}   
\usepackage{tabularx}   
\usepackage{microtype}  

\newcommand{\ShJ}[1]{\mathbf{Sh}_j(#1)}  
\newcommand{\aj}{a_j}                    
\newcommand{\ForcesJ}{\Vdash_j}          

\newcommand{\U}{U}                       
\newcommand{\Si}{S_i}                    
\newcommand{\Cover}{\mathcal{S}}         

\newcommand{\Vobs}{\mathsf{Vobs}}
\newcommand{\doop}{\mathsf{do}}
\newcommand{\doo}[1]{\mathsf{do}\!\left(#1\right)}
\newcommand{\Pint}{\mathsf{P}}           
\newcommand{\Gbar}[1]{\mathcal{G}_{\overline{#1}}}       
\newcommand{\Gunder}[1]{\mathcal{G}_{\underline{#1}}}    
\newcommand{\Gbarunder}[2]{\mathcal{G}_{\overline{#1},\,\underline{#2}}}

\newcommand{\rlabel}[1]{\textsf{\footnotesize[#1]}}

\providecommand{\CI}{\mathrel{\perp\!\!\!\perp}}

\providecommand{\given}{\,\mid\,} 

\providecommand{\C}{\mathcal{C}}      
\providecommand{\U}{\mathcal{U}}      
\providecommand{\Sh}{\mathbf{Sh}}     
\providecommand{\Ob}{\mathbf{Ob}}     

\providecommand{\J}{j}                 

\providecommand{\Vobs}{V_{\mathrm{obs}}}

\providecommand{\VdoC}{V_{\mathrm{do}(C)}}

\providecommand{\De}

\newcommand{\ZofW}{Z(W)}
\newcommand{\barZofW}{\overline{\ZofW}}

\newcommand{\GbarXZofW}{G_{\overline{X},\,\overline{Z(W)}}}

\title{Intuitionistic \(j\)-Do-Calculus in Topos Causal Models}

\author{ Sridhar Mahadevan \\
	Adobe Research and University of Massachusetts, Amherst\\
	\texttt{smahadev@adobe.com, mahadeva@umass.edu}
}

\hypersetup{
hidelinks,
pdftitle={Intuitionistic j-Do-Calculus in Topos Causal Models},
pdfsubject={Causal inference, sheaf semantics, and intuitionistic logic},
pdfauthor={Sridhar Mahadevan},
pdfkeywords={causal inference, do-calculus, topos theory, sheaves, Kripke--Joyal semantics},
}

\begin{document}
\maketitle

\begin{abstract} 
We develop a sheaf-valued, regime-indexed semantics for Pearl's
do-calculus.  A causal model varies over a site \((\C,\J)\): each stage
carries an ordinary structural causal model, and restriction maps preserve
the variable signature, graph surgery, causal kernels, and the conditionals
used by the calculus.  A causal assertion is \(\J\)-stable at a stage when it
holds on a \(\J\)-cover in the Kripke--Joyal semantics of
\(\Sh_{\J}(\C)\).  Under explicit compatibility and existence assumptions,
Pearl's three rules may be applied chartwise, and the resulting equalities of
interventional kernels descend by the separatedness axiom for sheaves.  This
gives a local-to-global soundness theorem for regime-aware causal reasoning.

Interventions are modeled throughout as structural-mechanism or stochastic
kernel replacement; conditioning is defined only when the required
finite-discrete conditional has positive normalizing mass (or, more
generally, when a compatible disintegration is supplied).  We do not identify
intervention and conditioning with a universal pair of left and right Kan
extensions.  The categorical role of sheafification is instead reflective
localization and descent.  The terminal site recovers ordinary set-based
do-calculus, whereas the trivial Grothendieck topology on a general site
recovers presheaf semantics and need not be Boolean.  The paper is
conceptual: it assumes the regime-indexed causal models and covers are given;
data-driven construction of these objects is addressed in companion work.
\end{abstract}


\newpage 

\tableofcontents

\newpage 

\section{Introduction} \label{sec:intro}

In this paper, we build on the recently proposed Topos Causal Models (TCMs) \citep{sm:neurips_tcm} that formulates causal inference in the category of toposes.  TCMs are part of a recent series of papers on categorical models of causality, including those using symmetric monoidal categories \citep{fong:ms,fritz:jmlr,Cho_2019,string-diagram-surgery}, as well as simplicial sets and higher-order categories \citep{DBLP:journals/entropy/Mahadevan23}.  Any causal model based on graphs  \citep{pearl-book,hedge,spirtes:book}  can be translated into a categorical language. Operations on causal models, such as interventions, can be modeled as functors on the objects of the associated symmetric monoidal category or simplicial set. Categorical approaches to causality also extend to the {\em potential outcomes} counterfactual framework \citep{rubin-book}. 

TCM is also part of a line of previous research of ours
\citep{DBLP:journals/entropy/Mahadevan23,sridhar:higher-algebraic-k-theory} termed {\em universal causality} (UC). This notion derives from the concept defined in category theory \citep{riehl2017category}: a property is universal if it can be defined in terms of an {\em initial} or {\em final} object in a category of diagrams, or in terms of a {\em  representable functor} using the Yoneda Lemma. For example, a structural causal model (SCM) \citep{pearl-book} was originally defined as a unique (deterministic) function mapping a collection of exogenous variables into a collection of endogenous variables,  by ``collating" local functions that serve as independent causal mechanisms \citep{scm-lewis,icm}. In UC, an SCM is further decomposed into diagrams that reveal its universal properties, such as categorical product, coproduct, limits and colimits, equalizers and coequalizers etc. These latter properties can be shown formally to be initial or final objects in a category of diagrams \citep{riehl2017category}, or as representable functors through the Yoneda Lemma \citep{maclane:71}.

A \emph{site} $(\C,\J)$ is a small category $\C$ whose objects index regimes
(e.g., laboratories or experimental contexts) and whose Grothendieck topology
$\J$ specifies which families $\{u_i\!\to u\}$ cover a stage $u$.  The topos
$\Sh_{\J}(\C)$ contains sets varying compatibly over these regimes.  We say
that a formula is \emph{$\J$-stable at $u$} when it is forced at $u$ in the
Kripke--Joyal semantics; equivalently, it can be verified on a covering family
and is compatible under restriction.  The gluing of causal conclusions,
however, requires more than logic alone: the causal kernels, interventions,
graph surgeries, and conditionals must themselves restrict coherently.

Pearl's do-calculus \citep{pearl-book} is a complete axiom system for interventional
identification in acyclic causal models under classical (Boolean) logic.
We give Pearl's rules a regime-indexed semantics inside
$\Sh_{\J}(\C)$.  At each chart we use ordinary SCM intervention semantics:
an intervention replaces a structural mechanism, or equivalently the
corresponding stochastic kernel.  The categorical contribution is a descent
principle: chartwise Pearl equalities that agree on overlaps glue to an
equality at the covered stage.  Thus the topology changes where a premise may
be verified; it does not replace the causal content of Pearl's mutilated-graph
conditions.

We use a number of informal terms like  \emph{stages/contexts} and \emph{regimes/charts}. Formally, in a site $(\mathcal C,J)$ a chart over an ambient object $U$ is a morphism $f:V\!\to\!U$; a $J$-cover of $U$ is a family $\{f_i:V_i\!\to\!U\}$ whose generated sieve lies in $J(U)$.
When we say that a statement holds “on a $J$-cover of $U$,” we mean it holds
chartwise for some covering sieve of $U$.  All soundness statements below
explicitly assume that restriction preserves the causal operations appearing
in the statement.

\paragraph{Contributions.}
(i) We define $\J$-local d-separation and interventional equality over a
regime site.  (ii) We state the compatibility assumptions under which
chartwise SCM semantics forms a sheaf of causal models.  (iii) We prove a
local-to-global soundness theorem and obtain $\J$-local forms of Pearl's three
rules without changing their mutilated-graph premises.  (iv) We distinguish
two reductions: the trivial topology yields stagewise presheaf semantics,
while the terminal site yields ordinary set-based do-calculus.  We make no
completeness claim for the resulting regime-indexed calculus.

\paragraph{Scope and companion work.}
This paper develops the theory of \(j\)-stable causal inference and a \(j\)-do-calculus inside \(\Sh_J(\mathcal C)\).
Our focus is conceptual: we assume access to the theoretical objects (e.g., stages \(U\), \(J\)-covers of \(U\), and the
internal interventional distribution \(\mathsf P^{\!\mathrm{int}}\)) and study their logical consequences.
A companion paper provides the algorithmic side: how to estimate the required entities from data and how to
instantiate \(j\)-do with standard discovery procedures (e.g., score-based and constraint-based methods), building on
recent surveys \citep{zanga2023surveycausaldiscoverytheory}. There we show how to (i) form data-driven \(J\)-covers
(via regime/section constructions), (ii) compute chartwise CIs after graph surgeries, and (iii) glue them to certify
the premises of the \(j\)-do rules in practice. 

\begin{table}[]
    \centering
    \begin{tabular}{l l}
\toprule
Classical (Pearl) & In $\Topos=\Sh_{\J}(\C)$ (internal)\\
\midrule
$X \indep Y \given Z$ & $\Topos \vDash (X \indep Y \given Z)$ (local truth)\\
$\Pbb(Y \given \Do(X),Z)$ & Internal conditional $\Pbb_{\Topos}(Y \mid \Do(X),Z)$\\
Back-door admissibility & $J$-stable screening in internal logic\\
Rule 1 (Insert/Delete obs.) & J1 (Thm.~\ref{thm:J1})\\
Rule 2 (Action/Obs. exch.) & J2 (Thm.~\ref{thm:J2})\\
Rule 3 (Insert/Delete action) & J3 (Thm.~\ref{thm:J3})\\
\bottomrule
\end{tabular} \vskip 0.2in
\begin{tabular}{l l}
\toprule
Classical (Pearl) & In $\Topos=\Sh_{\J}(\C)$ (internal) \\
\midrule
$X \indep Y \given Z$ &
$\Topos \vDash (X \indep Y \given Z)$ (local truth on a $\J$-cover) \\
$\Pbb(Y \mid \Do(X),Z)$ &
Internal conditional in the topos $\Topos$ \\
Back-door admissibility &
$Z$ screens-off $X\to Y$ \emph{locally}, i.e.\ $\J$-stable separation \\
Rule 1 (insert/delete obs.) & J1 above \\
Rule 2 (action/obs. exchange) & J2 above \\
Rule 3 (insert/delete action) & J3 above \\
Identifiable effect &
Derivable via J1–J3 in $\Topos$ (hence stable on covers) \\
\bottomrule
\end{tabular}
\caption{A sheaf-valued semantics for Pearl's rules.  The $\J$-local
premise is checked on a cover; causal compatibility and sheaf descent then
yield the interventional equality at the covered stage.}
\label{tab:j-do-calc}
\end{table}

\begin{center}

\end{center}

\section{From Classical Do-Calculus to $j$-Do-Calculus}
\label{scm-review} 

This section previews the passage from ordinary do-calculus to causal models
indexed by a site.  The Grothendieck topology $J$ and its corresponding
Lawvere--Tierney operator $j$ are equivalent presentations of the same
sheaf condition on a presheaf topos; the latter is not a further
generalization of the former.  We work primarily with the external site
notation and write \(U\Vdash_J\varphi\) for forcing in
\(\Sh_J(\C)\).  Intervention retains its ordinary causal meaning as mechanism
or kernel replacement.  The new ingredient is that the premises of Pearl's
rules may be verified on a cover and their compatible conclusions then
descend.  Table~\ref{tab:glossary} summarizes the notation.

\subsection{Classical Do-Calculus} 

 We briefly review the notion of a structural causal model (SCM) \citep{pearl-book}, and the classical notion of do-calculus. Succinctly, any SCM $M$ defines a unique function from exogenous variables to endogenous variables, and do-calculus models interventions as ``sub-functions": 

\begin{definition}\citep{pearl-book}
\label{scm}
    A {\bf structural causal model} (SCM) is defined as the triple $\langle U, V, F \rangle$ where $V = \{V_1, \ldots, V_n \}$ is  a set of {\em endogenous} variables, $U$ is a set of {\em exogenous} variables, $F$ is a set $\{f_1, \ldots, f_n \}$ of ``local functions" $f_i: U \cup (V \setminus V_i) \rightarrow V_i$ whose composition induces  a unique function $F$ from  $U$ to $V$. 
\end{definition}
\begin{definition}\citep{pearl-book}
\label{intervention-scm}
   Let $M = \langle U, V, F \rangle $ be a causal model defined as an SCM, and $X$ be a subset of variables in $V$, and $x$ be a particular realization of $X$.  A {\bf submodel} $M_x = \langle U, V, F_x \rangle $ of $M$ is the causal model $ M_x =  \langle U, V, F_x \rangle$, where $F_x = \{f_i : V_i \notin X \} \cup \{X = x \}$. 
\end{definition}
\begin{definition}\cite{pearl-book}
    \label{do-action}
    Let $M$ be an SCM, $X$ be a set of variables in $V$, and $x$ be a particular realization of $X$. The {\bf effect} of an action $\mbox{do}(X=x)$ on $M$ is given by the submodel $M_x$. 
\end{definition}
\begin{definition}\citep{pearl-book}
    \label{potential-outcome}
    Let $Y$ be a variable in $V$, and let $X$ be a subset of $V$. The {\bf potential outcome} of $Y$ in response to an action $\mbox{do}(X=x)$, denoted $Y_x(u)$, is the solution of $Y$ for the set of equations $F_x$. 
\end{definition}
\citet{scm-lewis} propose an axiomatic theory of counterfactuals based on the above definitions, where the key definition of a counterfactual is given as:
\begin{definition}
    Let $Y$ be a variable in $V$ and let $X$ be a subset of $V$. The counterfactual sentence ``The value that $Y$ would have obtained had $X$ been set to $x$" is defined as the potential outcome $Y_x(u)$.
\end{definition}

Pearl's \emph{do-calculus} (Pearl, 2009) provides three algebraic rules
for manipulating interventional expressions of the form $P(Y\mid do(Z),X,W)$
based on conditional independence statements in a causal graph $G$. The notation used
is as follows: $G_{\bar X}$ means delete all arrows into $X$ (surgical intervention on $X$);  $G_{\underline Z}$: delete all arrows out of $Z$; $Z(W)$ denotes the subset of $Z$ that are not ancestors of any node in $W$ in $G_{\bar X}$; and finally, $G_{\bar X,\barZofW}$ denotes the intervention that deletes arrows into those $Z$-nodes that are not ancestors of $W$. 
\begin{enumerate}
  \item \textbf{Rule 1 (Insertion/Deletion of Observations).}
  If $(Y \perp Z \mid X,W)_{G_{\bar X}}$, then
  \[
    P(Y \mid \mathrm{do}(X), Z, W) \;=\; P(Y \mid \mathrm{do}(X), W).
  \]

  \item \textbf{Rule 2 (Action/Observation Exchange).}
  If $(Y \perp Z \mid X,W)_{G_{\bar X,\underline Z}}$, then
  \[
    P \ \!\bigl(Y \mid \mathrm{do}(X), \mathrm{do}(Z), W\bigr)
    \;=\;
    P \ \!\bigl(Y \mid \mathrm{do}(X), Z, W\bigr).
  \]

  \item \textbf{Rule 3 (Insertion/Deletion of Actions).}
  If $(Y \perp Z \mid X,W)_{G_{\bar{X},\barZofW}}$, then
  \[
    P \ \!\bigl(Y \mid \mathrm{do}(X), \mathrm{do}(Z), W\bigr)
    \;=\;
    P \ \!\bigl(Y \mid \mathrm{do}(X), W\bigr).
  \]
\end{enumerate}

These rules form a sound and complete system for deriving
identities between observational and interventional distributions
using only the graphical structure of~$G$. 

\subsection{$j$-Do-Calculus: A Birds-Eye View}

In moving from classical do-calculus to $j$-do-calculus, we transition from causal models over graphs to general categories, specifically toposes. The simplest way to understand this transition is to note that a category ${\cal C}$ whose objects are functions $f: A \rightarrow B$ over sets, and whose arrows are commutative diagrams between functions $f$ and $g$, defined as ${\cal C}(f, g)$ defines a topos \citep{goldblatt:topos}. This result, which was shown in detail for the case of SCMs in \citep{sm:neurips_tcm}, shows that causal inference in SCMs and graphs is intrinsically topos-theoretic. One can expand this simple result to cover more cases. For example, the category of graphs ${\cal G}$ can be defined to consist of two objects $v$  and $e$, and two non-identity arrows from $v$ to $e$. Each graph then is defined as a presheaf ${\bf Sets}^{{\cal G}^{op}}$, a functor that maps the objects $v$ and $e$ to the set of edges $E$ and vertices $V$ of the actual graph, and that maps the two non-identity arrows between $v$ and $e$ to the initial and terminal vertex of each edge. More generally, any (small) category ${\cal C}$ can be converted into a topos through the Yoneda embedding ${\cal C} \rightarrow {\bf Sets}^{{{\cal C}^{op}}}$, defined as $c \mapsto {\cal C}(-, c)$, and called the presheaf. The category of presheafs forms a topos \citep{maclane:sheaves}. We will also develop a new set of rules of  $j$-do-calculus, described  in Figure~\ref{fig:j-do-calc}, which will be explained at length in the remainder of the paper. Table~\ref{tab:glossary} provides a convenient glossary of symbols that can be handy in reading the remainder of the paper. 

\begin{table}[p]
\centering
\small
\begin{tabularx}{\linewidth}{@{} l l X @{}}
\toprule
\textbf{Symbol} & \textbf{Type} & \textbf{Meaning / Typical usage} \\
\midrule
$\Cat$ & category & Site of “regimes/contexts” (objects are stages; arrows are refinements). \\
$j$ & L--T topology & Lawvere–Tierney topology on $\mathbf{Sets}^{\Cat^{\mathrm{op}}}$; enforces which sieves are “covering.” \\
$\ShJ{\Cat}$ & topos & Sheaves on $\Cat$ for $j$ (the $j$-reflective subtopos). \\
$\aj$ & functor & $j$-sheafification (left exact reflector $\mathbf{Sets}^{\Cat^{\mathrm{op}}}\!\to\!\ShJ{\Cat}$). \\
$\U$ & object & A stage (or (object) $U\in\C$ ) in $\Cat$. \\
$\Cover=\{\Si\!\hookrightarrow\!\U\}$ & family & A $j$-cover of $\U$ (local charts that jointly “see” $\U$). \\
$\ForcesJ$ & relation & Internal forcing in $\ShJ{\Cat}$; $\U \ForcesJ \varphi$ reads “$\varphi$ holds $j$-stably at $\U$.” \\
$X \ \CI \ Y \mid Z$ & formula & Conditional independence assertion (CI). \\
$\doo{x}$ & term & Pearl’s do-operator (surgical intervention) internalized in $\ShJ{\Cat}$. \\
$\Gbar{X}$ & graph & Mutilated graph with incoming edges to $X$ cut (intervening on $X$). \\
$\Gunder{Z}$ & graph & Graph with outgoing edges from $Z$ cut (treating $Z$ as “measurement”). \\
$\Pint(\cdot)$ & object & Internal probability in $\ShJ{\Cat}$; e.g., $\Pint(y\mid \doo{x}, z, w)$. \\
\bottomrule
\end{tabularx}
\caption{\textbf{Glossary of symbols and notation.} Informal reading: $j$ specifies which families of local charts count as covers; $\U \ForcesJ \varphi$ means \emph{every} chart in a $j$-cover of $\U$ validates $\varphi$, hence $\varphi$ is forced globally at $\U$.}
\label{tab:glossary}
\end{table}

\begin{figure}[p]
    \centering
\begin{center}
\fbox{%
\begin{minipage}{0.97\linewidth}
\textbf{\large \(j\)-do rules at a glance}

\medskip
All equalities are identities \emph{internal} to $\ShJ{\Cat}$ and read at stage $\U$ (i.e.\ under $\U \Vdash_j \cdots$).
Each premise means: there exists a \(j\)-cover $\mathcal S=\{S_i\!\to \U\}_i$ such that the stated CI holds on \emph{every} chart $S_i$ after the indicated graph surgery.

\medskip
\textbf{\rlabel{$j$-Rule 1: insert/delete observations}}
\[
\Bigl( Y \ \perp\  Z \ \bigm|\ X,W \ \text{in}\ \Gbar{X}\ \text{on a $j$-cover of $\U$} \Bigr)
\;\Longrightarrow\;
\Pint(y \mid \mathrm{do}(x), z, w) \;=\; \Pint(y \mid \mathrm{do}(x), w).
\]
\emph{Reading:} After cutting arrows into $X$, if every chart blocks $Z$ from $Y$ given $X,W$, then observing $Z$ is irrelevant under $\mathrm{do}(x)$.

\medskip
\textbf{\rlabel{$j$-Rule 2: action/observation exchange}}
\[
\Bigl( Y \ \perp\  Z \ \bigm|\ X,W \ \text{in}\ \Gbarunder{X}{Z}\ \text{on a $j$-cover of $\U$} \Bigr)
\;\Longrightarrow\;
\Pint(y \mid \mathrm{do}(x), \mathrm{do}(z), w) \;=\; \Pint(y \mid \mathrm{do}(x), z, w).
\]
\emph{Reading:} After cutting arrows into $X$ and \emph{out of} $Z$, intervening on $Z$ equals observing $Z$ under $\mathrm{do}(x)$, chartwise.

\medskip
\textbf{\rlabel{$j$-Rule 3: insert/delete actions}}
\[
\Bigl( Y \ \perp\  Z \ \bigm|\ X,W \ \text{in}\ \GbarXZofW\ \text{on a $j$-cover of $\U$} \Bigr)
\;\Longrightarrow\;
\Pint(y \mid \mathrm{do}(x), \mathrm{do}(z), w) \;=\; \Pint(y \mid \mathrm{do}(x), w).
\]
\emph{Reading:} After cutting arrows into $X$ and into the parents of $Z$ not in $W$ (i.e.\ $\overline{Z(W)}$), if every chart blocks $Z$ from $Y$ given $X,W$, then $\mathrm{do}(z)$ is irrelevant under $\mathrm{do}(x)$.

\medskip
\textbf{Conservativity.}
For the trivial topology the local premises reduce to stagewise
mutilated-graph \(d\)-separation.  Ordinary set-based Pearl semantics is
recovered on the terminal site \(\C=\mathbf 1\), for which
\(\Sh(\mathbf 1)\simeq\mathbf{Set}\).  A presheaf topos on a general site need
not be Boolean.

\textbf{Soundness (sketch).}
The CI premises hold locally on a \(j\)-cover; by locality and sheaf gluing, the equalities hold internally in $\ShJ(\Cat)$, hence at stage $\U$.
\end{minipage}}
\end{center}
    \caption{The Rules of $j$-do-calculus.}
    \label{fig:j-do-calc}
\end{figure}

\section{Causal Models Over a Topos of Sheaves}
\label{gdc-stoch}
The categorical framework underlying Topos Causal Models (TCMs) introduced
in \citep{sm:neurips_tcm} places causal data in a presheaf topos
\citep{maclane:sheaves}.  We impose a Grothendieck topology to select the
families on which compatible local data must glue.  On a presheaf topos this
external topology corresponds exactly to a Lawvere--Tierney topology on the
subobject classifier; we use the two equivalent presentations where each is
most convenient.
\subsection{Grothendieck Topology on Sites}
\begin{definition}
A {\bf sieve} for any object $x$ in any (small) category ${\cal C}$ is a subobject of its Yoneda embedding $\yo(x) = {\cal C}(-,x)$. If $S$ is a sieve on $x$, and $h: y \rightarrow x$ is any arrow in category ${\cal C}$, then 
    \[ h^*(S) = \{g \ | \ \mbox{cod}(g) = y, hg \in S \}\]
\end{definition}
\begin{definition}\citep{maclane1992sheaves}
A {\bf Grothendieck topology} on a category ${\cal C}$ is a function $J$ which assigns to each object $x$  of ${\cal C}$ a collection $J(x)$ of sieves on $x$ such that
\begin{enumerate}
    \item the maximum sieve $t_x = \{ f | \mbox{cod}(f) = x \}$ is in $J(x) $. 
    \item If $S \in J(x)$  then $h^*(S) \in J(y)$ for any arrow $h: y \rightarrow x$. 
    \item If $S \in J(x)$ and $R$ is any sieve on $x$, such that $h^*(R) \in J(y)$ for all $h: y \rightarrow x$, then $R \in J(C)$. 
\end{enumerate}
A $J$-cover is a covering family whose \emph{generated sieve} lies in $J(U)$.
\end{definition}
\begin{lemma}[Families vs.\ sieves]
A family $\{f_i\colon V_i\to U\}$ is $J$-covering iff its generated sieve
$\langle f_i\rangle$ lies in $J(U)$. Moreover, if $\{f_i\}$ refines $\{g_j\}$
(meaning each $f_i$ factors through some $g_j$), then
$\langle f_i\rangle \subseteq \langle g_j\rangle$.
\end{lemma}
We can now define categories with a given Grothendieck topology as {\em sites}. 
\begin{definition}
    A {\bf site} is defined as a pair $({\cal C}, J)$ consisting of a small category ${\cal C}$ and a Grothendieck topology $J$ on ${\cal C}$. 
\end{definition}
\begin{definition}
\label{subobj-defn}
    The {\bf subobject classifier} $\Omega$ is defined on any topos ${\bf Sets}^{C^{op}}$ as subobjects of the representable functors: 
    \[ \Omega(x) = \{S | S \ \ \mbox{is a subobject of} \ \ {\cal C}(-, x) \} \]
    and the morphism ${\bf true}: 1 \rightarrow \Omega$ is ${\bf true}(x) = x$ for any representable $x$. 
\end{definition}

\subsection{Lawvere-Tierney Topologies on a Topos}
\label{lawvere}

Grothendieck topologies on a small category and Lawvere--Tierney topologies
on its presheaf topos are two equivalent presentations of sheaf semantics
\citep{maclane:sheaves}.  We use the external Grothendieck presentation for
covers and the internal Lawvere--Tierney presentation when discussing the
associated modality.

\begin{definition}[Lawvere--Tierney causal topology]
Let $\E$ be an elementary topos with subobject classifier $\Omega$.
A Lawvere--Tierney topology is an arrow $j:\Omega\to\Omega$ satisfying
\[
j(\top)=\top,\qquad
j(p\wedge q)=j(p)\wedge j(q),\qquad
j(j(p))=j(p),
\]
where $\top={\bf true}$.  A subobject with characteristic map $\chi_P$ is
$j$-closed when $j\circ\chi_P=\chi_P$.  Compatibility of a causal model or
probability monad with the corresponding sheaf subtopos is additional
structure; it is not part of the definition of a Lawvere--Tierney topology.
\end{definition}

The precise relation is the following correspondence.

\begin{theorem}\citep{maclane:sheaves}
     If ${\cal C}$ is a small category, the Grothendieck topologies J on C correspond exactly to Lawvere- Tierney topologies on the presheaf topos ${\bf Sets}^{{\cal C}^{op}}$.
\end{theorem}

Figure~\ref{tierneytop} gives a diagrammatic illustration of the relationship between the two approaches. 


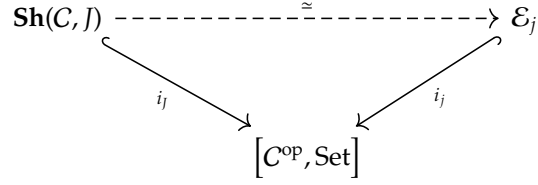
\begin{figure}[t]
  \centering
  \begin{tikzcd}[column sep=huge, row sep=large]
    \mathbf{Sh}(\mathcal{C},J)
      \arrow[dr, hook, "i_J"']
      \arrow[rr, dashed, "\simeq"]
      &&
    \mathcal{E}_j
      \arrow[dl, hook, "i_j"] \\[4pt]
      &
      {\big[\mathcal{C}^{\mathrm{op}},\mathrm{Set}\big]}
  \end{tikzcd}
  \caption{External Grothendieck topology \(J\) and internal Lawvere–Tierney topology \(j\) both induce subtopoi embedded in the presheaf topos \([\mathcal{C}^{\mathrm{op}},\mathrm{Set}]\).}
  \label{tierneytop}
\end{figure}

\subsection{Kripke-Joyal Semantics for Sheaves}

Every topos has an internal intuitionistic logic that derives from the fact that the subobject classifier $\Omega$ yields a poset of subobjects on which the semantics of a formal Mitchell-B\'enabou language describing objects and arrows in the category can be defined. This formal language is associated with a Kripke-Joyal semantics, which we will specialize  to a topos equipped with a Grothendieck topology, that is a site. This specialized structure captures how causal inference is woven in the fabric of the internal logic of a causal topos. Define ${\tt Sh}({\cal C, J})$ be a topos of sheaves with a specified Grothendieck topology ${\cal J}$, defined by the following diagram, where $\yo$ is the Yoneda embedding, and ${\cal P}$ is a presheaf: 
\[ {\cal C} \xrightarrow[]{\yo} {\cal P(C)} \xrightarrow[]{a} \Sh({\cal C,J}) \cong {\cal C}\]
where we know that the Yoneda embedding $\yo$ creates a full and faithful copy of the original category ${\cal C}$. Let us define the semantics for a sheaf element $\alpha \in X(C)$, where $ X(C) = \Sh({\cal C}, J)({\cal C}(-, C), X))$.  We will describe the Kripke-Joyal semantics in more detail later in the paper, but for now, a concise summary for the topos category of sheaves is as follows: 

\begin{enumerate}
    \item $C \Vdash \phi(\alpha) \wedge \psi(\alpha)$ if it holds that $C \Vdash \phi(\alpha)$ and $C \Vdash \psi(\alpha)$. 
    \item $C \Vdash \phi(\alpha) \vee \psi(\alpha)$ if there is a covering $\{ f_i: C_i \rightarrow C \}$ such that for each $i$, either $C_i \Vdash \phi(\alpha)$ or $C_i \Vdash \psi(\alpha)$. 
    \item $C \Vdash \phi(\alpha) \rightarrow \psi(\alpha)$ if for all $f: D \rightarrow C$, and $D \Vdash \phi(\alpha \circ f)$, it holds that $D \Vdash \psi(\alpha \circ f)$. 
    \item $C \Vdash \neg \phi(\alpha)$ holds if for all arrows $f: D \rightarrow C$ in ${\cal C}$, if $D \Vdash \phi(\alpha \circ f)$ holds, then the empty family is a cover of $D$. 
    \item $C \Vdash \exists y \ \phi(x, y)$ holds if there is a covering $\{ f_i: C_i \rightarrow C \}$ and elements $\beta_i \in Y(C_i)$ such that $C_i \Vdash \phi(\alpha \circ f_i, \beta_i)$ holds for each $i$. 
    \item Finally, for universal quantification, $C \Vdash \forall y \ \phi(x, y)$ holds if for all arrows $f: D \rightarrow C$ in the category ${\cal C}$, and all $\beta \in Y(D)$, it holds that $D \Vdash \phi(\alpha \circ f, \beta)$. 
\end{enumerate}

\subsection{$j$-Do-Calculus on Sites}

To transition from classical do-calculus to $j$-do-calculus, we need to provide a ``bridge" that maps from classical notions, like d-separation, to intuitionistic notions in $j$-do-calculus. We begin this transition by introducing some terms that will be used in the remainder of the paper. 

\paragraph{Stages and generalized elements.}
Let $(\C,J)$ be a site and $\Sh_J(\C)$ its sheaf topos. For any object $A\in\Ob(\C)$,
a \emph{generalized element of $A$ at stage $V$} is a morphism
$f:V\to A$ (equivalently, an element of the presheaf $yA(V)=\Hom_{\C}(V,A)$).
The special case $\mathbf{1}\to A$ (where $\mathbf{1}$ is terminal) is a
\emph{global element}. In what follows we fix an \emph{ambient context} (or
\emph{ambient object}) $U\in\Ob(\C)$ and call any arrow $f:V\to U$ a \emph{local
stage over $U$}.

\paragraph{Charts (``regimes'') and $J$-covers.}
A \emph{chart} (our earlier ``regime'') is precisely a local stage
$f:V\to U$. A family of charts $\{f_i:V_i\to U\}_{i\in I}$ generates the sieve
\[
\langle f_i\rangle \;=\; \{\, h:W\to U \mid \exists i,\ \exists g:W\to V_i\ \text{with}\ h=f_i\circ g \,\}.
\]
We call $\{f_i\}$ a \emph{$J$-cover of $U$} iff $\langle f_i\rangle\in J(U)$
(i.e.\ the generated sieve is $J$-covering).

\paragraph{Reading formulas ``at stage $U$''.}
Let $\varphi$ be a formula in the internal language. Write
$U \Vdash_J \varphi$ to mean that $\varphi$ is (internally) true at the
ambient object $U$ in $\Sh_J(\C)$. In Kripke–Joyal semantics this is equivalent to
the existence of a $J$-covering sieve $S\subseteq\Hom_{\C}(-,U)$ such that each
local stage $f:V\to U$ in $S$ \emph{forces} $\varphi$ after pullback:
\[
U \Vdash_J \varphi
\quad\Longleftrightarrow\quad
\exists\,S\in J(U)\ \text{with}\ \forall f:V\to U\ \text{in }S,\quad
V \Vdash_J \varphi|_f.
\]
Informally: \emph{$\varphi$ holds chartwise on a $J$-cover of $U$}.

\paragraph{Grothendieck topology and $J$-covers.}
A \emph{sieve} $S$ on $U$ is \emph{$J$-covering} iff $S\in J(U)$.
We will say that a family of charts $\{f_i\colon V_i\to U\}$ is a \emph{$J$-cover of $U$}
iff the sieve it generates is $J$-covering:
\[
\{f_i\}\ \text{is a $J$-cover of $U$}\quad\Longleftrightarrow\quad \langle f_i\rangle\in J(U).
\]
Thus our earlier “$J$-cover” phrase always refers to a \emph{covering family whose generated sieve is $J$-covering}.

\paragraph{Lawvere–Tierney topology $j$ and $J$.}
The Grothendieck topology $J$ on $\C$ corresponds to a Lawvere–Tierney topology
$j\colon\Omega\to\Omega$ on the presheaf topos $\widehat{\C}$; the sheaf topos
$\Sh_J(\C)$ is the $j$-sheaf subtopos of $\widehat{\C}$. We freely pass between
$J$ (external/topological) and $j$ (internal/logical) viewpoints; “$j$-closure”
of a subobject corresponds to saturation under $J$-covering sieves.

\noindent\textbf{Slogan.}
A conditional independence \(\varphi \equiv (X \ \CI \ Y \mid Z)\) is \emph{\(j\)-stable at a stage \(\U\)} iff the sieve of all refinements \(u: V \to \U\) that validate \(\varphi\) is a \(J\)-cover of \(\U\).

\paragraph{Site of causal contexts.}
Fix a finite variable set \(\mathcal V\) and a DAG \(G\) on \(\mathcal V\).
A \emph{stage} is a pair \(\U = (G,\sigma)\), where \(\sigma\) is a status profile that records which nodes are (i) conditioned/observed, (ii) intervened upon (incoming arrows cut), etc.
A \emph{morphism} \(u : (G',\sigma') \to (G,\sigma)\) is a refinement that is identity on node names and \emph{monotone in status} (a refinement may condition or intervene on more variables, but never less).
Stages and refinements form a category \(\Cat\).

\paragraph{Open paths and satisfaction.}
For disjoint \(X,Y,Z \subseteq \mathcal V\) and a stage \(\U=(G,\sigma)\), let \(\textsf{OpenPaths}_\U(X,Y \mid Z)\) be the set of \(G\)-paths from \(X\) to \(Y\) that are \(d\)-\emph{open} under the usual collider/non-collider rules, evaluated after applying the surgeries in \(\sigma\) (e.g., \(\mathrm{do}(\cdot)\)).
Write
\[
\U \models (X \ \CI \ Y \mid Z)
\quad\Longleftrightarrow\quad
\textsf{OpenPaths}_\U(X,Y \mid Z)=\varnothing.
\]

\paragraph{The sieve selected by a CI formula.}
Given \(\varphi\equiv(X \ \CI \ Y \mid Z)\) and \(\U\), define
\[
\mathsf S_\varphi(\U)
\;:=\;
\{\, u: V \to \U \text{ in } \Cat \;\mid\; V \models \varphi \,\}.
\]

\noindent\textbf{Lemma (sieve).}
\(\mathsf S_\varphi(\U)\) is a sieve on \(\U\) (i.e., closed under precomposition).

\emph{Proof sketch.}
If \(u:V\to \U\) validates \(\varphi\) and \(w:W\to V\) is any arrow, then \(W\) refines \(V\) monotonically in status, which can only block additional paths; hence \(W\models\varphi\) and \(u\!\circ\! w\in \mathsf S_\varphi(\U)\).
\(\square\)

\paragraph{Grothendieck topologies from admissible charts.}
Fix for each \(\U\) a family \(\{\rho_k: V_k \to \U\}_{k\in K}\) of \emph{admissible local views} (charts) used to test CI at \(\U\) (e.g., purely observational; or a mix including certain \(\mathrm{do}(\cdot)\)-surgeries).
Let \(J\) be the Grothendieck topology \emph{generated} by these bases: a sieve \(S\) covers \(\U\) iff it contains a jointly epimorphic family refining \(\{\rho_k\}\).
Two canonical choices:
\begin{itemize}\itemsep3pt
  \item \(J_{\mathrm{id}}\) (classical): basis \(=\{\mathrm{id}_\U\}\).
  \item \(J_{\mathrm{mix}}\): basis includes observational charts and specific interventional charts.
\end{itemize}

\paragraph{Forcing semantics ( \(j\)-stability ).}
Write
\[
\U \Vdash_J (X \ \CI \ Y \mid Z)
\quad:\Longleftrightarrow\quad
\mathsf S_\varphi(\U) \text{ is a \(J\)-cover of } \U.
\]

\noindent\textbf{Proposition (conservativity).}
With \(J_{\mathrm{id}}\),
\[
\U \Vdash_{J_{\mathrm{id}}} (X \CI Y \mid Z)
\quad\Longleftrightarrow\quad
\U \models (X \CI Y \mid Z).
\]
\emph{Reason.}
A sieve covers \(\U\) in \(J_{\mathrm{id}}\) iff it contains \(\mathrm{id}_\U\). Thus \(\mathsf S_\varphi(\U)\) covers iff \(\mathrm{id}_\U \in \mathsf S_\varphi(\U)\), i.e., \(\U\models\varphi\).
\(\square\)

\noindent\textbf{Proposition (soundness of \(j\)-stability).}
If \(\{\rho_k:V_k\to \U\}\) generates \(J\) at \(\U\) and \(V_k \models (X \ \CI \ Y \mid Z)\) for all \(k\), then \(\U \Vdash_J (X \ \CI \ Y \mid Z)\).

\emph{Reason.}
Each generator \(\rho_k\) lies in \(\mathsf S_\varphi(\U)\); hence the sieve they generate covers, and by upward closure of covering sieves, so does \(\mathsf S_\varphi(\U)\).
\(\square\)

\paragraph{Worked mapping: earthquake example.}
Let \(\U=(G,\sigma)\) with \(B\to A \leftarrow E\) and \(A\to C\).
Take \(J_{\mathrm{mix}}\) generated by two charts: an \emph{observational} chart \(\rho_{\mathrm{obs}}\) (no conditioning on colliders unless stated) and an \emph{interventional} chart \(\rho_{\mathrm{do}A}\) that cuts the incoming edges into \(A\).
Then:
\begin{align*}
&\text{(i) } \U \Vdash_{J_{\mathrm{mix}}} (B \ \CI \ 
 E)
&&\text{(collider closed in obs; parents cut under } \mathrm{do}(A)).\\
&\text{(ii) } \U \Vdash_{J_{\mathrm{mix}}} (B \ \CI \ C \mid A)
&&\text{(chain blocked by }A\text{ in both charts).}\\
&\text{(iii) } \U \not\Vdash_{J_{\mathrm{mix}}} (B \ \CI \ E \mid A)
&&\text{(conditioning on the collider opens the path in the obs chart).}
\end{align*}

\paragraph{Takeaway.}
A CI formula \(\varphi\) determines a sieve \(\mathsf S_\varphi\); a Grothendieck topology \(J\) encodes which local views count as \emph{covers}. Classical CI is truth at \(\U\); \(j\)-stability is truth on a \(J\)-cover of \(\U\)---i.e., gluable from admissible local regimes.

\paragraph{CI as an internal predicate.}
Fix a graph object $G$ (DAG with surgery) represented in $\C$.
For variables $X,Y,Z$ (as objects/indices in $G$), let
$\CI_G(X;Y\mid Z)$ denote the internal formula “$X\perp Y\mid Z$ in $G$”.
Our usage
\[
\text{``$Y\ \CI\ Z\mid X,W\ \text{in}\ \bar{G}^{(\cdot)}$ on a $J$-cover of $U$''}
\]
means precisely:
there exists a $J$-covering sieve $S\subseteq\Hom(-,U)$ such that for every
$f\colon V\to U$ in $S$, the (pulled-back, surgically modified) graph
satisfies $V \Vdash_J \CI_G(X;Y\mid Z)$.
By the clause above, this suffices to conclude $U\Vdash_J\CI_G(X;Y\mid Z)$.


\section{Illustrating $j$-stability with Simple Causal DAG models}

Let us begin to build intuition about $j$-stability using some simple examples first. The details will be explained later in the paper, but we want to convey the ideas at a high level first. The goal is to begin to concretize the above abstractions, and the ones to follow. The reader is alerted to the fact that not all terms used here have been properly defined yet, but before getting into precise definitions, the examples should help set the stage for the more precise terminology to follow. 

\subsection{Earthquake DAG} 

\begin{table}[t]
\centering
\small
\caption{$J$-stable CI facts on the classic Earthquake–Burglary DAG ($B\!\to\!A\!\leftarrow\!E$, $A\!\to\!C$). We use two charts: $S_{\text{obs}}$ (observational) and $S_{\text{do}A}$ with $\mathrm{do}(A)$ (incoming edges into $A$ cut).}
\begin{tabular}{l l p{0.52\linewidth} l}
\toprule
\textbf{Claim} & \textbf{Charts used (cover)} & \textbf{Blocking rationale (per chart)} & \textbf{Verdict} \\
\midrule
$B\!\perp\!E$
& $\{S_{\text{obs}},\,S_{\text{do}A}\}$
& In $S_{\text{obs}}$, collider $A$ blocks $B\!\leadsto\!E$; in $S_{\text{do}A}$, incoming edges to $A$ are cut, so $B$ and $E$ remain separated.
& $J$-stable \\
\addlinespace[2pt]
$B\!\perp\!C\mid A$
& $\{S_{\text{obs}},\,S_{\text{do}A}\}$
& In $S_{\text{obs}}$, the chain $B\!\to\!A\!\to\!C$ is blocked by conditioning on the mediator $A$; in $S_{\text{do}A}$, $C$ depends only on $A$ (parents of $A$ cut), so $B$ adds no info given $A$.
& $J$-stable \\
\addlinespace[2pt]
$B\!\perp\!E\mid A$
& Any cover containing $S_{\text{obs}}$
& Conditioning on the collider $A$ opens $B\!\to\!A\!\leftarrow\!E$ in $S_{\text{obs}}$; thus the CI fails on that chart.
& \textbf{Not} $J$-stable \\
\bottomrule
\end{tabular}
\label{tab:quake-jstable}
\end{table}

We begin with well-known Earthquake example from \citep{pearl:bnets-book}. Recall that in this case, the causal DAG had the following variables: $B$ = burglary, $E$ = earthquake, $A$ = alarm, $C$ = neighbor calls.  The DAG is then described by the following structure: 
\[
B \to A \leftarrow E,\qquad A \to C.
\]

\paragraph{Classical d-separation facts.}
\begin{enumerate}
  \item $B \perp\!\!\!\perp E$ \ (collider at $A$ is \emph{unconditioned}, hence blocks).
  \item $B \not\!\perp\!\!\!\perp C$ but $B \perp\!\!\!\perp C \mid A$ \ (the chain $B\!\to\!A\!\to\!C$ is blocked by conditioning on the non-collider $A$).
  \item $B \not\!\perp\!\!\!\perp E \mid A$ \ (conditioning on the collider $A$ opens the backdoor).
\end{enumerate}

\paragraph{Stage for the earthquake DAG.}
Fix the DAG \(G\) on variables \(V=\{B,E,A,C\}\) with arrows
\(B\to A\leftarrow E\) and \(A\to C\).
A \emph{stage} \(U\) for this example is a context that packages together:
(i) the fixed graph \(G\), and
(ii) a finite menu of \emph{local regimes} (also called \emph{charts})
that we regard as legitimate descriptions of \(U\).
Each chart \(S\) is specified by a pair \((I_S,\mathsf{Cond}_S)\) where
\(I_S\subseteq V\) is a set of intervention targets (incoming arrows into \(I_S\)
are cut—\emph{surgical} semantics), and \(\mathsf{Cond}_S\subseteq \mathcal{P}(V)\)
records which conditioning sets are admissible when we evaluate
conditional independences on that chart.
A family of charts \(\{S_i\to U\}_i\) is a \(J\)\emph{-cover} of \(U\) if, by design,
these charts jointly describe all local ways in which \(U\) may be investigated
(e.g., purely observational vs.\ a specific intervention).

Given a CI formula \(\varphi\) (e.g.\ \(X \perp\!\!\!\perp Y \mid Z\)), we say that
\(\varphi\) is \emph{\(j\)-stable at \(U\)} and write \(U \Vdash j(\varphi)\) iff there exists a
\(J\)-cover \(\{S_i\to U\}_i\) such that \(\varphi\) holds on every chart \(S_i\),
where \(\varphi\) is evaluated by d-separation on the intervened graph \(G^{I_{S_i}}\)
(with the requested conditioning sets required to lie in \(\mathsf{Cond}_{S_i}\)).

\paragraph{\(j\)-stable reading for the earthquake DAG (concrete cover).}
We take the following two charts as a \(J\)-cover of \(U\):
\begin{itemize}
  \item \(S_{\mathrm{obs}}\): observational chart with \(I_{\mathrm{obs}}=\varnothing\), and
  \(\mathsf{Cond}_{\mathrm{obs}}\) containing the sets we explicitly condition on
  in the claims below (notably \(\{A\}\) and \(\varnothing\)).
  \item \(S_{\mathrm{int}}\): interventional chart with \(I_{\mathrm{int}}=\{A\}\)
  (cut incoming arrows into \(A\)); \(\mathsf{Cond}_{\mathrm{int}}\) likewise
  contains the conditioning sets we use below.
\end{itemize}
Then:
\begin{description}
  \item[(J–1)] \(B \perp\!\!\!\perp^{\,j}_{U} E\).
  On \(S_{\mathrm{obs}}\), the collider \(A\) is unconditioned \(\Rightarrow B \perp\!\!\!\perp E\).
  On \(S_{\mathrm{int}}\), parents of \(A\) are cut \(\Rightarrow B \perp\!\!\!\perp E\).
  Hence \(U \Vdash j\big(B \perp\!\!\!\perp E\big)\).
  \item[(J–2)] \(B \perp\!\!\!\perp^{\,j}_{U} C \mid A\).
  On \(S_{\mathrm{obs}}\), the chain \(B\!\to\!A\!\to\!C\) is blocked by conditioning on the
  non-collider \(A\) \(\Rightarrow B \perp\!\!\!\perp C \mid A\).
  On \(S_{\mathrm{int}}\), \(C\) depends only on \(A\) \(\Rightarrow B \perp\!\!\!\perp C \mid A\).
  Hence \(U \Vdash j\big(B \perp\!\!\!\perp C \mid A\big)\).
\end{description}

\paragraph{$j$-stable reading (same DAG, now as a site).}
Let $U$ denote a stage (context). We consider two legitimate regimes that may obtain at $U$:
\[
\mathcal{S}=\{\,S_{\mathrm{obs}}\hookrightarrow U,\; S_{\mathrm{int}}\hookrightarrow U\,\},
\]
a $J$-cover of $U$ whose \emph{charts} are:
\begin{itemize}
  \item $S_{\mathrm{obs}}$ (observational): ordinary d-separation semantics (no intervention on $A$).
  \item $S_{\mathrm{int}}$ (interventional): $\mathrm{do}(A)$—incoming edges into $A$ are cut.
\end{itemize}
By definition, a conditional independence (CI) is \emph{$j$-stable at $U$} if it holds on each chart in a $J$-cover of $U$; we write $U \Vdash j(\cdot)$.

\noindent The following $j$-stable CIs hold:
\begin{description}
  \item[(J–1)] \textbf{$B \perp\!\!\!\perp^{\,j}_{U} E$.}
  On $S_{\mathrm{obs}}$, the collider $A$ is unconditioned, so $B \perp\!\!\!\perp E$.  
  On $S_{\mathrm{int}}$, the parents of $A$ are cut, so again $B \perp\!\!\!\perp E$.  
  Hence $U \Vdash j\!\big(B \perp\!\!\!\perp E\big)$.
  \item[(J–2)] \textbf{$B \perp\!\!\!\perp^{\,j}_{U} C \mid A$.}
  On $S_{\mathrm{obs}}$, the chain $B\!\to\!A\!\to\!C$ is blocked by conditioning on the non-collider $A$, so $B \perp\!\!\!\perp C \mid A$.  
  On $S_{\mathrm{int}}$, $C$ depends only on $A$ (parents of $A$ are cut), so $B \perp\!\!\!\perp C \mid A$ again.  
  Therefore $U \Vdash j\!\big(B \perp\!\!\!\perp C \mid A\big)$.
\end{description}

\paragraph{$j$-stable reading (same DAG, now as a site).}
Fix an ambient stage $U$ in which we do not a priori know how analysts treat $A$.  
Exhibit a $J$-cover
\[
\mathcal{S}=\{\,S_{\mathrm{obs}}\hookrightarrow U,\; S_{\mathrm{int}}\hookrightarrow U\,\}
\]
with the following \emph{charts}:
\begin{itemize}
  \item $S_{\mathrm{obs}}$: an observational chart where conditioning on $A$ is admissible (we do not condition on descendants of colliders unless stated).
  \item $S_{\mathrm{int}}$: an interventional chart with $\mathrm{do}(A)$ (incoming edges into $A$ cut), so $B\to A\leftarrow E$ is surgically removed.
\end{itemize}

\noindent We claim the following $j$-stable CIs hold at $U$:
\begin{description}
  \item[(J–1)] \textbf{$B \perp\!\!\!\perp^{\,j}_{U} E$.}  
  On $S_{\mathrm{obs}}$, the collider $A$ is \emph{not} conditioned, so $B \perp\!\!\!\perp E$.  
  On $S_{\mathrm{int}}$, the incoming edges into $A$ are cut, so $B$ and $E$ do not meet at $A$, again $B \perp\!\!\!\perp E$.  
  Thus each chart in the cover validates the CI, hence $U \Vdash j\!\big(B \perp\!\!\!\perp E\big)$.
  \item[(J–2)] \textbf{$B \perp\!\!\!\perp^{\,j}_{U} C \mid A$.}  
  On $S_{\mathrm{obs}}$, standard d-separation on the chain $B\!-\!A\!-\!C$ gives $B \perp\!\!\!\perp C \mid A$.  
  On $S_{\mathrm{int}}$, $C$ depends only on $A$ (the parents of $A$ are cut), so $B \perp\!\!\!\perp C \mid A$ again.  
  Therefore $U \Vdash j\!\big(B \perp\!\!\!\perp C \mid A\big)$.
\end{description}

\paragraph{Non-example (why the classical warning persists).}
$B \not\!\perp\!\!\!\perp E \mid A$ on the observational chart: conditioning on the collider $A$ opens the path.  
Hence there is no cover $\mathcal{S}$ that includes $S_{\mathrm{obs}}$ and forces $B \perp\!\!\!\perp E \mid A$ on \emph{every} chart.  
This agrees with the classical rule: conditioning on colliders breaks independence.

\paragraph{Intuition.}
$j$-stability certifies a CI at stage $U$ by exhibiting a family of legitimate \emph{local} views (charts in a $J$-cover of $U$) such that every $X\leadsto Y$ path is blocked on each chart by the usual collider/non-collider rules.  
Because the CI is true in all charts that jointly \emph{cover} the epistemic situation at $U$, the forcing relation $U \Vdash j(\cdot)$ holds globally.

\subsection{Pollution DAG}

\begin{figure}[t]
    \centering
    \includegraphics[width=0.85\linewidth]{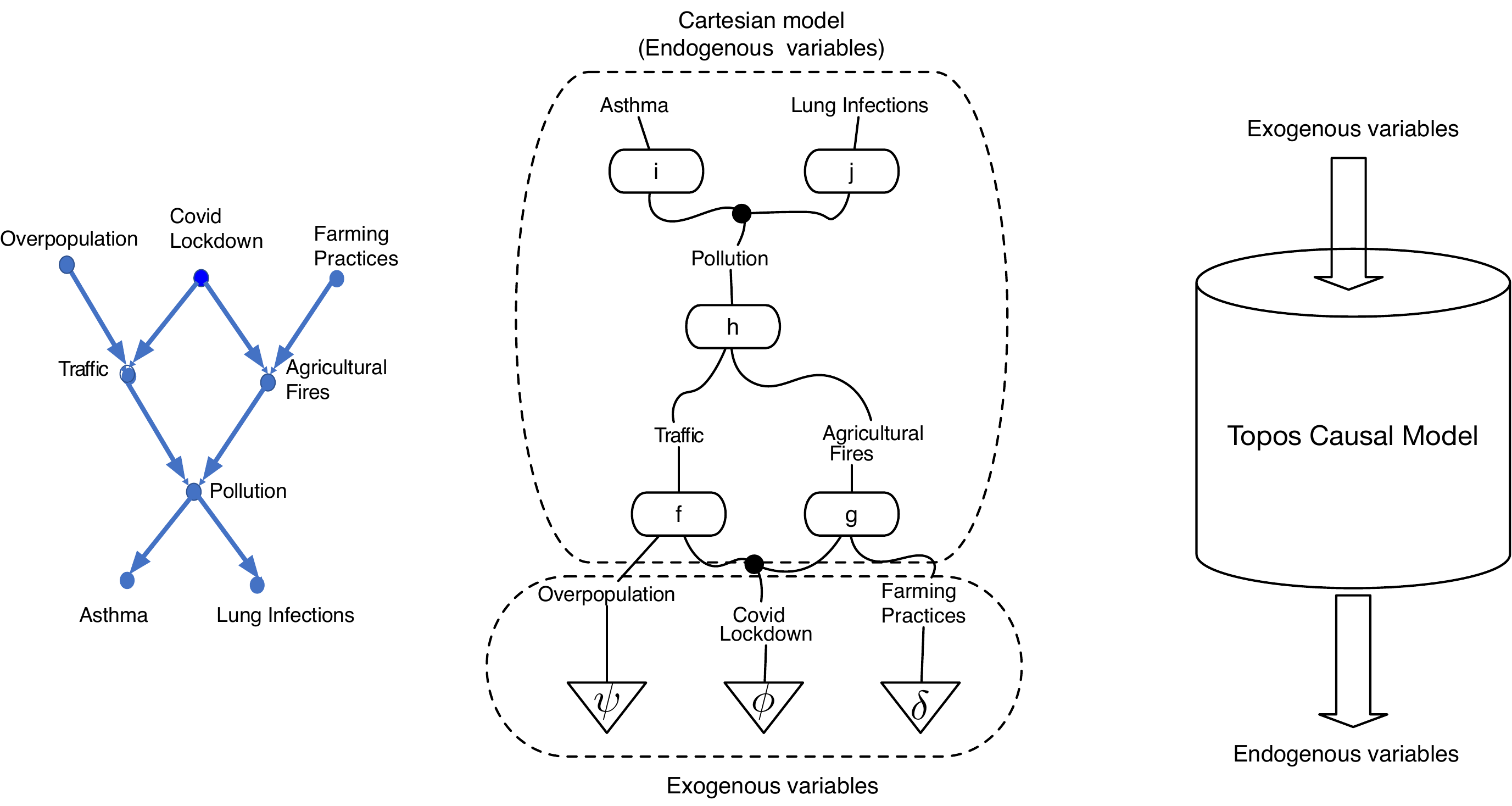}
    \caption{A simple causal model of pollution in New Delhi, India \citep{sm:neurips_tcm,DBLP:journals/entropy/Mahadevan23}.}
    \label{fig:pollution-dag}
\end{figure}

Next, we describe the concept of $j$-stability using the Pollution DAG from \citep{sm:neurips_tcm}, shown in Figure~\ref{fig:pollution-dag}.  We view a \emph{stage} \(U\) as an ambient situation in which analysts may be working either purely observationally or under well-specified interventions; a \emph{$J$-cover} \(\mathcal S\) of \(U\) is a family of charts \(S\hookrightarrow U\) (observational or interventional) whose union covers all admissible analytic contexts in \(U\). A CI statement \(\varphi\) is \emph{$J$-stable at \(U\)} if \(\varphi\) holds, by ordinary \(d\)-separation, on every chart in some cover of \(U\).

\medskip
\noindent\textbf{Cover A (Mobility–policy cover).}
Let \(\mathcal S_{\mathrm{mob}}=\{S_{\mathrm{obs}}\hookrightarrow U,\,S_{\mathrm{lock}}\hookrightarrow U\}\) with:
\begin{itemize}
\item \(S_{\mathrm{obs}}\): purely observational use of the graph.
\item \(S_{\mathrm{lock}}\): the interventional chart \(\mathrm{do}(\mathrm{Lockdown})\), cutting incoming edges into \(\mathrm{Traffic}\) from \(\mathrm{Lockdown}\) (and fixing its value).
\end{itemize}

\emph{Claim A1 (mediator blocking is $J$-stable).} 
\[
\mathrm{Traffic}\;\perp\!\!\!\perp\; \mathrm{Asthma}\; \mid \; \mathrm{Pollution}
\quad\text{is $J$-stable at }U.
\]
\emph{Reason.} On \(S_{\mathrm{obs}}\), all paths from \(\mathrm{Traffic}\) to \(\mathrm{Asthma}\) pass through the mediator \(\mathrm{Pollution}\); conditioning on the mediator blocks them by ordinary \(d\)-separation. On \(S_{\mathrm{lock}}\), \(\mathrm{Traffic}\) is set but the unique path to \(\mathrm{Asthma}\) still runs through \(\mathrm{Pollution}\), hence the same conditional independence holds. Thus every chart in \(\mathcal S_{\mathrm{mob}}\) blocks the path, so the CI is $J$-stable.

\emph{Claim A2 (same logic for respiratory outcomes).}
\[
\mathrm{Traffic}\;\perp\!\!\!\perp\; \mathrm{LungInfections}\; \mid \; \mathrm{Pollution}
\quad\text{is $J$-stable.}
\]
\emph{Reason.} Identical mediator argument as A1, with \(\mathrm{LungInfections}\) in place of \(\mathrm{Asthma}\).

\medskip
\noindent\textbf{Cover B (Fire–weather cover).}
Let \(\mathcal S_{\mathrm{fire}}=\{S_{\mathrm{dry}}\hookrightarrow U,\,S_{\mathrm{wet}}\hookrightarrow U\}\), where both charts are observational but represent distinct regimes for the Agricultural-Fires mechanism (high/low propensity). The graph’s adjacency is unchanged; only the mechanism strength varies.

\emph{Claim B1 (parents of a collider are $J$-independent unless we condition on it).}
\[
\mathrm{Traffic}\;\perp\!\!\!\perp\; \mathrm{AgriculturalFires}
\quad\text{is $J$-stable.}
\]
\emph{Reason.} In both \(S_{\mathrm{dry}}\) and \(S_{\mathrm{wet}}\), \(\mathrm{Traffic}\) and \(\mathrm{AgriculturalFires}\) meet only at the common child \(\mathrm{Pollution}\) (a collider). Without conditioning on \(\mathrm{Pollution}\) or its descendants, the collider blocks the path in each chart, so the independence holds on the entire cover.

\emph{Non-example (conditioning on the collider is not $J$-stable).}
\[
\mathrm{Traffic}\;\perp\!\!\!\perp\; \mathrm{AgriculturalFires}\; \mid \; \mathrm{Pollution}
\quad\text{is \emph{not} $J$-stable.}
\]
\emph{Reason.} Conditioning on the collider \(\mathrm{Pollution}\) opens the path in every chart; there is no cover that makes the CI true chartwise.  

Table~\ref{tab:tcm-jstable} summarizes these findings about $j$-stability. 

\begin{table}[t]
\centering
\small
\caption{$J$-stable CI facts on the Traffic–Pollution DAG (from the TCM figure). A stage $U$ is covered by charts (observational/interventional) and a CI is $J$-stable if it holds by ordinary $d$-separation in every chart of the cover.}
\resizebox{\linewidth}{!}{
\begin{tabular}{l l p{0.52\linewidth} l}
\toprule
\textbf{Claim} & \textbf{Charts used (cover)} & \textbf{Blocking rationale (per chart)} & \textbf{Verdict} \\
\midrule
$\mathrm{Traffic}\!\perp\!\mathrm{Asthma}\mid \mathrm{Pollution}$ 
& Cover~A: $S_{\text{obs}}$, $S_{\text{lock}}=\mathrm{do}(\mathrm{Lockdown})$
& In both charts, all $\mathrm{Traffic}\!\leadsto\!\mathrm{Asthma}$ paths go via mediator $\mathrm{Pollution}$; conditioning on the mediator blocks them (standard chain rule).
& $J$-stable \\
\addlinespace[2pt]
$\mathrm{Traffic}\!\perp\!\mathrm{LungInfections}\mid \mathrm{Pollution}$
& Cover~A: $S_{\text{obs}}$, $S_{\text{lock}}$
& Same mediator argument: any path from $\mathrm{Traffic}$ to $\mathrm{LungInfections}$ must pass through $\mathrm{Pollution}$; conditioning blocks in both charts.
& $J$-stable \\
\addlinespace[2pt]
$\mathrm{Traffic}\!\perp\!\mathrm{AgriculturalFires}$
& Cover~B: $S_{\text{dry}}$, $S_{\text{wet}}$ (both observational)
& Parents of the collider $\mathrm{Pollution}$: without conditioning on the collider or its descendants, the $\mathrm{Traffic}\!\to\!\mathrm{Pollution}\!\leftarrow\!\mathrm{AgriculturalFires}$ path is blocked in each chart.
& $J$-stable \\
\addlinespace[2pt]
$\mathrm{Traffic}\!\perp\!\mathrm{AgriculturalFires}\mid \mathrm{Pollution}$
& Any cover containing an observational chart 
& Conditioning on the collider $\mathrm{Pollution}$ opens $\mathrm{Traffic}\!\to\!\mathrm{Pollution}\!\leftarrow\!\mathrm{AgriculturalFires}$, so the CI fails in that chart.
& \textbf{Not} $J$-stable \\
\bottomrule
\end{tabular}}
\label{tab:tcm-jstable}
\end{table}

\subsection{Instantiating sieves and a Grothendieck topology on two DAGs}

We reuse the site \(\Cat\) of \emph{stages} \(\U=(G,\sigma)\) and \emph{refinements} \(u:V\!\to\!\U\) (status–monotone: refinements may condition/intervene on \emph{more} variables, never less). For a CI formula \(\varphi\equiv(X \CI Y \mid Z)\), recall the sieve
\[
\mathsf S_\varphi(\U)\;=\;\{\,u:V\!\to\!\U \mid V\models\varphi\,\}.
\]
A Grothendieck topology \(J\) is specified by choosing, for each \(\U\), a family of admissible \emph{charts} \(\{\rho_k:V_k\!\to\!\U\}\) (observational/interventional views). A sieve \(S\) covers \(\U\) iff it contains a jointly epimorphic refinement of that family. We use two concrete choices:
\[
J_{\mathrm{id}}:\text{ basis }=\{\mathrm{id}_\U\},\qquad
J_{\mathrm{mix}}:\text{ basis }=\{\rho_{\mathrm{obs}},\,\rho_{\mathrm{do}}\},
\]
where \(\rho_{\mathrm{obs}}\) is a purely observational chart and \(\rho_{\mathrm{do}}\) is a specific do-surgery chart indicated below for each DAG.

\vspace{0.4em}
\paragraph{(A) Earthquake/Alarm DAG.}
Variables: \(B=\text{burglary},\; E=\text{earthquake},\; A=\text{alarm},\; C=\text{neighbor calls}\).
Edges: \(B\!\to\!A\leftarrow\!E\) and \(A\!\to\!C\).
Let \(\U=(G,\sigma_{\emptyset})\) be the ambient stage with no conditioning or interventions fixed.

\emph{Charts generating \(J_{\mathrm{mix}}\).}
\begin{itemize}
\item \(\rho_{\mathrm{obs}}:V_{\mathrm{obs}}\!\to\!\U\): observational; no collider is conditioned unless stated.
\item \(\rho_{\mathrm{do}A}:V_{\mathrm{do}A}\!\to\!\U\): interventional; perform \(\mathrm{do}(A)\), i.e., cut \(B\!\to\!A\) and \(E\!\to\!A\).
\end{itemize}

\emph{Truth on charts and \(j\)-stability.} Let \(\varphi_1=(B\CI E)\), \(\varphi_2=(B\CI C\mid A)\), \(\varphi_3=(B\CI E\mid A)\).
\[
\begin{array}{l|ccc}
\text{Chart} & \varphi_1 & \varphi_2 & \varphi_3\\\hline
V_{\mathrm{obs}} & \text{true (collider closed)} & \text{true (chain blocked by }A\text{)} & \text{false (collider opened)}\\
V_{\mathrm{do}A} & \text{true (parents of }A\text{ cut)} & \text{true (}C\text{ depends only on }A\text{)} & \text{true (}B,E\text{ disconnected)}
\end{array}
\]
Hence \(\mathsf S_{\varphi_1}(\U)\) and \(\mathsf S_{\varphi_2}(\U)\) contain both generators, so they are \(J_{\mathrm{mix}}\)-covers and
\[
\U\Vdash_{J_{\mathrm{mix}}}\varphi_1,\qquad \U\Vdash_{J_{\mathrm{mix}}}\varphi_2.
\]
By contrast, \(\mathsf S_{\varphi_3}(\U)\) omits \(\rho_{\mathrm{obs}}\), so it does \emph{not} cover \(\U\); thus \(\U\not\Vdash_{J_{\mathrm{mix}}}\varphi_3\).
With \(J_{\mathrm{id}}\) we recover the classical verdicts, since covering reduces to truth at \(\U\) itself.

\vspace{0.4em}
\paragraph{(B) Pollution/Smoker/Cancer/X-ray DAG.}
Variables: \(P=\text{pollution},\; S=\text{smoker},\; C=\text{cancer},\; X=\text{x-ray}\).
Edges: \(P\!\to\!C,\; S\!\to\!C,\; C\!\to\!X\) (the standard Pearl example).
Ambient stage \(\U=(G,\sigma_{\emptyset})\).

\emph{Charts generating \(J_{\mathrm{mix}}.\)}
\begin{itemize}
\item \(\rho_{\mathrm{obs}}:V_{\mathrm{obs}}\!\to\!\U\): observational.
\item \(\rho_{\mathrm{do}C}:V_{\mathrm{do}C}\!\to\!\U\): \(\mathrm{do}(C)\) (cut \(P\!\to\!C\) and \(S\!\to\!C\)).
\end{itemize}

We examine \(\psi_1=(P\CI S)\), \(\psi_2=(P\CI X\mid C)\), \(\psi_3=(P\CI S\mid C)\).
\[
\begin{array}{l|ccc}
\text{Chart} & \psi_1 & \psi_2 & \psi_3\\\hline
V_{\mathrm{obs}} & \text{true (no path }P\!\leadsto\!S)& \text{true (}X\text{ is child of }C\text{)} & \text{false (conditioning on collider }C\text{ opens }P\!\dashrightarrow\!S)\\
V_{\mathrm{do}C} & \text{true (}C\text{ parents cut)} & \text{true (}X\!\perp\!P\mid C\text{ vacuously)} & \text{true (}P,S\text{ disconnected)}
\end{array}
\]
Thus \(\mathsf S_{\psi_1}(\U)\) and \(\mathsf S_{\psi_2}(\U)\) contain both generators and cover \(\U\), giving
\[
\U\Vdash_{J_{\mathrm{mix}}}\psi_1,\qquad \U\Vdash_{J_{\mathrm{mix}}}\psi_2.
\]
But \(\mathsf S_{\psi_3}(\U)\) misses \(\rho_{\mathrm{obs}}\), so \(\U\not\Vdash_{J_{\mathrm{mix}}}\psi_3\).
Again, \(J_{\mathrm{id}}\) collapses to classical \(d\)-separation at \(\U\).

\paragraph{Remarks.}
(i) In both DAGs, \(\mathsf S_\phi(\U)\) is a \emph{sieve} by monotonicity of refinements: precomposing with a further refinement can only block more paths.  
(ii) \(J_{\mathrm{mix}}\) encodes the methodological stance “a CI may be certified if it holds on a fixed menu of admissible local views (e.g., observational and one specified do-chart).’’  
(iii) Choosing richer bases (e.g., including \(\mathrm{do}\)-charts at additional nodes) strengthens \(j\)-stability: more sieves cover, so more CIs become \(j\)-stable, while \(J_{\mathrm{id}}\) recovers classical independence exactly (\(\U\Vdash_{J_{\mathrm{id}}}\phi \iff \U\models\phi\)).

\medskip
\noindent\textbf{Summary.} These examples show how $j$-stability is a \emph{coverwise} lift of ordinary \(d\)-separation: pick a semantically appropriate family of observational/interventional charts (the cover) for the ambient stage \(U\), and then require the classical blocking rules to hold in each chart. When they do, the CI is forced by \(j\) at \(U\) (hence $J$-stable).%

\begin{lemma}[Collider–opening CI and covers]
Let $\Cat$ be the site of stages and refinements described in the paper.
Fix disjoint variable-sets $X,Y,Z$ in a DAG $G$, and the CI formula
$\varphi=(X \CI Y \mid Z)$.
Assume that in $G$ every $X\leadsto Y$ path contains a collider $C$ with $C\in Z$ (so $\varphi$ is \emph{collider–opening} in the classical sense).

Define two Grothendieck topologies by bases of charts at an ambient stage $\U=(G,\sigma_\emptyset)$:
\[
J_{\mathrm{obs}}:=\text{basis } \{\rho_{\mathrm{obs}}\},\qquad
J_{\mathrm{do}(C)}:=\text{basis } \{\rho_{\mathrm{do}(C)}\},
\]
where $\rho_{\mathrm{obs}}$ is purely observational and
$\rho_{\mathrm{do}(C)}$ is the interventional chart cutting all incoming arrows into $C$.

Then:
\begin{enumerate}
\item \emph{(Never covers observationally)} $\U \not\Vdash_{J_{\mathrm{obs}}} \varphi$.
Equivalently, the sieve $\mathsf S_{\varphi}(\U)=\{\,u:V\!\to\!\U \mid V\models \varphi\,\}$ does not cover $\U$ under $J_{\mathrm{obs}}$.

\item \emph{(Always covers under the collider–cut do–chart)}
$\U \Vdash_{J_{\mathrm{do}(C)}} \varphi$.
Equivalently, $\mathsf S_{\varphi}(\U)$ covers $\U$ under $J_{\mathrm{do}(C)}$.
\end{enumerate}
\end{lemma}

\begin{proof}
(1) In the observational chart $\rho_{\mathrm{obs}}:\Vobs\to\U$, conditioning on the collider $C\in Z$ \emph{opens} every $X\leadsto Y$ path (classical $d$-separation). Hence $\Vobs \not\models \varphi$, so $\rho_{\mathrm{obs}}\notin \mathsf S_{\varphi}(\U)$. Since $\{\rho_{\mathrm{obs}}\}$ is a covering family for $J_{\mathrm{obs}}$, no sieve missing it can cover; thus $\U \not\Vdash_{J_{\mathrm{obs}}}\varphi$.

(2) In the do–chart $\rho_{\mathrm{do}(C)}:\VdoC\to\U$, the surgery removes all incoming arrows into $C$, so every $X\leadsto Y$ path that previously hinged on $C$ is severed. Consequently $X \CI Y \mid Z$ holds in $\VdoC$ (indeed $X$ and $Y$ are $d$-separated regardless of whether we condition on $C$), i.e.\ $\VdoC \models \varphi$. Therefore $\rho_{\mathrm{do}(C)}\in \mathsf S_{\varphi}(\U)$; because $\{\rho_{\mathrm{do}(C)}\}$ is a covering family for $J_{\mathrm{do}(C)}$, the sieve $\mathsf S_{\varphi}(\U)$ covers $\U$, and $\U \Vdash_{J_{\mathrm{do}(C)}} \varphi$.
\end{proof}

\begin{corollary}[Earthquake/Alarm and Pollution examples]
For the Earthquake DAG ($B\to A\leftarrow E$, $A\to C$) with $\varphi_3=(B\CI E\mid A)$ and $C=A$,
we have $\U\not\Vdash_{J_{\mathrm{obs}}}\varphi_3$ but $\U\Vdash_{J_{\mathrm{do}(A)}}\varphi_3$.
For the Pollution DAG ($P\to C\leftarrow S$, $C\to X$) with $\psi_3=(P\CI S\mid C)$ and $C$ the collider,
$\U\not\Vdash_{J_{\mathrm{obs}}}\psi_3$ yet $\U\Vdash_{J_{\mathrm{do}(C)}}\psi_3$.
\end{corollary}

\begin{remark}[Designing $J$ as methodological stance]
More generally, if a chosen basis $\mathcal{B}=\{\rho_k:V_k\!\to\!\U\}$ has the property that each chart $\rho_k$ blocks all $X\leadsto Y$ paths given $Z$ (by observation rules or appropriate do–surgeries), then the sieve $\mathsf S_{(X\CI Y\mid Z)}(\U)$ contains $\mathcal{B}$ and is therefore covering. Thus $J$ encodes \emph{which local views count as legitimate evidence} for a CI:
adding do–charts that “neutralize” colliders makes collider–opening CIs $j$-stable, whereas a purely observational $J_{\mathrm{obs}}$ never certifies them.
\end{remark}

\section{Probabilistic Inference in TCM}
\label{sec:monads} 

In \citep{sm:neurips_tcm}, causal models were defined either as objects of a topos category called ${\cal C_{TCM}}$ (which will be often denoted simply by $\E$ in this paper), or as functor objects in a presheaf topos $\Set^{{\cal C}^{op}_{TCM}}$. Here, we will get into a more specialized modeling framework, which is less general, but gives us the opportunity to develop a more refined language for translating classical do-calculus into an intuitionistic framework.  We begin by introducing a categorical framework for probabilistic inference. Some of the main ideas that will be explained in detail below are summarized at a high level in Figure~\ref{fig:do-mini-example} and Table~\ref{tab:classical-vs-tcm-compact}, and Figure~\ref{fig:dag-vs-tcm-companion-compact} gives a high-level summary.  A standard way to model probability distributions categorically is through monads \citep{giri}. For simplicity, we restrict ourselves to the case of distribution monads over finite sets. Let $\mathbf{FinSet}$ (or more generally $\mathcal{E}$) denote the base category of finite sets (or objects of the ambient topos).  The \emph{distribution monad} \citep{jacobs:book} 
\[
\Dist : \mathbf{FinSet} \longrightarrow \mathbf{FinSet}
\]
maps an object $X$ to the set of finitely supported probability measures on $X$:
\[
\Dist(X) := \Big\{\, p:X\to[0,1] \ \big|\  \sum_{x\in X}p(x)=1 \,\Big\}.
\]
For a morphism $f:X\to Y$, $\Dist(f):\Dist(X)\to\Dist(Y)$ is the pushforward of measures,
\[
\Dist(f)(p)(y) := \sum_{x:\,f(x)=y} p(x).
\]
The unit $\eta_X:X\to\Dist(X)$ sends $x\mapsto \delta_x$ (Dirac measure), and the multiplication 
$\mu_X:\Dist(\Dist(X))\to\Dist(X)$ is integration (flattening of distributions of distributions).

The \emph{Kleisli category} $\mathbf{FinStoch}$ of $\Dist$ has:
\begin{itemize}
\item the same objects as $\mathbf{FinSet}$,
\item morphisms $X\to Y$ given by stochastic kernels $X\to\Dist(Y)$,
\item composition given by convolution:
  \[
  (g\ast f)(x) := \int_Y g(y)\,d f(x)(y).
  \]
\end{itemize}
We write $\Dist(Y)$ also when $\mathcal{E}$ is a topos and $\Dist$ is lifted to an \emph{internal distribution monad},  
so that morphisms $X\to\Dist(Y)$ represent internal stochastic maps. 
All constructions in this paper are interpreted internally in the Kleisli (Markov) category of $\Dist$.

\subsection{The Kleisli Category of the Distribution Monad}
\label{subsec:markov-category}

\paragraph{Definition of the distribution monad.}
Let $\E$ be a topos (or a suitable cartesian category)
equipped with a specified \emph{finite-distribution monad}
\[
\Dist_\E : \E \to \E,
\]
interpreting probabilistic choice or convex combination.
Each object $X\in\E$ is mapped to an object $\Dist_\E(X)$ of finitely supported
distributions over~$X$, and each morphism $f:X\to Y$ is mapped to the
\emph{pushforward} $\Dist_\E(f):\Dist_\E(X)\to \Dist_\E(Y)$.
The monad structure consists of:
\begin{align*}
  \eta_X &: X \to \Dist_\E(X), & \eta_X(x) &= 1\,|x\rangle \quad \text{(Dirac embedding)},\\
  \mu_X &: \Dist_\E(\Dist_\E(X)) \to \Dist_\E(X),
  & \mu_X\!\left(\sum_j q_j\Big|\sum_i p_{ji}\,|x_{ji}\rangle\Big\rangle\right)
    &= \sum_{j,i} q_j p_{ji}\,|x_{ji}\rangle,
\end{align*}
which flatten a distribution of distributions.
We treat the existence of this internal monad, its strength, and its
compatibility with the selected sheaf subtopos as standing hypotheses.
They do not follow from the topos axioms alone.  In the main soundness
theorem one may instead work objectwise with finite probability kernels,
which is the concrete setting used below.

\paragraph{The Kleisli category.}
The \emph{Kleisli category} of $\Dist_\E$, denoted
\[
\Kl(\Dist_\E),
\]
is the category having:
\begin{itemize}
  \item the same objects as $\E$;
  \item morphisms $X\!\to\! Y$ given by arrows $X\!\to\!\Dist_\E(Y)$ in~$\E$,
        written $f:X\!\Rightarrow\!Y$;
  \item composition defined by \emph{Kleisli convolution}:
        for $f:X\!\Rightarrow\!Y$ and $g:Y\!\Rightarrow\!Z$, their composite is
        \[
        g\odot f \;:=\;
        \mu_Z \circ \Dist_\E(g) \circ f
        \;:\;
        X \to \Dist_\E(Z).
        \]
\end{itemize}
The identity on $X$ is the Dirac morphism $\eta_X:X\to\Dist_\E(X)$.
Associativity and unitality of $\odot$ follow from the monad laws.

\paragraph{Intuition.}
A morphism $f:X\Rightarrow Y$ represents a \emph{stochastic kernel}:
it assigns to each generalized element $x:N\!\to\!X$ an internal probability
distribution $f(x):N\!\to\!\Dist_\E(Y)$.
Composition $g\odot f$ corresponds to \emph{marginalizing over the intermediate
variable $Y$}:
\[
(g\odot f)(x)
  \;=\;
  \int_Y g(y)\,df(x)(y)
  \;=\;
  \sum_{y\in Y} g(y)\,f(x)(y)
  \quad\text{(internally in $\E$)}.
\]
Hence $\Kl(\Dist_\E)$ behaves as the internal version of the category of
finite-state Markov kernels, and is often called the
\emph{Markov category} of~$\E$.

\paragraph{Commutative and strong structure.}
The monad $\Dist_\E$ is \emph{commutative} and \emph{strong}:
there exists a natural transformation
\[
st_{X,Y}: X\times \Dist_\E(Y) \to \Dist_\E(X\times Y),
\qquad
st_{X,Y}(x,p)=\sum_{y} p(y)\,|(x,y)\rangle,
\]
which allows one to handle dependent random variables
and to interpret causal composition diagrammatically.
This strength is precisely what defines interventions:
given a context $\Gamma$, policy $\mu:\Gamma\Rightarrow Z$, and structural kernel
$k:\Gamma\times Z\Rightarrow Y$, the interventional composite
\[
\Do_Z(k;\mu) \;=\;
\mu_Y \circ \Dist_\E(k)\circ st_{\Gamma,Z}\circ \langle \id,\mu\rangle
\]
is the Kleisli composition $\Gamma\Rightarrow Y$
that corresponds to integrating $k$ against~$\mu$.

\paragraph{Categorical properties.}
The Kleisli category $\Kl(\Dist_\E)$ satisfies the axioms of a
\emph{Markov category} (Fritz, 2020):
\begin{itemize}
  \item it has a symmetric monoidal structure inherited from $\E$;
  \item the comonoid structure $(X\leftarrow X\times X \leftarrow X)$
        represents duplication and deletion of deterministic information;
  \item every morphism $X\Rightarrow Y$ is a stochastic map,
        and deterministic maps arise from the embedding
        $\eta_Y\circ f : X\to\Dist_\E(Y)$.
\end{itemize}
Thus $\Kl(\Dist_\E)$ internalizes the category of
\emph{probabilistic processes, kernels, and causal mechanisms}.
It is the natural semantic environment for Topos Causal Models.

\begin{figure}[p]
  \centering
  \begin{minipage}{0.54\textwidth}
    \centering
    \begin{tikzpicture}[>=Latex, node distance=25mm]
      \tikzstyle{v}=[circle,draw,minimum size=10mm,inner sep=0pt]
      \node[v] (Z) {$Z$};
      \node[v, right=of Z] (X) {$X$};
      \node[v, right=of X] (Y) {$Y$};

      \draw[->] (Z) -- (X);
      \draw[->] (X) -- (Y);
      \draw[->, bend right=35] (Z) to (Y);

      \node[below=9mm of X] {\scriptsize Classical causal graph};
    \end{tikzpicture}

    \vspace{0.3em}
    \small (a) Classical DAG: \(Z \to X \to Y\), with optional confounding \(Z\to Y\).
    Intervening with \(do(X)\) means cutting incoming arrows to \(X\)
    and fixing its distribution via a policy~$\mu_X$.
  \end{minipage} \newline \vskip 0.2in 
  \hfill
  \begin{minipage}{0.9\textwidth}
    \centering
    \begin{tikzcd}[row sep=large, column sep=large]
      Z
        \arrow[r, "k_X"]
        \arrow[dr, swap, "{\id_Z}"] &
      \Dist(X)
        \arrow[r, "{st}"] &
      \Dist(Z\times X)
        \arrow[r, "{\Dist(k_Y)}"] &
      \Dist(\Dist(Y))
        \arrow[r, "\mu_Y"] &
      \Dist(Y) \\
      & Z
        \arrow[ur, bend right=5, swap, "{\langle \id_Z,\mu_X\rangle}"]
        \arrow[ur, dashed, bend right=30, swap, "{\Do_X(k_Y;\mu_X)}"]
        &
    \end{tikzcd}

    \vspace{0.3em}
    \small (b) TCM diagram: the observational channel
    uses the learned kernel $k_X:Z\to\Dist(X)$;
    the intervention replaces it with $\mu_X$
    and integrates through the Kleisli composition
    $\Do_X(k_Y;\mu_X)=\mu_Y\!\circ\!\Dist(k_Y)\!\circ\!st\!\circ\!\langle \id_Z,\mu_X\rangle$.
  \end{minipage}

  \caption{From classical to TCM view of causal do-calculus interventions \(do(X)\).
  (a) In the DAG, incoming edges to \(X\) are cut and replaced by a fixed policy.
  (b) In the TCM, this is expressed by replacing the kernel
  $k_X:Z\to\Dist(X)$ with a chosen $\mu_X$ and composing via the
  monadic integration law. Both yield the equality
  $P(Y\mid do(X),Z)=\int_X P(Y\mid X,Z)\,d\mu_X(X)$,
  which in the internal logic reads
  $Z\vdash P(Y\mid do(X),Z)=P(Y\mid Z)$
  whenever $Y\!\perp\! X\mid Z$ in the cut model.}
  \label{fig:do-mini-example}
\end{figure}

\begin{table}[p]
\centering
\begin{scriptsize}
\renewcommand{\arraystretch}{1.15}
\setlength{\tabcolsep}{3.8pt}
\begin{tabular}{|p{3.4cm}|p{4.8cm}|p{7.5cm}|}
\hline
\textbf{Causal Operation (Classical)} &
\textbf{TCM / Category‐Theoretic Analogue} &
\textbf{Description / Interpretation} \\
\hline
Node (variable) &
Object in $\E$ &
Each random variable $X$ is an object of the topos $\E$
representing possible states. \\
\hline
Directed edge $X\!\to\!Y$ &
Stochastic morphism $k_Y\!:X\!\to\!\Dist_\E(Y)$ &
A causal mechanism mapping each $x$ to a distribution on $Y$. \\
\hline
Joint distribution factorization &
Kleisli composition in $\Kl(\Dist_\E)$ &
Composition of stochastic morphisms yields the global joint law. \\
\hline
Conditioning on $X{=}x$ &
Comprehension subobject $\iota_x\!:\Gamma{\mid}X{=}x\hookrightarrow\Gamma$ &
Restrict to a subobject and renormalize (Bayesian update). \\
\hline
Marginalization &
Integration $\mu_Y\!:\Dist(\Dist(Y))\!\to\!\Dist(Y)$ &
Collapse nested distributions (expectation operator). \\
\hline
Observation (likelihood weighting) &
Restriction + normalization &
Apply a predicate $\chi$ as a subobject, then renormalize. \\
\hline
Intervention $do(X\!\sim\!\mu_X)$ &
Kernel replacement $\Do_X(k;\mu_X)
=\mu_Y\!\circ\!\Dist(k)\!\circ\!st\!\circ\!\langle\id,\mu_X\rangle$ &
Replace incoming kernel $k_X$ by $\mu_X$ and propagate. \\
\hline
Cutting an edge $Z\!\to\!X$ &
Replace $k_X\!:\!Z\!\to\!\Dist(X)$ by constant kernel &
Removes parent dependence (“mutilation”). \\
\hline
Conditional independence $Y\!\perp\!X\mid Z$ &
Factorization $k_Y\!=\!k_{0,Y}\!\circ\!\pi_Z$ in $\E$ &
Independence $\Leftrightarrow$ equality of arrows. \\
\hline
Rule 1: Insert/Delete observations &
Equality $\Do_X(k;\mu_X)\!=\!k_0$ under independence &
Integrating constant kernel yields same $k_0$. \\
\hline
Rule 2: Action/observation exchange &
Equality in $M_{\overline{X(W)}}$ via comprehension and replacement &
Observation/intervention equivalence under factorization. \\
\hline
Rule 3: Insert/Delete actions &
Stability of $k$ under substitution of $\mu_X$ &
Intervention on irrelevant variable leaves $k$ unchanged. \\
\hline
d‐separation test &
Pullback condition on subobjects &
Independence encoded as a commuting pullback square in $\E$. \\
\hline
\end{tabular}
\caption{Classical vs.\ TCM semantics.
Each do‐calculus rule has a categorical analogue:
internal morphism equality, kernel replacement,
or subobject inclusion.}
\label{tab:classical-vs-tcm-compact}
\end{scriptsize}
\end{table}

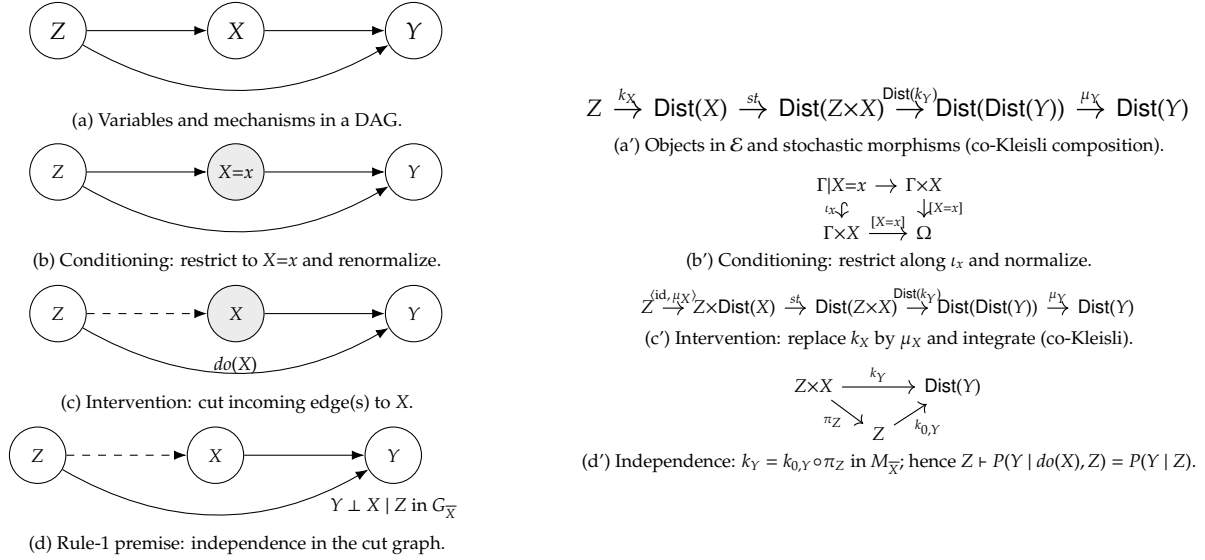
\begin{figure}[t]
  \centering
  \begin{small}

  \begin{minipage}{0.37\linewidth}
    \centering

    \begin{tikzpicture}[>=Latex, node distance=16mm]
      \tikzstyle{v}=[circle,draw,minimum size=7.5mm,inner sep=0pt]
      \node[v] (Z) {$Z$};
      \node[v, right=of Z] (X) {$X$};
      \node[v, right=of X] (Y) {$Y$};
      \draw[->] (Z) -- (X);
      \draw[->] (X) -- (Y);
      \draw[->, bend right=30] (Z) to (Y);
    \end{tikzpicture}

    \vspace{0.2em}
    \scriptsize (a) Variables and mechanisms in a DAG.

    \vspace{0.6em}

    \begin{tikzpicture}[>=Latex, node distance=16mm]
      \tikzstyle{v}=[circle,draw,minimum size=7.5mm,inner sep=0pt]
      \node[v, fill=gray!15] (X) {$X{=}x$};
      \node[v, left=of X] (Z) {$Z$};
      \node[v, right=of X] (Y) {$Y$};
      \draw[->] (Z) -- (X);
      \draw[->] (X) -- (Y);
      \draw[->, bend right=30] (Z) to (Y);
    \end{tikzpicture}

    \vspace{0.2em}
    \scriptsize (b) Conditioning: restrict to $X{=}x$ and renormalize.

    \vspace{0.6em}

    \begin{tikzpicture}[>=Latex, node distance=16mm]
      \tikzstyle{v}=[circle,draw,minimum size=7.5mm,inner sep=0pt]
      \node[v] (Z) {$Z$};
      \node[v, right=of Z, fill=gray!15] (X) {$X$};
      \node[v, right=of X] (Y) {$Y$};
      \draw[->, dashed] (Z) -- (X); 
      \draw[->] (X) -- (Y);
      \draw[->, bend right=30] (Z) to (Y);
      \node[below=1.0mm of X] {\scriptsize $do(X)$};
    \end{tikzpicture}

    \vspace{0.2em}
    \scriptsize (c) Intervention: cut incoming edge(s) to $X$.

    \vspace{0.6em}

    \begin{tikzpicture}[>=Latex, node distance=16mm]
      \tikzstyle{v}=[circle,draw,minimum size=7.5mm,inner sep=0pt]
      \node[v] (Z) {$Z$};
      \node[v, right=of Z] (X) {$X$};
      \node[v, right=of X] (Y) {$Y$};
      \draw[->, dashed] (Z) -- (X);
      \draw[->] (X) -- (Y);
      \draw[->, bend right=30] (Z) to (Y);
      \node[below=1.0mm of Y] {\scriptsize $Y \perp X \mid Z$ in $G_{\overline X}$};
    \end{tikzpicture}

    \vspace{0.2em}
    \scriptsize (d) Rule-1 premise: independence in the cut graph.
  \end{minipage}
  \hfill
  \begin{minipage}{0.58\linewidth}
    \centering

    \begin{tikzcd}[row sep=small, column sep=small]
      Z \arrow[r, "k_X"] &
      \Dist(X) \arrow[r, "st"] &
      \Dist(Z{\times}X) \arrow[r, "\Dist(k_Y)"] &
      \Dist(\Dist(Y)) \arrow[r, "\mu_Y"] &
      \Dist(Y)
    \end{tikzcd}

    \vspace{0.2em}
    \scriptsize (a') Objects in $\E$ and stochastic morphisms (Kleisli composition).

    \vspace{0.6em}

    \begin{tikzcd}[row sep=small, column sep=small]
      \Gamma{\mid}X{=}x \arrow[d, hook, "\iota_x"']
        \arrow[r] & \Gamma{\times}X \arrow[d, "{[X{=}x]}"] \\
      \Gamma{\times}X \arrow[r, "{[X{=}x]}"] & \Omega
    \end{tikzcd}

    \vspace{0.2em}
    \scriptsize (b') Conditioning: restrict along $\iota_x$ and normalize.

    \vspace{0.6em}

    \begin{tikzcd}[row sep=small, column sep=small]
      Z \arrow[r, "{\langle \mathrm{id},\, \mu_X\rangle}"] &
      Z{\times}\Dist(X) \arrow[r, "st"] &
      \Dist(Z{\times}X) \arrow[r, "\Dist(k_Y)"] &
      \Dist(\Dist(Y)) \arrow[r, "\mu_Y"] &
      \Dist(Y)
    \end{tikzcd}

    \vspace{0.2em}
    \scriptsize (c') Intervention: replace $k_X$ by $\mu_X$ and integrate (Kleisli).

    \vspace{0.6em}

    \begin{tikzcd}[row sep=small, column sep=small]
      Z{\times}X \arrow[dr, swap, "\pi_Z"] \arrow[rr, "k_Y"] & & \Dist(Y) \\
      & Z \arrow[ur, swap, "k_{0,Y}"] &
    \end{tikzcd}

    \vspace{0.2em}
    \scriptsize (d') Independence: $k_Y = k_{0,Y}\!\circ\!\pi_Z$ in $M_{\overline X}$; hence
    $Z \vdash P(Y\mid do(X),Z)=P(Y\mid Z)$.
  \end{minipage}

  \caption{Side-by-side correspondence between the classical DAG view (left) and the TCM categorical view (right).
  Each DAG operation (conditioning, edge deletion, independence) maps to a categorical construction:
  comprehension subobject $+$ normalization, kernel replacement $+$ integration, and factorization in $\E$.}
  \label{fig:dag-vs-tcm-companion-compact}

  \end{small}
\end{figure}

\paragraph{Summary diagram.}
\[
\begin{tikzcd}[row sep=large, column sep=huge]
  X \arrow[r, "f"] & \Dist_\E(Y)
    \arrow[r, "\Dist_\E(g)"] &
  \Dist_\E(\Dist_\E(Z))
    \arrow[r, "\mu_Z"] &
  \Dist_\E(Z)
\end{tikzcd}
\qquad
\text{represents}\quad
(g\odot f):X\Rightarrow Z.
\]

\paragraph{Connection to TCMs.}
In a TCM, every causal mechanism
$P(Y\mid \mathrm{Pa}(Y)):\mathrm{Pa}(Y)\!\Rightarrow\!Y$
is a morphism in $\Kl(\Dist_\E)$.
Interventions and observations are morphism replacements or pullbacks
within this category, and all do-calculus equalities are internal equalities
between Kleisli arrows of~$\Dist_\E$.

\subsection{Example: Causal Models as Functors}
\label{subsec:causal-models-as-functors}

\paragraph{1. Causal graph as a category.}
Let $\mathcal{C}_G$ be the free category generated by the graph
\[
X \;\longrightarrow\; Y \;\longrightarrow\; Z.
\]
Objects: $\mathrm{Ob}(\mathcal{C}_G)=\{X,Y,Z\}$.
Morphisms:
\[
\mathrm{Hom}(X,Y)=\{f\},\quad
\mathrm{Hom}(Y,Z)=\{g\},\quad
\mathrm{Hom}(X,Z)=\{g\circ f\},\quad
\mathrm{Hom}(X,X)=\{\id_X\},\;\dots
\]

\paragraph{2. Target category of stochastic maps.}
Let $\Kl(\Dist)$ denote the Kleisli category of the
finite-support distribution monad $\Dist$ on \textbf{Set}.
Objects are finite sets, morphisms $A\to B$ are stochastic matrices
$A\to \Dist(B)$, and composition is
\[
(h\circ f)(a)(c) = \sum_{b\in B} f(a)(b)\,h(b)(c).
\]

\paragraph{3. Causal model as a functor.}
A causal model is a functor
\[
F:\mathcal{C}_G \longrightarrow \Kl(\Dist),
\]
given on objects and morphisms by
\[
F(X)=\{x_1,x_2\},\quad
F(Y)=\{y_1,y_2\},\quad
F(Z)=\{z_1,z_2\},
\]
and stochastic matrices
\[
F(f):X\to\Dist(Y)
=\begin{bmatrix}
0.8 & 0.2\\
0.3 & 0.7
\end{bmatrix},\qquad
F(g):Y\to\Dist(Z)
=\begin{bmatrix}
0.9 & 0.1\\
0.4 & 0.6
\end{bmatrix}.
\]
Composition in $\mathcal{C}_G$ gives
\[
F(g\circ f) = F(g)\circ F(f),
\]
which in $\Kl(\Dist)$ is the matrix product:
\[
F(g\circ f)
=\begin{bmatrix}
0.8 & 0.2\\
0.3 & 0.7
\end{bmatrix}
\!\!\begin{bmatrix}
0.9 & 0.1\\
0.4 & 0.6
\end{bmatrix}
=
\begin{bmatrix}
0.8\!\times\!0.9+0.2\!\times\!0.4 & 0.8\!\times\!0.1+0.2\!\times\!0.6\\
0.3\!\times\!0.9+0.7\!\times\!0.4 & 0.3\!\times\!0.1+0.7\!\times\!0.6
\end{bmatrix}
=
\begin{bmatrix}
0.80 & 0.20\\
0.57 & 0.43
\end{bmatrix}.
\]
Thus $F(g\!\circ\!f):X\to\Dist(Z)$ represents
the induced causal influence from $X$ to $Z$.

\paragraph{4. Diagrammatic view.}
\begin{center}
\begin{tikzcd}[column sep=huge, row sep=large]
  X \arrow[r, "f"] \arrow[rr, bend left=30, "g\!\circ\!f"] &
  Y \arrow[r, "g"] & Z
\end{tikzcd}
\[\quad\mapsto\quad \]
\begin{tikzcd}[column sep=huge, row sep=large]
  F(X) \arrow[r, "F(f):X\to\Dist(Y)"] \arrow[rr, bend left=30, "F(g\circ f)"] &
  F(Y) \arrow[r, "F(g):Y\to\Dist(Z)"] & F(Z)
\end{tikzcd}
\end{center}

\paragraph{5. Functorial laws.}
The functor $F$ preserves identities and composition:
\[
F(\id_X)=\id_{F(X)},\qquad
F(g\circ f)=F(g)\circ F(f).
\]
Hence, causal composition (via functional or stochastic dependence)
is represented as categorical composition in the Markov category.

\paragraph{6. Generalization inside a topos.}
If $\E$ is a topos, we can internalize this by taking
\[
F : \mathcal{C}_G \longrightarrow \Kl(\Dist_\E),
\]
where each variable is now an object of $\E$,
each arrow a stochastic morphism in $\E$,
and the same functorial laws hold internally.
This internalizes causal semantics in any intuitionistic setting,
allowing higher-order, context-dependent, or sheaf-based models.

\subsection{Interventions as Natural Transformations}
\label{subsec:interventions-as-nattrans}

Given a causal graph category $\mathcal{C}_G$,
a causal model is a functor
\[
F : \mathcal{C}_G \longrightarrow \Kl(\Dist_\E).
\]
An intervention $do(X\!\sim\!\mu_X)$
produces a modified functor
\[
F_{do(X)} : \mathcal{C}_G \longrightarrow \Kl(\Dist_\E)
\]
that agrees with $F$ on all nodes except $X$ and its outgoing arrows,
which are replaced by constant kernels using $\mu_X$.

The relationship between $F$ and $F_{do(X)}$
is expressed by a \emph{natural transformation}
\[
\eta^{(X)} : F \Longrightarrow F_{do(X)}.
\]
Each component $\eta^{(X)}_V : F(V)\to F_{do(X)}(V)$
represents the “surgical” replacement effect of the intervention
propagated through the functorial semantics.

\paragraph{Commuting diagram.}
\begin{center}
\begin{tikzcd}[column sep=huge, row sep=large]
  X \arrow[r, "f"] \arrow[d, "do(X)" left]
    & Y \arrow[r, "g"] \arrow[d, dashed, "\eta^{(X)}_Y" right]
    & Z \arrow[d, dashed, "\eta^{(X)}_Z" right] \\[2pt]
  F_{do(X)}(X)
    \arrow[r, "{F_{do(X)}(f)}"]
    & F_{do(X)}(Y)
    \arrow[r, "{F_{do(X)}(g)}"]
    & F_{do(X)}(Z)
\end{tikzcd}
\end{center}

\noindent
Commutativity expresses the functorial consistency condition:
\[
\eta^{(X)}_Z \circ F(g\!\circ\!f)
  = F_{do(X)}(g\!\circ\!f) \circ \eta^{(X)}_X,
\]
which ensures that causal dependencies
propagate coherently under intervention.

\paragraph{Example.}
Continuing the chain $X\!\to\!Y\!\to\!Z$,
let $F(f):X\to\Dist(Y)$ and $F(g):Y\to\Dist(Z)$ as before.
Then $F_{do(X)}$ is identical to $F$
except that $F_{do(X)}(f)$ is replaced by the constant kernel
\[
F_{do(X)}(f)(\ast) = \mu_X \in \Dist(X).
\]
The natural transformation component at $Y$
acts as
\[
\eta^{(X)}_Y
  = \Do_X(F(f);\mu_X)
  = \mu_Y \circ \Dist(F(f))
    \circ st \circ \langle \id, \mu_X\rangle.
\]
At the object level, $\eta^{(X)}_Y$
sends each original stochastic map to its
intervened counterpart,
and naturality guarantees that every downstream
composition (e.g.\ $g\circ f$) is updated coherently.

\paragraph{Interpretation.}
This functor–natural transformation perspective
unifies the ``mutilation'' operation of Pearl’s graph semantics
with the algebraic structure of the distribution monad.
Causal models become \emph{functors},
and interventions become \emph{natural transformations}
between them.
The naturality square expresses precisely the invariance of
downstream mechanisms under the intervention.

\begin{center}
\begin{tikzpicture}[node distance=25mm,>=Latex]
  \tikzstyle{functor}=[rectangle,draw,rounded corners,fill=blue!10,minimum width=15mm]
  \node[functor] (F) {$F$};
  \node[functor, right=of F] (Fdo) {$F_{do(X)}$};
  \draw[->, thick] (F) -- node[above]{\(\eta^{(X)}:F\!\Rightarrow\!F_{do(X)}\)} (Fdo);
\end{tikzpicture}
\end{center}

\paragraph{Standing assumptions.}
We work in a Grothendieck topos $\Topos=\Sh_{\J}(\C)$ whose site $(\C,\J)$
indexes regimes/experiments. Random variables are objects of a probability
sheaf $\Obs$.  An intervention is a compatible family of mechanism-replacement
maps \(M_U\to M_U^{\doop{x}}\) whose restriction squares commute.
Conditional
independence lives in the subobject classifier $\OmegaE$ via a
$\CI$-sheaf; $\J$-stability means local truth in Kripke--Joyal semantics.

\paragraph{Summary.}
\begin{itemize}
  \item $\mathcal{C}_G$ encodes causal structure (syntax).
  \item $\Kl(\Dist_\E)$ provides stochastic semantics.
  \item $F$ interprets the causal mechanisms.
  \item $do(X)$ induces a new functor $F_{do(X)}$.
  \item $\eta^{(X)}:F\Rightarrow F_{do(X)}$ expresses
        the intervention as a coherent natural transformation.
\end{itemize}

\section{Internal logic of toposes} 

We review the concept of internal logic that is intrinisc to every topos, including TCM's, and this logical language will play a central role in the extension of classical do-calculus to $j$-do-calculus. 

\subsection{Mitchell-B\'enabou Language}
\label{mbl}
We define the Mitchell-B\'enabou language (MBL), a typed local set theory (see Section~\ref{lst}) associated with a causal topos. Given the topos category ${\cal C}_\Omega$, we define the types of MBL as causal model objects $M$ of ${\cal C}_\Omega$. For each type $M$ (e.g., an SCM),  we assume the existence of variables $x_M, y_M, \ldots$, where each such variable has as its interpretation the identity arrow ${\bf 1}: M \rightarrow M$. We can construct product objects, such as $A \times B \times C$, where terms like $\sigma$ that define arrows are given the interpretation $\sigma: A \times B \times C \rightarrow D$. We can inductively define the terms and their interpretations in a topos category as follows (see \citep{maclane1992sheaves} for additional details): 
\begin{itemize}
    \item Each variable $x_M$ of type $M$ is a term of type $M$, and its interpretation is the identity $x_M = {\bf 1}: M \rightarrow M$ (e.g., $M$ may be an SCM or a causal model on a  Markov category). 

\item Terms $\sigma$ and $\tau$ of types $C$ and $D$ that are interpreted as $\sigma: A \rightarrow C$ and $\tau: B \rightarrow D$ can be combined to yield a term $\langle \sigma, \tau \rangle$ of type $C \times D$, whose joint interpretation is given as 

\[ \langle \sigma p, \tau q \rangle: X \rightarrow C \times D\]

where $X$ has the required projections $p: X \rightarrow A$  and $q: X \rightarrow B$. 

\item Terms $\sigma: A \rightarrow B$ and $\tau: C \rightarrow B$ of the same type $B$ yield a term $\sigma = \tau$ of type $\Omega$, interpreted as 

\[ (\sigma = \tau): W \xrightarrow[]{\langle \sigma p, \tau q \rangle} B \times B \xrightarrow[]{\delta_B} \Omega \]
where $\delta_B$ is the characteristic map of the diagonal functor $\Delta B \rightarrow B \times B$. These diagonal maps correspond to the ``copy" procedure in Markov categories \citep{Fritz_2020}. 

\item Arrows $f: A \rightarrow B$ and a term $\sigma: C \rightarrow A$ of type $A$ can be combined to yield a term $f \circ \sigma$ of type $B$, whose interpretation is naturally a composite arrow: 

\[ f \circ \sigma: C \xrightarrow[]{\sigma} A \xrightarrow[]{f} B\]

\item For exponential objects, terms $\theta: A \rightarrow B^C$ and $\sigma: D \rightarrow C$ of types $B^C$ and $C$, respectively, combine to give an ``evaluation" map of type $B$, defined as 

\[ \theta (\sigma): W \rightarrow B^C \times C \xrightarrow[]{e} B \]

where $e$ is the evaluation map, and $W$ defines a map $\langle \theta p, \sigma q \rangle$, where once again $p: W \rightarrow A$ and $q: W \rightarrow D$ are projection maps. 

\item Terms $\sigma: A \rightarrow B$ and $\tau: D \rightarrow \Omega^B$ combine to yield a term $\sigma \in \tau$ of type $\Omega$, with the following interpretation: 

\[ \sigma \in \tau: W \xrightarrow[]{\langle \sigma p, \tau q \rangle} B \times \Omega^B \xrightarrow[]{e} \Omega \]

\item Finally, we can define local functions as $\lambda$ objects, such as 

\[ \lambda x_C \sigma: A \rightarrow B^C \]

where $x_C$ is a variable of type $C$ and $\sigma: C \times A \rightarrow B$. 
\end{itemize}

We combine terms $\alpha, \beta$ etc. of type $\Omega$ using logical connectives $\wedge, \vee, \Rightarrow, \neg$, as well as quantifiers, to get composite terms, where each of the logical connectives is now defined over the subobject classifier $\Omega$. 


\begin{itemize}
    \item $\wedge: \Omega \times \Omega \rightarrow \Omega$ is interpreted as the {\em meet} operation in the partially ordered set of subobjects (given by the Heyting algebra). 

    \item $\vee: \Omega \times \Omega \rightarrow \Omega$ is interpreted as the {\em join} operation in the partially ordered set of subobjects (given by the Heyting algebra).  This operation gives the definition of a disjunction of two properties. 

    \item $\Rightarrow: \Omega \times \Omega \rightarrow \Omega$ is interpreted as an adjoint functor, as defined previously for a Heyting algebra. Thus, the property of implication over SCMs is modeled as an adjoint functor. 
    
\end{itemize}

We can combine these logical connectives with the term interpretation as arrows, relegating some details to \citep{maclane1992sheaves}. We now turn to the Kripke-Joyal semantics of this language. 

\subsection{Kripke-Joyal Semantics for a Causal Topos}
\label{kj}

 We now define the Kripke-Joyal semantics for the Mitchell-B\'enabou language of a causal topos. Any free variable $x$ must have some causal model $X$ of ${\cal C}_\Omega$ as its type.  For any causal model $M$ in ${\cal C}_\Omega$, define a {\em generalized element} as a morphism $\alpha: N \rightarrow M$. To understand this definition, note that we can define an element of a causal model by the  morphism $x: {\bf 1} \rightarrow M$. Thus, a generalized element $\alpha: N \rightarrow M$ represents the ``stage of definition" of $M$ by $N$. We specify the semantics of how an SCM $N$ supports any formula $\phi(\alpha)$, denoted by $N \Vdash \phi(\alpha)$, as follows: 
\[ N \Vdash \phi(\alpha) \ \ \ \mbox{if and only if } \ \ \ \mbox{Im} \ \alpha \leq \{ x | \phi(x) \} \]
Stated in the form of a commutative diagram, this ``forcing" relationship holds if and only if $\alpha$ factors through $\{x | \phi(x) \}$, where $x$ is a variable of type $M$ (recall that objects $M$ of a topos form its types), as shown in the following commutative diagram. \footnote{The concept of ``forcing" is generalized from set theory \citep{maclane1992sheaves}.}
 \begin{center}
\begin{tikzcd}
	&& {\{x | \phi(x) \}} && {{\bf 1}} \\
	\\
	N && M && {{\bf \Omega} }
	\arrow[from=1-3, to=1-5]
	\arrow[tail, from=1-3, to=3-3]
	\arrow["{{\bf True}}", from=1-5, to=3-5]
	\arrow[dashed, from=3-1, to=1-3]
	\arrow["\alpha"', from=3-1, to=3-3]
	\arrow["{\phi(x)}", from=3-3, to=3-5]
\end{tikzcd}
 \end{center}
This diagram provides an interesting way to define causal interventions in a causal topos, because it defines submodels of $M$.   Building on this definition, if $\alpha, \beta: N \rightarrow M$ are parallel arrows, we can give semantics to the formula $\alpha = \beta$ by the following statement: 

 \[ N \xrightarrow[]{\langle \alpha, \beta \rangle} M \times M \xrightarrow[]{\delta_M} \Omega\]

following the definitions in the previous section for the composite $\langle \alpha, \beta \rangle$ and $\delta_X$ in the Mitchell-B\'enabou language. We can extend the previous commutative diagram to show that $U \Vdash \alpha = \beta$ holds if and only if $\langle \alpha, \beta \rangle$ factors through the diagonal map $\Delta$: 

\begin{center}
\begin{tikzcd}
	&& M && {{\bf 1}} \\
	\\
	N && {M \times M} && {{\bf \Omega} }
	\arrow[from=1-3, to=1-5]
	\arrow["\Delta", tail, from=1-3, to=3-3]
	\arrow["{{\bf True}}", from=1-5, to=3-5]
	\arrow[dashed, from=3-1, to=1-3]
	\arrow["{\langle \alpha, \beta \rangle}"', from=3-1, to=3-3]
	\arrow["{\delta_M}", from=3-3, to=3-5]
\end{tikzcd}
\end{center}


\begin{itemize}
    \item {\bf Monotonicity:} If $U \Vdash \phi(x)$, then we can pullback the interpretation through any arrow $f: U' \rightarrow U$ in a topos ${\cal C}$ to obtain $U' \Vdash \phi(\alpha \circ f)$. 
    \begin{small}
\begin{tikzcd}
	&&&& {\{x | \phi(x) \}} && {{\bf 1}} \\
	\\
	{U'} && U && X && {{\bf \Omega} }
	\arrow[from=1-5, to=1-7]
	\arrow[tail, from=1-5, to=3-5]
	\arrow["{{\bf True}}", from=1-7, to=3-7]
	\arrow[dashed, from=3-1, to=1-5]
	\arrow["f", from=3-1, to=3-3]
	\arrow[dashed, from=3-3, to=1-5]
	\arrow["\alpha"', from=3-3, to=3-5]
	\arrow["{\phi(x)}", from=3-5, to=3-7]
\end{tikzcd}
\end{small}

    \item {\bf Local character:} Analogously, if $f: U' \rightarrow U$ is an epic arrow, then from $U' \Vdash \phi(\alpha \circ f)$, we can conclude $U \Vdash \phi(x)$. 
\end{itemize}


\begin{theorem}[Kripke--Joyal clauses on a site]
Let $\alpha\in M(U)$ be a generalized element at a stage $U$ of
$(\C,J)$.  For formulas $\phi(x)$ and $\psi(x)$:
\begin{enumerate}
  \item \(U\Vdash\phi(\alpha)\wedge\psi(\alpha)\) iff both conjuncts are
  forced at \(U\).
  \item \(U\Vdash\phi(\alpha)\vee\psi(\alpha)\) iff there is a
  \(J\)-cover \(\{u_i:U_i\to U\}\) such that, for every \(i\), either
  \(U_i\Vdash\phi(\alpha|_{U_i})\) or
  \(U_i\Vdash\psi(\alpha|_{U_i})\).
  \item \(U\Vdash\phi(\alpha)\Rightarrow\psi(\alpha)\) iff for every
  \(u:V\to U\), forcing \(\phi(\alpha|_V)\) at \(V\) implies forcing
  \(\psi(\alpha|_V)\) at \(V\).
  \item \(U\Vdash\neg\phi(\alpha)\) iff for every \(u:V\to U\),
  \(V\Vdash\phi(\alpha|_V)\) implies \(V\Vdash\bot\).
  \item \(U\Vdash\exists y\,\phi(\alpha,y)\) iff there is a \(J\)-cover
  \(\{u_i:U_i\to U\}\) and local witnesses \(\beta_i\in Y(U_i)\) such
  that \(U_i\Vdash\phi(\alpha|_{U_i},\beta_i)\) for every \(i\).
  \item \(U\Vdash\forall y\,\phi(\alpha,y)\) iff for every \(u:V\to U\)
  and every \(\beta\in Y(V)\), one has
  \(V\Vdash\phi(\alpha|_V,\beta)\).
\end{enumerate}
\end{theorem}

\begin{proof}
These are the standard Kripke--Joyal clauses for the internal language of
\(\Sh_J(\C)\); see \citet{maclane1992sheaves}.
\end{proof}


\subsection{Local Set Theory}
\label{lst}

The Mitchell-B\'enabou language is an example of a ``local set theory" \cite{bell}.  A {\em local set theory} \citep{bell} is defined as a language ${\cal L}$ specified by the following classes of symbols: 

\begin{enumerate}
    \item Symbols ${\bf 1} $ and $\Omega$ representing the {\em unity} type and {\em truth-value} type symbols. 

    \item A collection of symbols ${\bf A}, {\bf B}, {\bf C}, \ldots $ called {\em ground type symbols}. 

    \item A collection of symbols ${\bf f}, {\bf g}, {\bf h}, \ldots$ called {\em function} symbols. 
\end{enumerate}

 We will use the topos-theoretical constructions to construct composite types. We can use an inductive procedure to recursively construct {\bf type symbols} of ${\cal L}$ as follows: 

\begin{enumerate}
    \item  Symbols ${\bf 1} $ and $\Omega$ are type symbols (the terminal object and the subobject classifier in a causal topos). 

    \item Any ground type symbol is a type symbol. For a causal topos, each SCM is a ground type symbol. 

    \item If ${\bf A}_1, \ldots, {\bf A}_n$ are type symbols, so is their product ${\bf A}_1 \times \ldots {\bf A}_n$, where for $n=0$, the type of $\prod_{i=1}^n {\bf A}_i$ is ${\bf 1}$. The product ${\bf A}_1 \times \ldots {\bf A}_n$ has the {\em product type} symbol.  These constructs allow defining an algebra of causal models. 

     \item If ${\bf A}$ is a type symbol, so is ${\bf PA}$. The type ${\bf PA}$ is called the {\em power} type. \footnote{Note that in a topos, these will be interpreted as {\em power objects}, generalizing the concept of power sets.}  We thus can give meaning to concept of a ``powerset" of a causal model, where we interpret the subobject classifier as defining the abstract semantics of a powerset for each SCM. 
\end{enumerate}

Thus, a product of SCMs will define product types. Given an SCM $M$, we can define its power type as well, which is an abstract notion of the ``power set" of a causal model (if you interpret this in the context of subobject classifiers, it means that we are defining a family of submodels). 
For each type symbol ${\bf A}$, the language ${\cal L}$ contains a set of {\em variables} $x_{\bf A}, y_{\bf A}, z_{\bf A}, \ldots$. In addition, ${\cal L}$ contains the distinguished ${\bf *}$ symbol. Each function symbol in ${\cal L}$  is assigned a {\em signature} of the form ${\bf A} \rightarrow {\bf B}$. \footnote{In a topos, these will correspond to arrows of the category.} We can define the {\em terms} of the local set theory language ${\cal L}$ recursively as follows: 

\begin{itemize}
    \item ${\bf *}$ is a term of type ${\bf 1}$. 

    \item for each type symbol ${\bf A}$, variables $x_{\bf A}, y_{\bf A}, \ldots$ are terms of type ${\bf A}$. 

    \item if ${\bf f}$ is a function symbol with signature ${\bf A} \rightarrow {\bf B}$, and $\tau$ is a term of type ${\bf A}$, then ${\bf f}(\tau)$ is a term of type ${\bf B}$. 

    \item If $\tau_1, \ldots, \tau_n$ are terms of types ${\bf A}_1, \ldots, {\bf A}_n$, then $\langle \tau_1, \ldots \tau_n \rangle$ is a term of type ${\bf A}_1 \times \ldots {\bf A}_n$, where if $n=0$, then $\langle \tau_1, \ldots \tau_n \rangle$ is of type ${\bf *}$. 

    \item If $\tau$ is a term of type ${\bf A}_1 \times {\bf A}_n$, then for $1 \leq i \leq n$, $(\tau)_i$ is a term of type ${\bf A}_i$. 

    \item if $\alpha$ is a term of type $\Omega$, and $x_{\bf A}$ is a variable of type ${\bf A}$, then $\{x_{\bf A} : \alpha \}$ is a term of type ${\bf PA}$.  

    \item if $\sigma, \tau$ are terms of the same type, $\sigma = \tau$ is a term of type $\Omega$. 

    \item if $\sigma, \tau$ are terms of the types ${\bf A}, {\bf PA}$, respectively, then $\sigma \in \tau$ is a term of type ${\bf \Omega}$. 
\end{itemize}

A term of type ${\bf \Omega}$ is called a {\em formula}. The language ${\cal L}$ does not yet have defined any logical operations, because in a typed language, logical operations can be defined in terms of the types, as illustrated below. 

\begin{itemize}
    \item $\alpha \Leftrightarrow \beta$ is interpreted as $\alpha = \beta$. 

    \item {\bf true} is interpreted as ${\bf *} = {\bf *}$. 

    \item $\alpha \wedge \beta$ is interpreted as $\langle \alpha, \beta \rangle = \langle {\bf true}, {\bf false} \rangle$. 

    \item $\alpha \Rightarrow \beta$ is interpreted as $(\alpha \wedge \beta) \Leftrightarrow \alpha$

    \item $\forall x \ \alpha$ is interpreted as $\{x : \alpha\} = \{x : {\bf true} \}$

    \item ${\bf false}$ is interpreted as $\forall \omega \ \omega$. 

    \item $\neg \alpha$ is interpreted as $\alpha \Rightarrow {\bf false}$.

    \item $\alpha \vee \beta$ is interpreted as $\forall \omega \ [(\alpha \Rightarrow \omega \wedge \beta \Rightarrow \omega) \Rightarrow \omega]$

    \item $\exists x \ \alpha$ is interpreted as $\forall \omega [ \forall x (\alpha \Rightarrow \omega) \Rightarrow \omega ]$

\end{itemize}

Finally, we have to specify the inference rules, which are given in the form of {\em sequents}. We will just sketch out a few, and the rest can be seen in \citep{bell}. A sequent is a formula $\Gamma: \alpha$ where $\alpha$ is a formula, and $\Gamma$ is a possibly empty finite set of formulae. The basic axioms include $\alpha: \alpha$ (tautology), $:x_1 = {\bf *}$ (unity), a rule for forming projections of products, a rule for equality, and another for comprehension. Finally, the inference rules are given in the form: 

\begin{itemize}
    \item {\em Thinning:}
    \[
  \begin{prooftree}
    \hypo{\Gamma : \alpha}
    \infer1{\beta, \Gamma: \alpha}
  \end{prooftree}
\]
\item {\em Cut}: 

\[
  \begin{prooftree}
    \hypo{\Gamma : \alpha, \  \ \alpha, \Gamma: \beta}
    \infer1{\Gamma: \beta}
  \end{prooftree}
\]

\item {\em Equivalence}: 

\[
  \begin{prooftree}
    \hypo{\alpha, \Gamma : \beta \ \ \beta, \Gamma: \alpha}
    \infer1{\Gamma: \alpha \Leftrightarrow \beta}
  \end{prooftree} 
\]

\end{itemize}

A full list of inference rules with examples of proofs is given in \citep{bell}. Now that we have the elements of a local set theory defined as shown above, we need to connect its definitions with a causal topos. That is the topic of the next section.

\section{Kripke--Joyal Forcing and Internal Semantics in a TCM}
\label{subsec:kripke-joyal}

Before explaining how to prove the validity of do-calculus statements in a TCM, we need to provide more detail on the Kripke-Joyal intuitionistic semantics that constitutes the ``semantic engine" in a topos. 

\paragraph{Motivation.}
In a classical set-based semantics, a statement such as
``$\forall x\in X,\, \varphi(x)$''
is true if $\varphi(x)$ holds for each element $x\in X$.
In a topos, however, \emph{elements} are generalized---they are morphisms
$\alpha:N\to X$ from a test object~$N$.
The Kripke--Joyal forcing relation provides the inductive definition of what it means
for a formula $\varphi$ to hold at a stage $N$ and under a generalized element
$\alpha:N\to\Gamma$.
This endows the internal logic of a topos with semantics analogous
to intuitionistic Kripke models.

\paragraph{Definition (Forcing relation).}
For an elementary topos~$\E$, the \emph{forcing relation}
\[
N \Vdash_\E \varphi[\alpha]
\]
reads ``$\varphi$ holds at stage $N$ under assignment $\alpha:N\to\Gamma$.''
The semantics is defined inductively:
\begin{align*}
N &\Vdash (x=y)[\alpha]
    &&\text{iff } \alpha^*x=\alpha^*y \text{ in }\E,\\[2pt]
N &\Vdash (\varphi\wedge\psi)[\alpha]
    &&\text{iff } N\Vdash\varphi[\alpha]\text{ and }N\Vdash\psi[\alpha],\\[2pt]
N &\Vdash (\varphi\Rightarrow\psi)[\alpha]
    &&\text{iff for all }u:N'\to N,\;
        N'\Vdash\varphi[\alpha\!\circ\!u]\Rightarrow
        N'\Vdash\psi[\alpha\!\circ\!u],\\[2pt]
N &\Vdash (\exists x:A)\,\varphi(x)[\alpha]
    &&\text{iff there exists an epi }e:M\to N
        \text{ and }a:M\to A\text{ s.t. }M\Vdash\varphi[a,\alpha\!\circ\!e],\\[2pt]
N &\Vdash (\forall x:A)\,\varphi(x)[\alpha]
    &&\text{iff for all }u:N'\to N,\text{ and all }a:N'\to A,\\[-2pt]
    &&&\hspace{3cm}N'\Vdash\varphi[a,\alpha\!\circ\!u].
\end{align*}
Truth is thus \emph{monotone}:
if $N\Vdash\varphi[\alpha]$ and $u:N'\to N$, then $N'\Vdash\varphi[\alpha\!\circ\!u]$.
This matches intuitionistic semantics where information grows along morphisms.

\paragraph{Interpretation in TCMs.}
In a Topos Causal Model, contexts $\Gamma$ denote joint variable spaces,
and a stage $\alpha:N\to\Gamma$ represents a local assignment or partial observation.
A judgment
\[
\Gamma \vdash \varphi
\]
is true in~$\E$ if $N\Vdash_\E\varphi[\alpha]$ for all $\alpha:N\to\Gamma$.
For instance, the sequent
\[
\Gamma \vdash P(Y\mid do(Z),X)=P(Y\mid X)
\]
is internally true iff for every $\alpha:N\to\Gamma$,
the two stochastic morphisms $N\to\Dist(Y)$
given by $P(Y\mid do(Z),X)\circ\alpha$ and $P(Y\mid X)\circ\alpha$
coincide in $\Kl(\Dist_\E)$.
Hence equality of arrows is verified \emph{stagewise}.

\paragraph{Forcing for interventions.}
For $k:\Gamma\times Z\to\Dist_\E Y$ and policy $\mu:\Gamma\to\Dist_\E Z$,
the formula
\[
\Gamma\vdash \Do_Z(k;\mu)=k_0
\]
is true in~$\E$ iff
for every stage $\alpha:N\to\Gamma$,
\[
\int_Z k(\alpha,z)\,d\mu(\alpha)(z)=k_0(\alpha)
\quad\text{in }\Dist(Y).
\]
Thus the ``proof'' of a do-calculus identity
reduces to verifying equality of integrals pointwise
at every generalized element~$\alpha$.
This is precisely how the proof of Rule~1 was phrased.

\paragraph{Semantic intuition.}
Each object $N$ of~$\E$ is a ``stage of information,'' and morphisms
$u:N'\to N$ represent refinements of context.
Kripke--Joyal forcing guarantees that
truth of causal assertions (e.g.\ conditional independence, intervention equalities)
is preserved by refinement:
if a relation holds at a stage, it holds at all more informative stages.
Hence, causal reasoning in a TCM corresponds to
\emph{constructive reasoning about local data} that remains stable
under restriction.

\paragraph{Summary.}
Kripke--Joyal forcing provides the bridge between
syntactic sequents of internal logic and semantic equality of morphisms.
In particular:
\begin{itemize}
  \item generalized elements $\alpha:N\to X$ replace concrete elements of $X$;
  \item truth values are subobjects (elements of the internal Heyting algebra);
  \item forcing ensures stagewise stability of causal equations.
\end{itemize}
Therefore, to prove a statement like
$(Y\!\perp\!Z\mid\Gamma)\Rightarrow
  (P(Y\mid do(Z),\Gamma)=P(Y\mid\Gamma))$
in a TCM,
one checks the equality at each stage $N$---exactly the Kripke--Joyal semantics
of internal equality.


\begin{figure}[t]
  \centering
  \begin{minipage}{0.47\textwidth}
    \centering
    \begin{tikzcd}[row sep=large, column sep=huge]
      N' \arrow[d, "u"'] \arrow[r, "\alpha\circ u"] & \Gamma \\
      N \arrow[ur, "\alpha"'] &
    \end{tikzcd}

    \vspace{0.4em}
    \small
    (a) \textbf{Kripke--Joyal monotonicity.} If \(N \Vdash \varphi[\alpha]\) and \(u:N'\to N\),
    then \(N' \Vdash \varphi[\alpha\circ u]\).
    Intuitively, truth is preserved under refinement \(u\).
  \end{minipage}
  \hfill
  \begin{minipage}{0.47\textwidth}
    \centering
    \begin{tikzcd}[row sep=large, column sep=huge]
      N' \arrow[d, "u"'] \arrow[r, bend left=15, "f\circ\alpha\circ u"] \arrow[r, bend right=15, swap, "g\circ\alpha\circ u"] & \Dist(Y) \\
      N \arrow[r, bend left=15, "f\circ\alpha"] \arrow[r, bend right=15, swap, "g\circ\alpha"] & \Dist(Y) \arrow[u, equal, phantom, pos=0.0, ""]
    \end{tikzcd}

    \vspace{0.4em}
    \small
    (b) \textbf{Stagewise equality.} To show \(\Gamma \vdash f=g\) (with \(f,g:\Gamma\to \Dist(Y)\)),
    check \(f\circ\alpha=g\circ\alpha\) for every \(\alpha:N\to\Gamma\).
    Then for any refinement \(u:N'\to N\), also \(f\circ\alpha\circ u=g\circ\alpha\circ u\).
  \end{minipage}

  \caption{Kripke--Joyal forcing diagrams. (a) Truth of a formula \(\varphi\) at a stage
  \(N\) is preserved along refinements \(u:N'\to N\). (b) Internal equalities (e.g.\ the
  do-calculus identity \(\Do_Z(k;\mu)=k_0\)) are verified as equality of morphisms at every
  stage and remain equal under refinement.}
  \label{fig:kj-forcing-diagram}
\end{figure}

\section{From d-separation to $j$-stability: do-calculus on sites} 


\paragraph{Setup.}
Let $(\C,J)$ be a site of regimes/contexts (objects $U\!\in\!\C$, arrows are restrictions),
and let $\V$ be a fixed finite set of variables.
Consider a presheaf of DAGs
\[
  \G:\C^{\mathrm{op}}\to \mathbf{DAG}_{\V},\qquad
  U \longmapsto G_U=(\V,E_U),
\]
with restriction maps $\G(f):G_U\to G_V$ for each $f:V\to U$.  We assume
that graph surgery commutes with restriction:
\[
  (G_U{}_{\bar X})|_V=G_V{}_{\bar X},\qquad
  (G_U{}_{\underline Z})|_V=G_V{}_{\underline Z},
\]
and similarly for the surgery used in Rule~3.

At each stage $U$ let $M_U$ be an acyclic SCM with graph $G_U$ and
finite-state stochastic kernels.  For every $f:V\to U$, restriction is
assumed to preserve the variable signature and the following operations:
\begin{enumerate}
  \item mechanism replacement defining $\doop{x}$;
  \item the mutilated models used by Pearl's three rules;
  \item every conditional kernel that occurs below; and
  \item equality of kernels on overlaps.
\end{enumerate}
Conditionals are used only on events of positive normalizing mass.  More
general state spaces may be substituted provided compatible disintegrations
are supplied.  These assumptions are what make the chartwise kernels into a
separated sheaf; Kripke--Joyal logic alone does not manufacture
conditionals or interventions.
For $U\!\in\!\C$ and $X,Y,Z\subseteq \V$, write
$\mathrm{dsep}_{G_U}(X;Y\given Z)$ for the usual Pearl d-separation in $G_U$
(\emph{every} undirected path from $X$ to $Y$ is blocked by $Z$ using the
standard non-collider/collider rules).

Define the fiberwise predicate
\[
  \varphi_{X\,\indep\, Y \,\given\, Z}(U)\;:\Longleftrightarrow\;
  \mathrm{dsep}_{G_U}(X;Y\given Z).
\]

We use the external forcing notation
\[
  U \Vdash_J \psi
  \;\;\Longleftrightarrow\;\;
  \exists\text{ a $J$-covering sieve }S\text{ on }U
  \text{ such that } \forall(f:V\to U)\in S,\; V \Vdash \psi.
\]
Equivalently, in the ambient presheaf topos this is the forcing clause for
the Lawvere--Tierney closure $j(\psi)$.  We do not apply a second modality
after passing to $\Sh_J(\C)$.

\begin{definition}[{\boldmath $j$-d-separation / $j$-stable CI}]
For $U\!\in\!\C$, we say $X$ is $j$-d-separated from $Y$ by $Z$ at $U$
(written $X \perp^{\,j}_U Y \given Z$) iff
\[
  U \Vdash_J \varphi_{X\,\indep\, Y \,\given\, Z}.
\]
Equivalently: there exists a $J$-cover $S$ of $U$ such that
$\mathrm{dsep}_{G_V}(X;Y\given Z)$ holds for every $f:V\to U$ in $S$.
\end{definition}

\begin{definition}[{\boldmath $j$-closed path (path-wise view)}]
A family of paths from $X$ to $Y$ is \emph{$j$-closed by $Z$ at $U$}
when there is one common $J$-cover on which every restricted path is blocked
by $Z$ using the usual collider/non-collider clauses.  Thus
$X \perp^{\,j}_U Y\given Z$ precisely when the family of all such paths is
$j$-closed.  Requiring one common cover avoids changing the quantifier order.
\end{definition}

\begin{proposition}[Conservativity]
If $J$ is the trivial topology (only the maximal sieve covers),
then for all $U\!\in\!\C$ and $X,Y,Z\subseteq\V$,
\[
  X \perp^{\,j}_U Y \given Z
  \;\Longleftrightarrow\;
  \mathrm{dsep}_{G_U}(X;Y\given Z).
\]
\end{proposition}

\begin{remark}[Trivial topology is not Booleanity]
The proposition is stagewise.  In general
\(\Sh(\C,J_{\mathrm{triv}})\simeq[\C^{\mathrm{op}},\Set]\), whose internal
logic need not be Boolean.  Ordinary set-based causal semantics is recovered
from the terminal site \(\C=\mathbf 1\), for which
\(\Sh(\mathbf 1)\simeq\Set\).
\end{remark}

\begin{proposition}[Heredity (stability under restriction)]
If $X \perp^{\,j}_U Y \given Z$ and $g:W\to U$ in $\C$,
then $X \perp^{\,j}_W Y \given Z$.
\end{proposition}

\begin{remark}[No monotonicity in the conditioning set]
Neither ordinary nor $j$-local d-separation is monotone under enlarging the
conditioning set: conditioning on a collider, or one of its descendants, can
open a previously blocked path.  Consequently no such monotonicity property
is used below.
\end{remark}

\begin{theorem}[Local-to-global soundness for compatible causal models]
\label{thm:local-global-soundness}
Suppose $\mathsf P$ is an internal stochastic model (a compatible family $\{P_U\}$)
arising from the compatible SCMs above and is \emph{fiberwise
global-Markov} to $\G$.  Then
\[
  X \perp^{\,j}_U Y \given Z
  \;\Longrightarrow\;
  U \Vdash_J \bigl(X \indep_{\mathsf P} Y \given Z\bigr).
\]
More generally, let $L$ and $R$ be two compatible sections of the sheaf of
causal kernels over $U$.  If $L|_{V_i}=R|_{V_i}$ on every member of a
$J$-cover $\{V_i\to U\}$, then $L=R$ at $U$.
\end{theorem}

\begin{proof}
If $X \perp^{\,j}_U Y \given Z$, choose a $J$-cover $S$ with
$\mathrm{dsep}_{G_V}(X;Y\given Z)$ for all $f:V\to U$ in $S$.
By the fiberwise global-Markov property, $X \indep_{P_V} Y \given Z$ holds for each such $V$.
This is precisely the Kripke--Joyal local clause.  For the second statement,
the restrictions of $L$ and $R$ are matching families with identical local
components.  Uniqueness in the sheaf axiom (equivalently, separatedness)
implies $L=R$.
\end{proof}

\begin{definition}[{\boldmath $j$-faithfulness / $j$-perfect map}]
We say $\mathsf P$ is \emph{$j$-faithful} to $\G$ at $U$ iff
\[
  U \Vdash_J \bigl( X \indep_{\mathsf P} Y \given Z \bigr)
  \quad\Longleftrightarrow\quad
  X \perp^{\,j}_U Y \given Z
  \qquad \text{for all } X,Y,Z\subseteq\V.
\]
\end{definition}

\begin{corollary}[Stagewise and classical reductions]
When $J$ is trivial, $j$-d-separation reduces stagewise to ordinary
d-separation.  When in addition $\C=\mathbf 1$, the definitions and the
fiberwise Markov assumption reduce to their usual set-based forms.
\end{corollary}

\paragraph{Practical reading.}
In applications where $J$ encodes “admissible regimes,”
$X \perp^{\,j}_U Y \given Z$ means:
\emph{there exists a covering family of regimes refining $U$ on which the usual
d-separation (collider/non-collider) checks all pass.}
This is the precise way \(j\) “wraps” the classical rules without changing them.


\providecommand{\indep}{\mathrel{\perp\!\!\!\perp}}
\providecommand{\given}{\,\mid\,}

\newcommand{\GbarX}{\G^{\overline{X}}}          
\newcommand{\GunderZ}{\G^{\underline{Z}}}       
\newcommand{\GbarZofW}{\G^{\overline{Z(W)}}}    

\newcommand{\PU}{P_U}

\begin{theorem}[{\boldmath \(j\)-Rule 1: insertion/deletion of observations}]
\label{thm:J1}
Fix \(U\in\C\). Let \(\GbarX\) denote the presheaf obtained by deleting all incoming arrows to \(X\).
If
\[
  U \Vdash_J\mathrm{dsep}_{G^{\overline{X}}} (Y;Z \given X,W),
\]
then
\[
  \PU\bigl(y \,\given\, \doop{x},z,w\bigr) \;=\; \PU\bigl(y \,\given\, \doop{x},w\bigr).
\]
\end{theorem}

\begin{proof}
Apply Pearl's Rule~1 on every chart of the witnessing cover.  Surgery and
conditioning commute with restriction by assumption, so both sides form
compatible sections.  The equality descends by
Theorem~\ref{thm:local-global-soundness}.
\end{proof}

\begin{theorem}[{\boldmath \(j\)-Rule 2: action/observation exchange}]
\label{thm:J2}
Let \(\G^{\overline{X},\underline{Z}}\) delete incoming to \(X\) and outgoing from \(Z\).
If
\[
  U \Vdash_J\mathrm{dsep}_{G^{\overline{X},\underline{Z}}} (Y;Z \given X,W),
\]
then
\[
  \PU\bigl(y \,\given\, \doop{x},\doop{z},w\bigr) \;=\; \PU\bigl(y \,\given\, \doop{x},z,w\bigr).
\]
\end{theorem}

\begin{proof}
Apply Pearl's Rule~2 chartwise and descend the resulting compatible kernel
equality as in the preceding proof.
\end{proof}

\begin{theorem}[{\boldmath \(j\)-Rule 3: insertion/deletion of actions}]
\label{thm:J3}
Let \(Z(W)\subseteq Z\) be those \(z\in Z\) that are \emph{not} ancestors of any node in \(W\) in \(G^{\overline X}\),
and let \(\GbarZofW\) delete incoming arrows to the nodes in \(Z(W)\) (in \(G^{\overline X}\)).
If
\[
  U \Vdash_J\mathrm{dsep}_{G^{\overline{X},\overline{Z(W)}}}(Y;Z \given X,W),
\]
then
\[
  \PU\bigl(y \,\given\, \doop{x},\doop{z},w\bigr) \;=\; \PU\bigl(y \,\given\, \doop{x},w\bigr).
\]
\end{theorem}

\begin{proof}
Apply Pearl's Rule~3 chartwise.  Compatibility on overlaps and
separatedness give the asserted equality at \(U\).
\end{proof}

\paragraph{Reading.}
Replace each classical d-separation premise by its \(j\)-d-separation version evaluated by the Kripke–Joyal clause:
\(\;U \Vdash_J(\cdot)\;\) means the premise holds on some \(J\)-cover of \(U\).
The causal rule remains Pearl's rule; the new theorem is that its compatible
chartwise conclusions descend.

To illustrate the above abstract definitions, we now present three examples of $j$-stability. To reiterate the definition of cover, the main definition that is useful to remember is the following: 

\begin{definition}
    \emph{Cover schema for $j$-separation.} A family
    $\{S_i\to U\}$ witnesses $X\indep^j_UY\mid Z$ only when it is
    $J$-covering and, on \emph{every} chart $S_i$, every restricted path
    from $X$ to $Y$ is blocked by $Z$ according to ordinary d-separation.
    Different paths may be blocked for different reasons, but checking
    different path fragments on unrelated charts is insufficient.
\end{definition}

\paragraph*{Cover used in Figure~\ref{fig:j-closure-schematic}}
Fix the ambient site $U$ that contains the whole path from $X$ to $Y$.
Let $\mathrm{Col}$ be the set of collider vertices on the path and
let $Z$ be the intended conditioning set.  
One possible witnessing cover $S=\{S_Z\}\cup\{S_v : v\in\mathrm{Col}\}$
must satisfy:
\begin{itemize}\setlength{\itemsep}{2pt}
\item[\(\triangleright\)] \emph{Observational chart:}
  $S_Z\to U$ is a chart on which the variables in $Z$ are measurable and
  admissible and the entire restricted path is blocked.
\item[\(\triangleright\)] \emph{Collider charts:}
  for each collider $v\in\mathrm{Col}$ we include a chart $S_v\to U$ that locally closes the collider.  Concretely, either  
  (i) $S_v$ forbids conditioning on $v$ or any descendant of $v$ (so the collider remains closed), or  
  (ii) $S_v$ is an interventional chart (e.g. cutting incoming arrows to $v$) that breaks the collider backdoor.  
  In both cases the entire restricted path, together with every alternative
  path from $X$ to $Y$, must be blocked in $S_v$.  Only then does
  Kripke--Joyal forcing yield $X \indep^j_U Y \mid Z$.
\end{itemize}

\paragraph*{Cover used in Figure~\ref{fig:collider-descendant-not-jstable}}
Here we intentionally take the \emph{trivial} cover $\{U\to U\}$,
where $U$ permits conditioning on the descendant $D$ of the collider.  
In $U$ the collider is opened by $D$, so $X \nindep Y \mid D$ holds in $U$ and there is no $j$-stability witness.  This panel illustrates the “problematic” ambient view before refinement.

\paragraph*{Cover used in Figure~\ref{fig:jclosure-restores-blocking}}
We refine $U$ by a two-chart $J$-cover $S=\{S_{\mathrm{int}},S_{\mathrm{obs}}\}$:
\begin{itemize}\setlength{\itemsep}{2pt}
\item[\(\triangleright\)] \emph{Interventional chart} $S_{\mathrm{int}}\hookrightarrow U$:
  disable the link $C\to D$ (e.g.\ an intervention on $C$ or $D$ that cuts $C\to D$).  
  Then conditioning on $D$ \emph{does not} open the collider, so $X \indep Y \mid D$ holds in $S_{\mathrm{int}}$.
\item[\(\triangleright\)] \emph{Observational chart} $S_{\mathrm{obs}}\hookrightarrow U$:
  the $\sigma$–algebra admits $D$ but forbids conditioning on $D$ or any descendant of the collider $C$; the collider remains closed, so $X \indep Y$ holds regardless of $D$.
\end{itemize}
These charts jointly cover the ambiguous situation in $U$, so by $J$-closure the sequent
$X \indep^j_U Y \mid D$ is forced.  Intuitively: we glue two legitimate ways of blocking
(the interventional “cut” and the observational “don’t condition on descendants”),
and $J$ authorizes this family as covering.


\tikzset{
  >={Stealth[length=1.8mm]},
  vnode/.style = {circle, draw, inner sep=1.4pt, minimum size=12pt},
  xnode/.style = {vnode, fill=blue!18},
  ynode/.style = {vnode, fill=red!18},
  znode/.style = {vnode, fill=green!20},
  colnode/.style = {diamond, aspect=2, draw, inner sep=1pt, fill=orange!25},
  panel/.style = {rounded corners=2pt, dashed, draw=gray!65, line width=0.7pt}
}

\begin{figure}[t]
\centering
\begin{tikzpicture}[x=2.8cm,y=2.5cm]  

  \node[panel, minimum width=6.6cm, minimum height=3.6cm,
        label={[gray!60]above:$V_1$}] (P1) at (-6.0,0) {};
  \node[panel, minimum width=7.2cm, minimum height=4.0cm,
        label={[gray!60]above:$U$}]  (PU) at ( -1.0,0) {};
  \node[panel, minimum width=6.6cm, minimum height=3.6cm,
        label={[gray!60]above:$V_2$}] (P2) at ( 6.0,0) {};

  \node[xnode]   (XU) at (-2.2, 0.25) {$X$};
  \node[znode]   (ZU) at (-0.7, 0.25) {$Z$};
  \node[ynode]   (YU) at ( 2.2, 0.25) {$Y$};
  \node[colnode] (CU) at ( 0.7,-0.80) {};
  \draw[->,thick] (XU) -- (ZU);
  \draw[->,thick] (ZU) -- (YU);
  \draw[->,thick] (XU) -- (CU);
  \draw[->,thick] (YU) -- (CU);
  \node[below=2pt of CU, black!65] {\scriptsize collider};
  \node[black!60] at (0,1.2) {\scriptsize $Z$ in conditioning set?};

  \node[xnode]   (X1) at (-8.2, 0.25) {$X$};
  \node[znode]   (Z1) at (-6.7, 0.25) {$Z$};
  \node[ynode]   (Y1) at (-5.2, 0.25) {$Y$};
  \node[colnode] (C1) at (-6.7,-0.90) {};
  \draw[->,thick] (X1) -- (Z1);
  \draw[->,thick] (Z1) -- (Y1);
  \draw[->,thick] (X1) -- (C1);
  \draw[->,thick] (Y1) -- (C1);
  \draw[decorate,decoration={brace,amplitude=5pt},gray!65]
    ($(X1)+(0,0.55)$) -- node[above,yshift=4pt,black!60]{\scriptsize blocked by $Z$}
    ($(Y1)+(0,0.55)$);

  \node[xnode]   (X2) at ( 4.2, 0.25) {$X$};
  \node[vnode]   (M2) at ( 5.7, 0.25) {$M$};
  \node[ynode]   (Y2) at ( 7.2, 0.25) {$Y$};
  \node[colnode] (C2) at ( 5.7,-0.90) {};
  \draw[->,thick] (X2) -- (M2);
  \draw[->,thick] (M2) -- (Y2);
  \draw[->,thick] (X2) -- (C2);
  \draw[->,thick] (Y2) -- (C2);
  \node[below=2pt of C2, black!65] {\scriptsize collider closed};

\draw[->, very thick, black!40]
  (P1.east) to[bend right=-12]
  node[midway, below=5pt, fill=white, inner sep=1pt]{\scriptsize $f_1:V_1\!\to\!U$}
  (PU.west);

\draw[->, very thick, black!40]
  (P2.west) to[bend left=-12]
  node[midway, below=5pt, fill=white, inner sep=1pt]{\scriptsize $f_2:V_2\!\to\!U$}
  (PU.east);

\end{tikzpicture}
\caption{\textbf{$j$-closure (schematic).} A path ambiguous in $U$ becomes
blocked on a $J$-cover $S$, hence $X \perp^{\,j}_{U} Y \mid Z$.}
\label{fig:j-closure-schematic}
\end{figure}
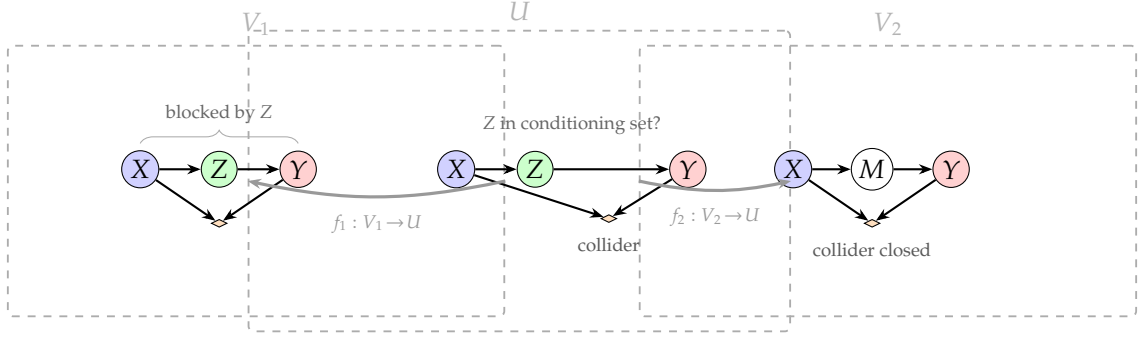

\tikzset{
  >={Stealth[length=1.8mm]},
  vnode/.style = {circle, draw, inner sep=1.4pt, minimum size=12pt},
  xnode/.style = {vnode, fill=blue!18},
  ynode/.style = {vnode, fill=red!18},
  znode/.style = {vnode, fill=green!20},
  hnode/.style = {vnode, fill=gray!20},
  colnode/.style = {diamond, aspect=2, draw, inner sep=1pt, fill=orange!25},
  panel/.style = {rounded corners=2pt, dashed, draw=gray!65, line width=0.6pt}
}

\begin{figure}[t]
\centering
\begin{tikzpicture}[x=3.5cm,y=3.5cm]
  \node[panel, minimum width=8.5cm, minimum height=4.2cm,
        label={[black]above:$U$}] (PU) at (0,0) {};

  \node[xnode]   (X)  at (-2.6,0.2) {$X$};
  \node[ynode]   (Y)  at ( 2.6,0.2) {$Y$};
  \node[colnode] (C)  at ( 0.0,0.2) {};
  \node[znode]   (D)  at ( 0.0,-1.2) {$D$};

  \draw[->,thick] (X) -- (C);
  \draw[->,thick] (Y) -- (C);
  \draw[->,thick] (C) -- (D);

  \node[black] at (0,-1.8) {\small conditioning on descendant $D$};
  \draw[-, ultra thick] ($(X)+(-0.1,0.6)$) -- node[above] {\small $X \nindep Y \given D$ in $U$} ($(Y)+(0.1,0.6)$);
\end{tikzpicture}
\caption{\textbf{Collider opened by conditioning on a descendant.}
In the ambient site \(U\), conditioning on a descendant \(D\) opens the collider, so \(X \nindep Y \given D\). This configuration is \emph{not} \(j\)-stable unless the \(J\)-cover removes (or closes) the path.}
\label{fig:collider-descendant-not-jstable}
\end{figure}
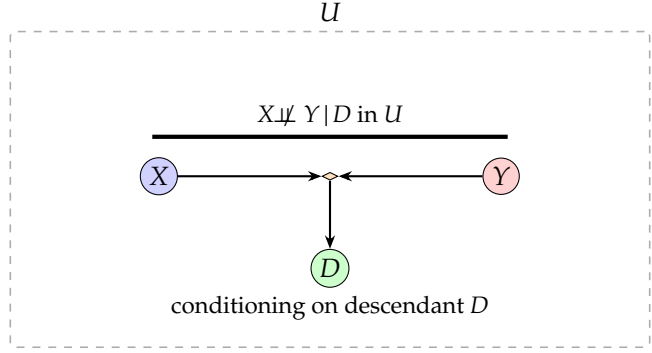

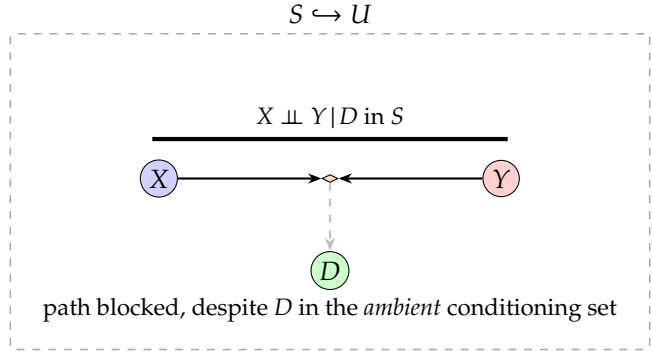
\begin{figure}[t]
\centering
\begin{tikzpicture}[x=3.5cm,y=3.5cm]
  \node[panel, minimum width=8.5cm, minimum height=4.2cm,
        label={[black]above:$S \hookrightarrow U$}] (PS) at (0,0) {};

  \node[xnode]   (X)  at (-2.6,0.2) {$X$};
  \node[ynode]   (Y)  at ( 2.6,0.2) {$Y$};
  \node[colnode] (C)  at ( 0.0,0.2) {};
  \node[znode]   (D)  at ( 0.0,-1.2) {$D$};

  \draw[->,thick] (X) -- (C);
  \draw[->,thick] (Y) -- (C);
  \draw[->,thick,gray!50,dashed] (C) -- (D);

  \node[black] at (0,-1.8) {\small path blocked, despite $D$ in the \emph{ambient} conditioning set};
  \draw[-, ultra thick] ($(X)+(-0.1,0.6)$) -- node[above] {\small $X \indep Y \given D \ \text{in}\ S$} ($(Y)+(0.1,0.6)$);
\end{tikzpicture}
\caption{\textbf{\(j\)-closure restores blocking.} On a \(J\)-cover \(S\), the offending link is removed/closed, so the collider path is blocked and \(X \indep Y \given D\) holds \emph{in} \(S\). Hence the CI is \(j\)-stable at \(U\).}
\label{fig:jclosure-restores-blocking}
\end{figure}

\paragraph{Worked example (backdoor vs.\ \(j\)-closure).}
Let \(U\) contain the graph \(U \to X \to Y\) and \(U \to Y\).
Classically, \(X \nindep Y\) but \(X \indep Y \given U\) by backdoor blocking.
Suppose we cannot condition on \(U\) everywhere, but our \(J\)-cover
\(S = \{\, S_1 \hookrightarrow U,\, S_2 \hookrightarrow U \,\}\) satisfies:
on \(S_1\) we can condition on (a proxy for) \(U\); on \(S_2\) we have an
intervention \(do(X)\) (i.e., the map \(S_2 \to U\) factors through the subtopos
where incoming arrows to \(X\) are deleted).
Then on both \(S_i\) the backdoor path is blocked, so \(X \indep Y \given Z_i\)
holds locally.
By \(j\)-stability, the sequent \(X \indep Y \given Z\) (with \(Z\) interpreting
the local data \((Z_1,Z_2)\)) holds in the \(J\)-closure of \(U\).
Intuitively, \(J\) collects the local “ways we can block the path”
(backdoor in \(S_1\), intervention in \(S_2\)) into a single global judgment.

\section{Sheaf Localization and the Limited Role of Kan Extensions}
\label{sec:sheaf-localization}

The soundness of the $j$-do rules above uses restriction, chartwise Pearl
semantics, and uniqueness of sheaf gluing.  It does not require a duality
between conditioning and intervention.  This section records the categorical
universal properties that are available without making that identification.

\paragraph{Sheafification.}
For a site $(\C,J)$, the inclusion of sheaves into presheaves has a left exact
left adjoint
\[
 a_J:[\C^{\mathrm{op}},\Set]\rightleftarrows
 \Sh_J(\C):i_J .
\]
Thus $\Sh_J(\C)$ is a reflective localization of the presheaf topos.
The functor $a_J$ imposes descent; it does not, by itself, define either a
causal intervention or a conditional probability.

\paragraph{Probability compatibility is an assumption.}
Let $T$ be a specified commutative probability monad on the ambient
presheaf topos.  To transport stochastic kernels through sheafification one
must additionally supply a monad $T_J$ on $\Sh_J(\C)$ and a coherent
isomorphism
\[
 a_JT\cong T_Ja_J
\]
compatible with the units and multiplications.  Equivalently, one may impose
conditions ensuring that $T$ preserves the relevant $J$-local equivalences
and sheaf objects.  Commutativity of $T$ and left exactness of $a_J$ alone do
not imply a unique such lift.  Our finite chartwise results avoid this issue
by assuming directly that kernels and their restrictions form sheaves.

\paragraph{The genuine free-cocompletion theorem.}
Let $\mathcal S$ be small and let
$y:\mathcal S\to[\mathcal S^{\mathrm{op}},\Set]$ be the Yoneda embedding.
For every cocomplete category $\mathcal D$, restriction along $y$ induces the
standard equivalence
\[
 \operatorname{Cocont}\bigl([\mathcal S^{\mathrm{op}},\Set],\mathcal D\bigr)
 \simeq [\mathcal S,\mathcal D],
\]
whose inverse sends $F:\mathcal S\to\mathcal D$ to
$\operatorname{Lan}_yF$.  This statement depends on the presheaf category
being the free cocompletion.  It does not hold for an arbitrary fully
faithful embedding into an arbitrary cocomplete category.  Passing to
$J$-sheaves further restricts the admissible semantics to those respecting
the $J$-descent relations.

\paragraph{Conditioning versus intervention.}
Van Belle's probability-theoretic result constructs the random-variable
functor on general probability spaces as a right Kan extension of its finite
counterpart; conditional expectation appears as its action on
measure-preserving maps \citep{vanbelle2023kan}.  This is a precise right-Kan
statement about probability functors.  By contrast, the intervention used in
this paper replaces a structural mechanism or stochastic kernel.  Ordinary
pushforward transports a fixed law and is not, in general, a causal
intervention.  We therefore make no claim that conditioning and intervention
form a left/right Kan duality.  A future base-change or Beck--Chevalley
comparison would require a separately defined indexed category of causal
models, existence of the relevant adjoints, and explicit exactness
hypotheses.

\section{Kernel Semantics and Worked Derivations}
\label{sec:do-tcm-translation}

Theorems~\ref{thm:J1}--\ref{thm:J3} give the sole formal statement of the
three rules.  We now make their finite-kernel semantics concrete and work
through representative derivations.  Observation and intervention remain
different operations: observation restricts and normalizes a law, whereas
intervention replaces a causal kernel and composes the modified model.

\begin{example}[Running example: regime-aware chain]
Let $\C$ have objects $e_0,e_1,e_2$ (regimes) covering a generic
stage $u$; write $e_i \to u \in \J(u)$. Consider internal variables
$X,Y,Z\in\Topos$ with stagewise edges
\[
  e_0 : X \to Y \to Z,\qquad
  e_1 : X \to Y,\; X \to Z,\qquad
  e_2 : X \to Y,\; Y \to Z.
\]
Then $X \indep Z \given Y$ holds at $e_0$ and $e_2$ but not at $e_1$.
Because $\{e_i\to u\}$ is a $\J$-cover, $X \indep Z \given Y$ is
$J$-stable at $u$ iff it holds on a covering family; here it fails
since it is refuted at $e_1$. Applying Theorem~\ref{thm:J2} at $e_0$ and $e_2$
but not at $u$ explains why identification that is valid in some regimes
need not be globally valid without $j$-stability.
\end{example}

\begin{figure}[t]
\centering
\begin{tikzpicture}[>=Latex, node distance=2.1cm]
  \tikzstyle{v}=[circle, draw, minimum size=18pt, inner sep=0pt]
  \node[v] (X) {$X$};
  \node[v, right=of X] (Y) {$Y$};
  \node[v, right=of Y] (Z) {$Z$};

  \draw[->, thick, blue!70] (X) -- (Y);
  \draw[->, thick, blue!70, bend left=20] (Y) to (Z);

  \draw[->, thick, red!70, bend right=18] (X) to (Z);

  \draw[->, thick, green!70] (Y) -- (Z);
\end{tikzpicture}
\caption{Regime-wise edges (colors) overlaid at a generic stage. 
$j$-stability demands that CI equalities hold on a cover; here
$X \indep Z \given Y$ fails globally because it fails in $e_1$ (red).}
\end{figure}
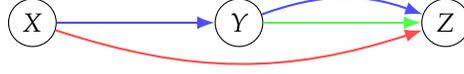

\subsection{Translation into TCM Language}
\label{sec:j-do-calculus}

Within a Topos Causal Model (TCM), variables are represented by objects
$X,Y,Z,W$ of the ambient topos~$\E$, and causal mechanisms by
stochastic morphisms (arrows in the Kleisli category $\Kl(\Dist_\E)$):
\[
P(Y\mid\mathrm{Pa}(Y)):\mathrm{Pa}(Y)\to \Dist_\E(Y).
\]
Interventions are represented by replacement morphisms for selected
mechanisms (with the resulting cut model recorded as part of the causal
structure), and observations by comprehension
subobjects (cf.~Lemmas~\ref{lem:obs-comprehension}
and~\ref{lem:do-kernel-replace}).

For the context $\Gamma:=X\times W$, let
\[
k:\Gamma\times Z\to\Dist_\E(Y)
\qquad\text{and}\qquad
k_0:\Gamma\to\Dist_\E(Y)
\]
denote the kernels corresponding respectively to
$P(Y\mid \Gamma,Z)$ and $P(Y\mid \Gamma)$ in the intervened model
$M_{\overline Z}$.
Conditional independence $Y\perp Z\mid \Gamma$
is expressed internally as the factorization
\begin{equation}\label{eq:indep-tcm}
k = k_0\circ\pi_\Gamma:\Gamma\times Z\to\Dist_\E(Y).
\end{equation}

\paragraph{Kernel irrelevance calculation.}
In the internal language of $\E$, independence~\eqref{eq:indep-tcm}
implies that for any intervention
$\mu:\Gamma\to\Dist_\E(Z)$,
\[
\Gamma\vdash
\int_Z k(\gamma,z)\,d\mu(\gamma)(z) = k_0(\gamma).
\]
Hence replacing the \(Z\)-kernel cannot affect \(Y\) when the structural
kernel is independent of \(Z\).
This equality is verified stagewise by Kripke–Joyal forcing:
at each stage $\alpha:N\to\Gamma$,
\[
\int_Z k(\alpha,z)\,d\mu(\alpha)(z)
  = \int_Z k_0(\alpha)\,d\mu(\alpha)(z)
  = k_0(\alpha).
\]

\paragraph{Relation to the three rules.}
The calculation above is a semantic lemma used after the appropriate
mutilated-graph premise has been established.  It is not by itself any one
of Pearl's three rules.  Rule~1 compares two observational conditionals in
\(G_{\bar X}\); Rule~2 compares observing and intervening on \(Z\) in
\(G_{\bar X,\underline Z}\); Rule~3 removes an intervention in
\(G_{\bar X,\overline{Z(W)}}\).  Theorems~\ref{thm:J1}--\ref{thm:J3} retain
these distinct surgeries.

\bigskip

\paragraph{Summary.}
Under the hypotheses of the local-to-global theorem, each chartwise rule
becomes an internal equality of stochastic arrows.  A useful special case is
\[
(Y\perp Z\mid\Gamma)\;\Rightarrow\;
\big(P(Y\mid do(Z),\Gamma)=P(Y\mid\Gamma)\big).
\]
The graphical premise remains essential; the equality is not a
Heyting-algebra tautology and does not follow from topos logic alone.


\subsection{Do-Calculus as Internal Equalities in a TCM}

\paragraph{Setup.}
Let $\mathcal{E}$ be a (Markov) topos supporting a distribution monad $\Dist$ on objects and stochastic morphisms in the Kleisli category. 
A Topos Causal Model (TCM) object $M$ specifies, for variables $X,Y,Z,W$, kernels
\[
P(X),\quad P(Y\mid \cdot),\quad P(Z\mid \cdot),\quad P(W\mid \cdot)
\]
factoring the joint in the usual causal way. Interventions replace selected
kernels by chosen deltas or stochastic policies. Cutting all incoming arrows
to $Z$ yields the \emph{mutilated object} $M_{\overline Z}$.

Throughout, write $\Gamma := X\times W$ for context. Internally (Mitchell–Bénabou language), a conditional $P(Y\mid \Gamma,Z)$ is a kernel
\[
k \;:\; \Gamma \times Z \longrightarrow \Dist(Y),
\]
while the marginal conditional $P(Y\mid \Gamma)$ is
\[
k_0 \;:\; \Gamma \longrightarrow \Dist(Y).
\]
Independence $Y \perp Z \mid \Gamma$ in $M_{\overline Z}$ is the internal factorization
\begin{equation}\label{eq:factor}
k \;=\; k_0 \circ \pi_\Gamma \quad:\; \Gamma \times Z \to \Dist(Y),
\end{equation}
i.e. $k$ ignores its $Z$-argument.

\paragraph{Interventions as integration.}
Given an intervention $do(Z\sim\mu_Z)$ with $\mu_Z:\mathbf{1}\to\Dist(Z)$ (delta for $do(Z=z)$), define the interventional conditional by pushforward
\begin{equation}\label{eq:push}
E_Z(k) \;:=\; \int_Z k(\,\cdot\,,z)\, d\mu_Z(z) \;:\; \Gamma \longrightarrow \Dist(Y).
\end{equation}

\begin{proposition}[Kernel irrelevance under mechanism replacement]
In $M_{\overline Z}$, if $Y \perp Z \mid \Gamma$ (i.e.\ \eqref{eq:factor} holds), then for every intervention $do(Z\sim\mu_Z)$,
\[
\Gamma \;\vdash\; E_Z(k) \;=\; k_0,
\]
hence replacing the \(Z\)-kernel does not change the resulting \(Y\)-kernel.
\end{proposition}

\begin{proof}[Kripke--Joyal proof]
Let $\alpha:N\to\Gamma$ be an arbitrary stage. We must show equality of the two arrows $N\to\Dist(Y)$ obtained by precomposition with $\alpha$.
Using \eqref{eq:push} and \eqref{eq:factor}:
\[
(E_Z(k)\circ \alpha)
\;=\; \int_Z k(\alpha,z)\, d\mu_Z(z)
\;=\; \int_Z (k_0\circ\pi_\Gamma)(\alpha,z)\, d\mu_Z(z)
\;=\; \int_Z k_0(\alpha)\, d\mu_Z(z)
\;=\; k_0(\alpha).
\]
Since this holds for all $\alpha$ and all $\mu_Z$, the internal equality is forced.
\end{proof}

\paragraph{Diagrammatic view.}
The factorization $k=k_0\circ\pi_\Gamma$ is equivalently a pullback/factorization condition:
\[
\begin{tikzcd}[column sep=large]
\Gamma\times Z \ar[rr,"k"] \ar[dr,swap,"\pi_\Gamma"] & & \Dist(Y) \\
& \Gamma \ar[ur,swap,"k_0"] &
\end{tikzcd}
\quad\text{(commutes in } \mathcal{E}\text{).}
\]
The intervention integrates along $\pi_Z:\Gamma\times Z\to Z$, which is vacuous under the factorization.

\medskip

\begin{corollary}[Observation and replacement under kernel irrelevance]
In $M_{\overline{Z(W)}}$, if $Y \perp Z \mid X,W$ then
\[
P(Y\mid do(Z),do(W),X) \;=\; P(Y\mid do(W),X).
\]
\end{corollary}

\begin{proof}[Sketch]
Work in context $\Gamma:=X$ inside $M_{\overline{Z(W)}}$. Independence gives $k(\gamma,z)=k_0(\gamma)$.
Observation of $Z$ corresponds to conditioning via the comprehension subobject; intervention corresponds to integrating against an arbitrary $\mu_Z$.
Both operations erase the $z$-argument of $k$ by factorization, hence yield
$k_0$, provided the observational normalizing mass is positive.  This is an
irrelevance corollary in the specified cut model, not a replacement for the
full action/observation exchange theorem.
\end{proof}

\paragraph{Remarks.}
(i) Rule 3 is analogous: when the cut graph makes $Z$ causally irrelevant for $Y$ given $(X,W)$, replacing the $Z$-kernel is isomorphic to the identity on conditionals. 
(ii) These equalities live in the Heyting algebra of subobjects of the TCM; they are internal entailments rather than external assumptions.

\paragraph{Standing setting.}
Work internally in a (presheaf) topos $\E$ with the pointwise finite-support distribution monad $\Dist_\E$ (cf.\ main text).
Objects $X\xto{k}\Dist(Y)$ are \emph{stochastic kernels} (arrows of $\Kl(\Dist_\E)$).
Integration/pushforward along a kernel $k:\Gamma\times Z\to\Dist(Y)$ against a
state $\mu:\Gamma\to\Dist(Z)$ is written
\[
\int\nolimits_Z k(\gamma,z)\,d\mu(\gamma)(z) \;:\; \Gamma \to \Dist(Y),
\]
defined objectwise (finite sums) in the presheaf case.

\medskip

\begin{lemma}[Observation = comprehension subobject + normalization]
\label{lem:obs-comprehension}
Let $\Gamma,Z,Y$ be objects in $\E$.
Let $p:\Gamma\to\Dist(Z)$ be a prior (kernel) and $k:\Gamma\times Z\to\Dist(Y)$ a likelihood kernel.
Let $\chi:\Gamma\times Z\to\Omega$ be a predicate (internal event) with comprehension mono
$\iota_\chi:\Gamma\!{\mid}\!\chi \hookrightarrow \Gamma\times Z$,
and write $\pi_\Gamma:\Gamma\times Z\to\Gamma$.

Define the \emph{observed} posterior kernel
\[
\mathrm{Obs}_\chi(k,p) \;:\; \Gamma \longrightarrow \Dist(Y)
\]
stagewise by, for each $\alpha:N\to\Gamma$,
\[
\big(\mathrm{Obs}_\chi(k,p)\circ \alpha\big)
\;:=\; \frac{\displaystyle\int_{z\in Z} \mathbf{1}_{\chi}(\alpha,z)\, k(\alpha,z)\, d p(\alpha)(z)}
               {\displaystyle\int_{z\in Z} \mathbf{1}_{\chi}(\alpha,z)\, d p(\alpha)(z)}.
\]
This definition is made only on the subobject of stages where the denominator
is strictly positive; no conditional probability is assigned at zero mass.

Then $\mathrm{Obs}_\chi(k,p)$ is the unique arrow $\Gamma\to\Dist(Y)$ such that
\[
\pi_\Gamma^{\!*}(p) \text{ conditioned on } \chi
\quad\text{and pushed forward by }k\quad =\quad
\mathrm{Obs}_\chi(k,p)\quad\text{(as a kernel $\Gamma\to\Dist(Y)$)}.
\]
Equivalently, observation $($conditioning on $\chi)$ is
\emph{(i)} restricting along the comprehension subobject $\iota_\chi$ and \emph{(ii)} normalizing.
\end{lemma}

\begin{proof}
Internally (Kripke–Joyal): at a stage $\alpha:N\to\Gamma$, the prior is the finite measure $p(\alpha)$ on $Z$,
the event indicator is $\mathbf{1}_\chi(\alpha,\cdot)$, and the likelihood is $k(\alpha,\cdot)$.
Restricting to the comprehension subobject multiplies by $\mathbf{1}_\chi$; normalization divides by its total mass.
Pushing forward along $k$ is integration of $k(\alpha,\cdot)$ against the normalized prior.
Uniqueness follows directly from equality of the normalized finite sums at
each component.
\end{proof}

\medskip

\begin{lemma}[Intervention = kernel replacement $+$ integration]
\label{lem:do-kernel-replace}
Let $k:\Gamma\times Z\to\Dist(Y)$ be a structural kernel (e.g.\ $P(Y\mid \Gamma,Z)$ in the cut object $M_{\overline{Z}}$).
An \emph{intervention on $Z$} with policy $\mu:\Gamma\to\Dist(Z)$ (delta for $do(Z{=}z)$) defines the interventional kernel
\[
\mathrm{Do}_Z(k;\mu) \;:=\; \int_{z\in Z} k(\gamma,z)\, d\mu(\gamma)(z) \;:\; \Gamma \longrightarrow \Dist(Y).
\]
Types:
\[
k:\Gamma\times Z\to\Dist(Y),\qquad
\mu:\Gamma\to\Dist(Z),\qquad
\mathrm{Do}_Z(k;\mu):\Gamma\to\Dist(Y).
\]
Moreover, if $k$ is independent of $Z$ in $M_{\overline{Z}}$ (i.e.\ $k=k_0\circ\pi_\Gamma$), then
$\mathrm{Do}_Z(k;\mu)=k_0$ for \emph{every} $\mu$.
\end{lemma}

\begin{proof}
Definition is the Kleisli composition (convolution) of $k$ with $\mu$.
At a stage $\alpha:N\to\Gamma$, $\mathrm{Do}_Z(k;\mu)\circ\alpha=\int k(\alpha,z)\,d\mu(\alpha)(z)$ by definition.
If $k=k_0\circ\pi_\Gamma$, then $k(\alpha,z)=k_0(\alpha)$ is constant in $z$, hence the integral returns $k_0(\alpha)$.
\end{proof}

\medskip

\paragraph{Corollary (Rule 1, internal form).}
In the cut object $M_{\overline Z}$, if $k=k_0\circ\pi_\Gamma$ (i.e.\ $Y\perp Z\mid \Gamma$),
then for every intervention $\mu:\Gamma\to\Dist(Z)$,
\[
\Gamma \vdash \; \mathrm{Do}_Z(k;\mu) \;=\; k_0,
\]
and for every observation predicate $\chi$ with prior $p:\Gamma\to\Dist(Z)$,
\[
\Gamma \vdash \; \mathrm{Obs}_\chi(k,p) \;=\; k_0
\quad\text{whenever}\quad
\int \mathbf{1}_\chi\, dp > 0.
\]

\subsection{Translation of Do-Calculus Rules into TCM Models}

The translation is now fixed by Theorems~\ref{thm:J1}--\ref{thm:J3}:
first select the appropriate mutilated model, then apply the corresponding
Pearl rule on every chart, and finally descend the compatible equality.
Lemma~\ref{lem:obs-comprehension} supplies observation under positive
normalizing mass, while Lemma~\ref{lem:do-kernel-replace} supplies
intervention by mechanism replacement.  A generic kernel-independence
calculation is not relabeled as all three rules.

\subsection{Presheaf + Distribution-Monad Formalization of TCMs}
\label{subsec:presheaf-dist-formalization}

\paragraph{Presheaf topos and internal distributions.}
Fix a small category $\mathcal{C}$ (of ``contexts'' or shapes). Let
\[
\E \;=\; \widehat{\mathcal{C}} \;=\; \Set^{\mathcal{C}^{op}}
\]
be the presheaf topos. For $F\in \E$, write $F(c)$ for sections at stage $c\in \mathcal{C}$ and $F(u):F(c)\to F(c')$ for restriction along $u:c'\to c$. 
Define the internal finite-support distribution monad pointwise:
\[
(\Dist_\E F)(c) \;=\; \Dist_{\Set}\!\big(F(c)\big), 
\qquad 
(\Dist_\E F)(u) \;=\; \Dist_{\Set}\!\big(F(u)\big).
\]
The Kleisli (``Markov'') category $\Kl(\Dist_\E)$ has the same objects as $\E$; morphisms $F\to G$ are natural transformations $F \Rightarrow \Dist_\E G$ (``stochastic natural transformations''). Composition is pointwise convolution.

\paragraph{Objects and kernels.}
A variable $X$ is a presheaf $X\in\E$. A \emph{stochastic kernel} in context $\Gamma$ with a controlled variable $Z$ and response $Y$ is a natural transformation
\[
k\;:\; \Gamma \times Z \;\Rightarrow\; \Dist_\E Y,
\]
i.e.\ for each $c\in \mathcal{C}$ a stochastic map $k_c:\Gamma(c)\times Z(c) \to \Dist(Y(c))$ that is natural in~$c$.

\paragraph{Observation via comprehension + normalization.}
An internal predicate $\chi:\Gamma\times Z \to \Omega$ yields its comprehension mono $\iota_\chi:\Gamma{\mid}\chi \hookrightarrow \Gamma\times Z$. 
Given a prior $p:\Gamma \Rightarrow \Dist_\E Z$ and likelihood $k:\Gamma\times Z \Rightarrow \Dist_\E Y$, define the observed posterior
\[
\Obs_\chi(k,p) \;:\; \Gamma \Rightarrow \Dist_\E Y
\]
stagewise by
\[
\big(\Obs_\chi(k,p)\big)_c(\gamma)
\;=\;
\frac{\displaystyle\sum_{z\in Z(c)} \mathbf{1}_\chi(\gamma,z)\; k_c(\gamma,z)\; p_c(\gamma)(z)}
     {\displaystyle\sum_{z\in Z(c)} \mathbf{1}_\chi(\gamma,z)\; p_c(\gamma)(z)}.
\]
This is defined on the subpresheaf of pairs \((c,\gamma)\) for which the
denominator is positive, and is pointwise restriction along $\iota_\chi$
followed by normalization.

\paragraph{Intervention as kernel replacement + integration.}
An intervention policy is a stochastic nat.\ transf.\ $\mu:\Gamma \Rightarrow \Dist_\E Z$. 
The interventional kernel is the Kleisli composite
\[
\Do_Z(k;\mu) \;:=\; \mu_Y \circ \Dist_\E(k) \circ st \circ \langle \id,\mu\rangle
\;:\; \Gamma \Rightarrow \Dist_\E Y,
\]
whose $c$-component is the usual integral
\[
\big(\Do_Z(k;\mu)\big)_c(\gamma) \;=\; \sum_{z\in Z(c)} k_c(\gamma,z)\; \mu_c(\gamma)(z).
\]

\paragraph{Independence as factorization.}
Conditional independence $Y \perp Z \mid \Gamma$ \emph{in the cut model} $M_{\overline Z}$ is the internal naturality equation
\[
k \;=\; k_0 \circ \pi_\Gamma \;:\; \Gamma\times Z \Rightarrow \Dist_\E Y,
\]
i.e.\ each $k_c(\gamma,z)$ ignores $z$ and equals $k_{0,c}(\gamma)$.

\paragraph{Kernel irrelevance inside the presheaf model.}
If $k=k_0\circ \pi_\Gamma$ in $M_{\overline Z}$, then for every policy $\mu:\Gamma \Rightarrow \Dist_\E Z$,
\[
\Do_Z(k;\mu) \;=\; k_0 \quad\text{in } \E,
\]
because pointwise $\sum_{z} k_c(\gamma,z)\,\mu_c(\gamma)(z)=k_{0,c}(\gamma)$.

\medskip

\paragraph{Worked example (constant presheaves).}
Let $\mathcal{C}$ be arbitrary and take \emph{constant} presheaves
$X,Y,Z,\Gamma$ with values the two-point set $\{0,1\}$, so all restrictions are identities.
Fix pointwise kernels (same at every stage $c$):
\[
P(Y{=}1\mid X{=}0)=0.1,\quad P(Y{=}1\mid X{=}1)=0.9,\qquad
P(Z{=}1\mid Y{=}0)=0.2,\quad P(Z{=}1\mid Y{=}1)=0.8.
\]
Thus $k:\Gamma\times Z\Rightarrow\Dist_\E Y$ is \emph{independent of $Z$} (take $\Gamma:=Y$ or impose the cut to $Z$), yielding $k=k_0\circ\pi_\Gamma$.
For any policy $\mu:\Gamma\Rightarrow\Dist_\E Z$ (e.g.\ $\mu(\gamma)=\delta_{z_0}$ for $do(Z{=}z_0)$),
\[
\big(\Do_Z(k;\mu)\big)_c(\gamma)
= \sum_{z} k_c(\gamma,z)\, \mu_c(\gamma)(z)
= k_{0,c}(\gamma),
\]
so internally $P(Y\mid do(Z),\Gamma)=P(Y\mid \Gamma)$ at every stage $c$.
By contrast, observation with $\chi(z){:=}[z{=}1]$ produces
\[
\big(\Obs_\chi(k,p)\big)_c(\gamma)
= \frac{k_{0,c}(\gamma)\;p_c(\gamma)(1)}{p_c(\gamma)(1)} = k_{0,c}(\gamma)
\quad \text{whenever } p_c(\gamma)(1)>0,
\]
showing that, in this special factorized kernel, both operations return the
same \(Y\)-kernel.  This calculation does not replace the mutilated-graph
premise of Pearl's action/observation rule.

\paragraph{Takeaway.}
For the pointwise finite-distribution monad used here, kernel composition is
computed objectwise.  A global causal conclusion still requires naturality,
the correct mutilated model, and---after sheafification---the compatibility
assumptions of Theorem~\ref{thm:local-global-soundness}.

\subsection{Exponential Objects and Their Role in TCMs}
\label{subsec:exponentials-tcm}

\paragraph{Definition.}
Every topos $\E$ is \emph{cartesian closed}:
for objects $A,B\in\E$ there exists an exponential object $B^A$
and an evaluation map $\mathrm{ev}:B^A\times A\to B$
such that for all $X$ there is a natural isomorphism
\[
\E(X\times A, B) \;\cong\; \E(X, B^A).
\]
Hence morphisms depending on parameters in~$A$
can be re-expressed as internal elements of $B^A$.

\paragraph{Conditional kernels.}
In a Topos Causal Model,
a stochastic kernel $k:\Gamma\times Z\to\Dist_\E(Y)$
can equivalently be seen as
\[
\tilde{k}:\Gamma\to (\Dist_\E Y)^{Z},
\qquad
k=\mathrm{ev}\circ\langle \tilde{k},\id_Z\rangle.
\]
This interpretation allows causal mechanisms and conditionals
to be treated as \emph{elements of an exponential object}.

\paragraph{Interventions as higher-order morphisms.}
The intervention operator acts internally as a morphism between exponentials:
\[
\Do_Z : (\Dist_\E Y)^{\Gamma\times Z}
         \times (\Dist_\E Z)^{\Gamma}
         \longrightarrow
         (\Dist_\E Y)^{\Gamma},
\]
whose externalization sends
$(k,\mu)$ to
$\Do_Z(k;\mu)
  = \mu_Y\circ\Dist_\E(k)\circ st\circ\langle \id,\mu\rangle$.
Thus interventions are higher-order arrows
within the cartesian closed structure of~$\E$.

\paragraph{Quantifiers and forcing.}
Exponentials support internal quantification over function spaces.
For example, the formula
$(\forall f:Z\to Y)\,\varphi(f)$
is interpreted using $Y^Z$, and its Kripke--Joyal clause reads:
for all $u:N'\to N$ and all $f:N'\to Y^Z$,
$N'\Vdash \varphi(f)[\alpha\!\circ\!u]$.
Hence statements such as
``for every policy $\mu$'' or
``there exists a kernel $k$''
are expressed directly in the internal logic.

\paragraph{Presheaf case.}
When $\E=\widehat{\mathcal{C}}$, exponentials are not generally obtained by
the naive pointwise function-set construction.  They satisfy
\[
(F^G)(c)\;=\;\mathrm{Nat}(G\times h_c, F),
\]
so an internal arrow $\Gamma\to (\Dist_\E Y)^{Z}$
assigns to each $c\in\mathcal{C}$
a natural family of stochastic maps
$Z(c)\to \Dist(Y(c))$ varying functorially in~$c$.

\begin{figure}[h]
  \centering
  \begin{tikzcd}[column sep=large, row sep=large]
    \Gamma \times Z \arrow[rr, "k"]
       \arrow[dr, swap, "{\langle \tilde{k}\circ\pi_\Gamma,\ \pi_Z\rangle}"] & & \Dist(Y) \\
    & \Dist(Y)^{Z} \times Z
        \arrow[ur, swap, "\mathrm{ev}"] &
  \end{tikzcd}

  \caption{Exponential adjunction for kernels. 
  Each stochastic kernel \(k:\Gamma\times Z\to \Dist(Y)\)
  corresponds uniquely to an internal element
  \(\tilde{k}:\Gamma\to \Dist(Y)^{Z}\),
  satisfying \(k = \mathrm{ev}\circ\langle \tilde{k}\circ\pi_\Gamma, \pi_Z\rangle\).}
  \label{fig:exp-eval-diagram}
\end{figure}

\paragraph{Summary.}
Exponentials provide the categorical infrastructure
for higher-order reasoning in TCMs:
they internalize conditionals and policies,
make intervention operators morphisms,
and enable quantification over functions
in the Kripke--Joyal semantics.
In short, they turn the causal calculus of TCMs
into a genuine higher-order internal logic.

\subsection{Example: Generalizing Do-Calculus in a Simple TCM}
\label{subsec:example-do-tcm}

\paragraph{Setup.}
Consider a simple causal system with two observable variables
$X$ (treatment) and $Y$ (outcome), and an optional confounder $Z$.
In the ordinary probabilistic semantics we have a factorization
\[
P(X,Y,Z) = P(Y\mid X,Z)\,P(X\mid Z)\,P(Z),
\]
and interventions replace $P(X\mid Z)$ by a chosen policy $\mu_X$.

\paragraph{Internal TCM formulation.}
Let $\E$ be a topos (e.g.\ the presheaf topos $\Set^{\mathcal{C}^{op}}$)
equipped with the internal finite-support distribution monad~$\Dist_\E$.
Objects $X,Y,Z$ represent the corresponding variables as presheaves,
and the causal mechanisms are stochastic morphisms
\[
k_Y : X\times Z \longrightarrow \Dist_\E(Y),
\qquad
k_X : Z \longrightarrow \Dist_\E(X).
\]
The joint is the Kleisli composite
\[
P(X,Y,Z)
  = (\id_Z \otimes k_X \otimes k_Y)
  : 1 \longrightarrow \Dist_\E(X\times Y\times Z).
\]

\paragraph{Observation and intervention.}
Observation of $X=x$ corresponds to restricting along the
comprehension subobject $\iota_{x}:Z{\mid}X{=}x\hookrightarrow Z\times X$
and renormalizing.
An intervention $do(X\!\sim\!\mu_X)$ replaces $k_X$
by a constant kernel $\mu_X:1\to \Dist_\E(X)$ and composes via the
Kleisli operation:
\[
\Do_X(k_Y;\mu_X)
  = \mu_Y \circ \Dist_\E(k_Y)
  \circ st
  \circ \langle \id,\mu_X\rangle
  : Z \to \Dist_\E(Y).
\]
At each stage $c\in\mathcal{C}$ this reduces to the ordinary formula
\[
(\Do_X(k_Y;\mu_X))_c(z)
   = \sum_{x\in X(c)} k_{Y,c}(x,z)\,\mu_{X,c}(x).
\]

\paragraph{Kernel irrelevance in the cut model.}
Suppose that inside the cut model $M_{\overline X}$ the kernel
$k_Y$ does not depend on~$X$, i.e.
$k_Y = k_{0,Y}\circ \pi_Z$ in $\E$.
Then by Kripke–Joyal semantics, for every stage $\alpha:N\to Z$,
\[
\int_X k_Y(\alpha,x)\,d\mu_X(x)
   = \int_X k_{0,Y}(\alpha)\,d\mu_X(x)
   = k_{0,Y}(\alpha),
\]
so internally
\[
Z \vdash P(Y\mid do(X),Z)=P(Y\mid Z).
\]
This is the kernel-irrelevance lemma expressed as an internal equality.  To
obtain a do-calculus rule one must additionally use the specific surgery and
premise in Theorems~\ref{thm:J1}--\ref{thm:J3}.

\paragraph{A special observation/intervention coincidence.}
If, in the model $M_{\overline{X(W)}}$, we have
$Y\!\perp\! X\mid (Z,W)$,
then both the observational restriction
and the interventional replacement yield the same morphism
$k_{0,Y}:\Gamma\to\Dist_\E(Y)$,
provided the observation has positive normalizing mass.  This factorized
special case is consistent with Rule~2 but is not its general statement.

\paragraph{Compositionality.}
Nested interventions compose through the monad multiplication:
for any two controlled variables $X,Z$ with policies
$\mu_X:\Gamma\to\Dist_\E(X)$ and $\mu_Z:\Gamma\to\Dist_\E(Z)$,
the associative law of $\Dist_\E$ ensures that
\[
\Do_Z(\Do_X(k;\mu_X);\mu_Z)
   = \Do_{X\times Z}\big(k;\ \mathrm{mix}(\mu_X,\mu_Z)\big),
\]
expressing associativity of sequential kernel replacement for this
commutative distribution monad.

\paragraph{Intuitive summary.}
At each stage $c\in\mathcal{C}$ this example reproduces an ordinary
finite-kernel equality.  Naturality makes the equality internal across all
stages.
Kripke–Joyal forcing expresses this as:
\[
\forall \alpha:N\to\Gamma,\quad
N\Vdash P(Y\mid do(X),\Gamma)=P(Y\mid \Gamma),
\]
The full do-calculus conclusions follow only under the causal and descent
hypotheses stated earlier; internal logic alone does not supply them.

\section{Exchangeable $j$-Stable Causality}
\label{sec:j-exchangeable}

A central theme in modern causal inference is \emph{symmetry}: many datasets consist of partially interchangeable units (patients, households, pixels). The recent ``do-Finetti'' viewpoint \citep{do-finetti} studies do-calculus under exchangeability. In our setting, we formulate the analogue at the level of a $j$-site and show how $j$-stability interacts with permutation invariance.

\paragraph{Setup.}
Let $\C$ be a causal site and $\widehat{\C}$ its presheaf topos. Fix a Lawvere--Tierney topology $j:\Omega\to\Omega$ on $\widehat{\C}$.
Write $\mathrm{Sym}(I)$ for the finite permutation group of a finite index set $I$ of units.
A family of random variables $X=(X_i)_{i\in I}$, outcomes $Y=(Y_i)_{i\in I}$, and covariates $Z=(Z_i)_{i\in I}$ is represented by an arrow
\[
\mathbf{X}=(\mathbf{X}_i)_{i\in I}:\;1\to \prod_{i\in I}\mathsf{X}_i
\qquad
\mathbf{Y}=(\mathbf{Y}_i)_{i\in I}:\;1\to \prod_{i\in I}\mathsf{Y}_i
\qquad
\mathbf{Z}=(\mathbf{Z}_i)_{i\in I}:\;1\to \prod_{i\in I}\mathsf{Z}_i,
\]
with the natural $\mathrm{Sym}(I)$–action by permuting factors.

\begin{definition}[{$j$–invariant (exchangeable) family}]
\label{def:j-exch}
The family $(\mathbf{X},\mathbf{Y},\mathbf{Z})$ is \emph{$j$–exchangeable} if for every $\pi\in\mathrm{Sym}(I)$ and every stage $U\in\C$, the equality in distribution
\[
(\mathbf{X},\mathbf{Y},\mathbf{Z}) \;\overset{d}{=}\; (\mathbf{X},\mathbf{Y},\mathbf{Z})\circ \pi
\]
is $j$–forced at $U$ (i.e., holds in every chart of some $j$–cover of $U$). Equivalently, the \emph{$j$–truth} of any internal sentence $\varphi(\mathbf{X},\mathbf{Y},\mathbf{Z})$ is invariant under relabeling by $\pi$.
\end{definition}

\begin{definition}[{$G$–invariant $j$}]
\label{def:G-invar-j}
Let $G\le \mathrm{Sym}(I)$. We say $j$ is \emph{$G$–invariant} if for every covering sieve $S$ on $U$ and every $g\in G$, the translated sieve $g\cdot S$ is again $j$–covering. Intuitively, $j$ does not distinguish unit labels.
\end{definition}

\paragraph{Interventions with symmetry.}
For a finite treatment set $S\subseteq I$ and value $x$, the intervention $\mathrm{do}(\mathbf{X}_S=x)$ is represented by a natural transformation that surgically sets coordinates in $S$ to $x$ and leaves others unchanged. When $j$ is $G$–invariant, the map $S\mapsto \mathrm{do}(\mathbf{X}_S=x)$ is $G$–equivariant.

\begin{proposition}[Permutation–invariance of $j$–stable effects]
\label{prop:perm-inv-effect}
Assume (i) $j$ is $G$–invariant for some $G\le \mathrm{Sym}(I)$, (ii) $(\mathbf{X},\mathbf{Y},\mathbf{Z})$ is $j$–exchangeable, and (iii) the $j$–stable rules of Section~\ref{sec:j-do-calculus} apply (e.g., $j$–Markov and backdoor rules). Then for any finite $S\subseteq I$ and any $\pi\in G$,
\[
\text{$j$–}\mathbb{E}\!\big[\mathbf{Y}\mid \mathrm{do}(\mathbf{X}_S=x),\mathbf{Z}\big]
\;=\;
\text{$j$–}\mathbb{E}\!\big[\mathbf{Y}\mid \mathrm{do}(\mathbf{X}_{\pi(S)}=x),\mathbf{Z}\big],
\]
$U$–locally for every stage $U$. In particular, the effect depends on $S$ only through the $G$–orbit of $S$ (e.g., its cardinality when $G=\mathrm{Sym}(I)$).
\end{proposition}

\begin{proof}[Idea]
$G$–invariance of $j$ transports covering sieves along permutations; $j$–exchangeability transports the local graphical/separation premises used by the $j$–rules. The $j$–do rules then produce identical local conclusions on permuted charts, hence identical $j$–forced effects.
\end{proof}

\paragraph{A $j$–de Finetti principle (informal).}
When $I$ is large and $(\mathbf{X}_i,\mathbf{Y}_i,\mathbf{Z}_i)_{i\in I}$ is $j$–exchangeable across $i$, one can state an internal version of de Finetti: there exists an \emph{internal random measure} $\boldsymbol{\Theta}$ such that, $j$–locally,
\[
(\mathbf{X}_i,\mathbf{Y}_i,\mathbf{Z}_i)_{i\in I} \;\text{is conditionally i.i.d. given } \boldsymbol{\Theta}.
\]
Operationally, this licenses the usual empirical–Bayes reductions (pooling across units) \emph{inside} the $j$–logic, and ensures that $j$–stable effects for interventions on $S$ depend only on the orbit type of $S$ (e.g., $|S|$). This recovers the spirit of ``do–Finetti'' while accounting for covers/interventions in TCM.%
\footnote{For a classical account, see recent treatments of do–calculus under exchangeability; here we phrase the equivalence internally, relative to a $G$–invariant $j$.}

\paragraph{Partial exchangeability and regimes.}
If units split into regimes $I=I_1\sqcup\cdots\sqcup I_R$ with $G=\mathrm{Sym}(I_1)\times\cdots\times \mathrm{Sym}(I_R)$, Definitions~\ref{def:j-exch}–\ref{def:G-invar-j} and Proposition~\ref{prop:perm-inv-effect} apply verbatim. $j$–stable effects are invariant under label permutations \emph{within} regimes and may depend on the treatment counts $(|S\cap I_r|)_{r=1}^R$.

\paragraph{Worked toy example.}
Let $I=\{1,\dots,n\}$, $\mathrm{Sym}(I)$ act by relabeling units, and consider the DAG
\[
\mathbf{X}_i \to \mathbf{Y}_i,\qquad
\mathbf{Z}_i \to \mathbf{X}_i,\ \mathbf{Z}_i \to \mathbf{Y}_i
\quad (i\in I),
\]
with $(\mathbf{Z}_i)$ i.i.d.\ unobserved. Suppose $j$ has two charts per stage: an observational chart and an interventional chart where incoming arrows to $\mathbf{X}_S$ are cut. If $(\mathbf{X},\mathbf{Y})$ is $j$–exchangeable and $j$ is $\mathrm{Sym}(I)$–invariant, then by Proposition~\ref{prop:perm-inv-effect}
\[
\text{$j$–}\mathbb{E}\big[\tfrac{1}{n}\sum_i \mathbf{Y}_i\ \big|\ \mathrm{do}(\mathbf{X}_S=x)\big]
\;=\; F(|S|,x)
\]
for some function $F$ that depends only on the \emph{number} treated, not their labels. If $|S|/n\to p$ along a cofinal system of stages, a $j$–de Finetti posterior over $\boldsymbol{\Theta}$ yields a limit $F(|S|,x)\to F_\infty(p,x)$.

\paragraph{Practical upshot.}
In $j$–stable learning, $G$–invariance lets us (i) aggregate evidence across permuted charts (variance reduction), (ii) constrain estimators to depend only on orbit features (e.g., treated fraction), and (iii) define interventions that \emph{commute} with relabeling. This is especially natural for panel/cluster data and for regime–wise exchangeability.

\paragraph{Connections to do–Finetti.}
The exchangeable $j$–stable framework recovers the core invariance statements of do–Finetti \citep{do-finetti} in the classical setting when $j$ is the trivial topology and $\C$ is a one–object site. Our formulation clarifies how such principles persist when (a) interventions are represented as morphisms on a site, and (b) causal judgments are taken $j$–locally via covers.

\paragraph{Related Work.}
Closest to our setting is the recent ``do–Finetti'' line of work \citep{do-finetti}, which investigates causal identifiability and transport under exchangeability in the classical (Boolean) semantics. Our $j$–stable formulation places exchangeability \emph{internally} to a presheaf topos via a $G$–invariant Lawvere–Tierney topology, so that symmetry is preserved chartwise and commutes with interventions. This connects to de Finetti–type representations and their array analogues (Aldous–Hoover) by interpreting conditional i.i.d.\ structure as $j$–local. It is complementary to invariance-based causal methods such as Invariant Causal Prediction (ICP) and pooled/transport rules, which correspond to particular choices of covers but do not supply an internal logic for interventions on orbits. Our treatment also aligns with recent interest in symmetry/equivariance in causal discovery and learning, but makes the group action explicit at the site/topology level, yielding orbit-wise effect functionals $F(|S|,x)$ in the fully exchangeable case.

\section{Summary}
\label{limitations}

We developed a regime-indexed, sheaf-valued semantics for Pearl's
do-calculus.  The causal work is performed chartwise by ordinary SCM graph
surgery and mechanism replacement.  The sheaf-theoretic work is descent:
when restrictions preserve the relevant causal operations, compatible local
equalities glue uniquely.  This gives a sound local-to-global rule under
explicit hypotheses; it is not a completeness theorem for identification on
arbitrary sites.

The revision also clarifies several boundaries.  Conditioning requires
positive normalizing mass or a supplied disintegration.  Intervention is
mechanism or kernel replacement, not ordinary pushforward.  A trivial
Grothendieck topology yields presheaf semantics but not Boolean logic in
general; the terminal site recovers the ordinary set-based case.  Finally,
the probability-theoretic right Kan construction of conditional expectation
does not by itself provide a dual left-Kan semantics for causal
intervention.  Establishing any such comparison would require additional
indexed causal structure and exactness hypotheses.

\section{Changes in This Revision}
\label{sec:revision-note}

Relative to the first arXiv version, this revision:
\begin{enumerate}
  \item replaces the claim that the trivial topology is the Boolean case by
  separate stagewise and terminal-site conservativity statements;
  \item removes the false claim that d-separation is monotone under enlarging
  the conditioning set;
  \item removes duplicate and incompatible formulations of the three rules
  and retains Pearl's original mutilated-graph premises;
  \item states the restriction, conditioning, intervention, overlap, and
  positivity assumptions required for local equalities to descend;
  \item corrects ``co-Kleisli'' to ``Kleisli'' for the stochastic category of
  a distribution monad;
  \item replaces the claimed automatic probability-monad lift and arbitrary
  TCM cocompletion theorem by the standard sheaf-localization and presheaf
  free-cocompletion statements;
  \item distinguishes observation by restriction and normalization from
  intervention by mechanism replacement; and
  \item withdraws any universal left/right Kan duality between intervention
  and conditioning, while recording the precise right-Kan result for
  conditional expectation from \citet{vanbelle2023kan}.
\end{enumerate}

\section{Acknowledgments}

This research has been funded by Adobe Corporation.


\end{document}